\newif\ifextended 
\newif\ifauthor 
\newif\iforcid 
\newif\ifcolor 
\newif\ifthesis 
\newif\iftprint 

\extendedtrue\authortrue\orcidtrue\colortrue\thesisfalse\tprintfalse

\documentclass[runningheads]{llncs}

\newcommand{\gfivewidth}{169mm}
\newcommand{\gfiveheight}{239mm}
\newlength{\evenmargin}
\ifthesis
  \usepackage[paperwidth=\gfivewidth, paperheight=\gfiveheight, textwidth=\textwidth, textheight=\textheight \iftprint\else, inner=\evenmargin, outer=\evenmargin\fi]{geometry}
\fi
\usepackage[T1]{fontenc}
\usepackage{graphicx}
\usepackage{bbding}

\ifcolor
\usepackage[colorlinks=true,citecolor=blue,urlcolor=blue,linkcolor=blue]{hyperref}%

\else
\usepackage[hidelinks]{hyperref}%
\fi

\usepackage{mathtools}
\usepackage{amsfonts}
\usepackage{amssymb}

\usepackage{multicol}

\usepackage{subcaption}

\usepackage{listings}
\lstdefinestyle{alg}{%
  basicstyle=\sffamily\footnotesize,
  columns=fullflexible,
  morekeywords={let,match,if,else,in,then,with,function,for,return},
  numbersep=5pt,
  morecomment=[l]{//},
  commentstyle = \rmfamily\ifcolor\color{olive}\else\color{gray}\fi,
}
\lstdefinelanguage{calc}{%
  morecomment=[l]{//},
  columns=fullflexible,
  commentstyle = \ifcolor\color{olive}\else\color{gray}\fi,
}

\lstdefinelanguage{CorePPL}{%
  morekeywords={mexpr,let,assume,observe,true,false,include,type,con,in,lam,match,with,then,else,never,recursive,weight,resample,if},
  morecomment=[l]{--},
  commentstyle = \ifcolor\color{olive}\else\color{gray}\fi,
  numbers=left,
  xleftmargin=2em,
  numbersep=3pt,
  columns=fullflexible,
  basicstyle=\ttfamily\scriptsize,
}

\usepackage{tikz}
\usetikzlibrary{patterns}

\usepackage{pgfplots}
\pgfplotsset{compat=1.16}
\usepgfplotslibrary{statistics}

\usepackage{algorithm}

\newcommand{\ttt}[1]{\texttt{\upshape #1}}
\newcommand{\tsf}[1]{\textsf{\upshape #1}}
\newcommand{\tsc}{\textsc}
\newcommand{\tbf}{\textbf}
\newcommand{\mi}{\mathit}
\newcommand{\term}{\textbf{\tsf{t}}}
\newcommand{\termv}{\textbf{\tsf{v}}}
\newcommand{\unaligned}{\mi{unaligned}}
\newcommand{\s}{\enspace}
\newcommand\concat{\mathbin\Vert}
\newcommand{\false}{\text{\upshape false}}
\newcommand{\true}{\text{\upshape true}}
\newcommand{\termanf}{\term_\textrm{\normalfont ANF}}
\newcommand{\Termanf}{T_\textrm{\normalfont ANF}}
\newcommand\restt{{\widehat{A}_\term}}

\newcommand{\absval}{\textbf{\tsf{a}}}
\newcommand{\cstr}{\textbf{\tsf{c}}}

\newcommand{\pat}{\tsf{\textbf{p}}}

\newcommand\termgeo{\term_\mi{geo}}
\newcommand\termanfex{\term_\mi{example}}

\newcommand\sem[3]{{\displaystyle\hspace{1mm}\prescript{#2\vphantom{#3}}{\vphantom{#1}}{\Downarrow}^{#3\vphantom{#2}}_{#1}\hspace{1mm}}}

\newcommand\alignedcolor{gray}
\newcommand\unalignedcolor{white}

\usepackage[absolute]{textpos}

\makeatletter

\iforcid
\def\orcidID#1{\smash{\href{http://orcid.org/#1}{\protect\raisebox{-1.25pt}{\protect\includegraphics{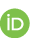}}}}}
\else
\def\orcidID#1{}
\fi
\makeatother

\begin{document}

%
\title{%
  Automatic Alignment in Higher-Order Probabilistic Programming Languages%
  \thanks{This project is financially supported by the Swedish Foundation for Strategic Research (FFL15-0032 and RIT15-0012), and also partially supported by the Wallenberg Al, Autonomous Systems and Software Program (WASP) funded by the Knut and Alice Wallenberg Foundation, and the Swedish Research Council (grants 2018-04620 and 2021-04830). The research has also been carried out as part of the Vinnova Competence Center for Trustworthy Edge Computing Systems and Applications at KTH Royal Institute of Technology.}
}

\titlerunning{Automatic Alignment in Higher-Order PPLs}

\author{%
  Daniel Lundén\inst{1}(\raisebox{-2.4pt}{\Envelope})\orcidID{0000-0003-3127-5640} \and%
  Gizem Çaylak\inst{1}\orcidID{0000-0001-9703-6912} \and %
  Fredrik Ronquist\inst{2,3}\orcidID{0000-0002-3929-251X} \and \\ %
  David Broman\inst{1}\orcidID{0000-0001-8457-4105}
}

\authorrunning{D. Lundén et al.}

\institute{%
  EECS and Digital Futures, KTH Royal Institute of Technology, Stockholm, Sweden, \email{\{dlunde,caylak,dbro\}@kth.se} \and
  Department of Bioinformatics and Genetics, Swedish Museum of Natural History, Stockholm, Sweden, \email{fredrik.ronquist@nrm.se} \and
  Department of Zoology, Stockholm University, Stockholm, Sweden
}
\maketitle              

\ifauthor
\begin{textblock*}{1.2\textwidth}(1in + \oddsidemargin - 0.1\textwidth,10pt)
  \noindent
  \scriptsize
  This is an author-prepared version\ifextended, extended with appendices and minor additions to the main text\fi.
  © The Author(s) 2023.
  This version of the contribution\ifextended, except the extensions,\fi\ has been accepted for publication at ESOP 2023, after peer review, but is not the Version of Record (the version published by Springer) and does not reflect post-acceptance improvements, or any corrections.
  The Version of Record is available online at: \url{https://doi.org/10.1007/978-3-031-30044-8_20}.
\end{textblock*}
\fi

\begin{abstract}
  Probabilistic Programming Languages (PPLs) allow users to encode statistical inference problems and automatically apply an \emph{inference algorithm} to solve them.
  Popular inference algorithms for PPLs, such as sequential Monte Carlo (SMC) and Markov chain Monte Carlo (MCMC), are built around \emph{checkpoints}---relevant events for the inference algorithm during the execution of a probabilistic program.
  Deciding the location of checkpoints is, in current PPLs, not done optimally.
  To solve this problem, we present a static analysis technique that automatically determines
  checkpoints in programs, relieving PPL users of this task.
  The analysis identifies a set of checkpoints that execute in the same order in every program run---they are \emph{aligned}.
  We formalize alignment, prove the correctness of the analysis, and implement the analysis as part of the higher-order functional PPL Miking CorePPL.
  By utilizing the alignment analysis, we design two novel inference algorithm variants: \emph{aligned SMC} and \emph{aligned lightweight MCMC}. We show, through real-world experiments, that they significantly improve inference execution time and accuracy compared to standard PPL versions of SMC and MCMC.


  \keywords{Probabilistic programming \and Operational semantics \and Static analysis.}
\end{abstract}


\section{Introduction}

Probabilistic programming languages (PPLs) are languages used to encode statistical inference problems, common in research fields such as phylogenetics~\cite{ronquist2021universal}, computer vision~\cite{gothoskar20213dp3}, topic modeling~\cite{blei2003latent}, data cleaning~\cite{lew2021pclean}, and cognitive science~\cite{goodman2016probabilistic}.
PPL implementations automatically solve encoded problems by applying an \emph{inference algorithm}.
In particular, automatic inference allows users to solve inference problems without having in-depth knowledge of inference algorithms and how to apply them.
Some examples of PPLs are WebPPL~\cite{goodman2014design}, Birch~\cite{murray2018automated}, Anglican~\cite{wood2014new}, Miking CorePPL~\cite{lunden2022compiling}, Turing~\cite{ge2018turing}, and Pyro~\cite{bingham2019pyro}.

Sequential Monte Carlo (SMC) and Markov chain Monte Carlo (MCMC) are general-purpose families of inference algorithms often used for PPL implementations.
These algorithms share the concept of \emph{checkpoints}: relevant execution events for the inference algorithm.
For SMC, the checkpoints are \emph{likelihood updates}~\cite{wood2014new,goodman2014design} and determine the \emph{resampling} of executions.
Alternatively, users must sometimes manually annotate or write the probabilistic program in a certain way to make resampling explicit~\cite{lunden2022compiling,murray2018automated}.
For MCMC, checkpoints are instead \emph{random draws}, which allow the inference algorithm to manipulate these draws to construct a Markov chain over program executions~\cite{wingate2011lightweight,ritchie2016c3}.
When designing SMC and MCMC algorithms for \emph{universal PPLs}%
\footnote{%
  A term coined by Goodman et al.~\cite{goodman2008church}.
  Essentially, it means that the types and numbers of random variables cannot be determined statically.
}%
, both the \emph{placement} and \emph{handling} of checkpoints are critical to making the inference both efficient and accurate.


For SMC, a standard inference approach is to resample at \emph{all} likelihood updates~\cite{goodman2014design,wood2014new}.
This approach produces correct results asymptotically~\cite{lunden2021correctness} but is highly problematic for certain models~\cite{ronquist2021universal}.
Such models require non-trivial and SMC-specific manual program rewrites to force good resampling locations and make SMC tractable.
Overall, choosing the likelihood updates at which to resample significantly affects SMC execution time and accuracy.

For MCMC, a standard approach for inference in universal PPLs is \emph{lightweight MCMC}~\cite{wingate2011lightweight}, which constructs a Markov chain over random draws in programs.
The key idea is to use an \emph{addressing transformation} and a \emph{runtime database} of random draws.
Specifically, the database enables matching and reusing random draws between executions according to their \emph{stack traces}, even if the random draws may or may not occur due to randomness during execution.
However, the dynamic approach of looking up random draws in the database through their stack traces is expensive and introduces significant runtime overhead.

To overcome the SMC and MCMC problems in universal PPLs, we present a static analysis technique for higher-order functional PPLs that \emph{automatically} determines checkpoints in a probabilistic program that always occur in the same order in every program execution---they are \emph{aligned}.
We formally define alignment, formalize the alignment analysis, and prove the soundness of the analysis with respect to the alignment definition.
The novelty and challenge in developing the static analysis technique is to capture alignment properties through the identification of expressions in programs that may evaluate to \emph{stochastic values} and expressions that may evaluate due to \emph{stochastic branching}.
Stochastic branching results from \ttt{if} expressions with stochastic values as conditions or function applications where the function itself is stochastic.
Stochastic values and branches pose a significant challenge when proving the soundness of the analysis.

We design two new inference algorithms that improve accuracy and execution time compared to current approaches.
Unlike the standard SMC algorithm for PPLs~\cite{wood2014new,goodman2014design}, \emph{aligned SMC} only resamples at aligned likelihood updates.
Resampling only at aligned likelihood updates guarantees that each SMC execution resamples the same number of times, which makes expensive global termination checks redundant~\cite{lunden2022compiling}.
We evaluate aligned SMC on two diversification models from Ronquist et al.~\cite{ronquist2021universal} and a state-space model for aircraft localization, demonstrating significantly improved inference accuracy and execution time compared to traditional SMC.
Both models---constant rate birth-death (CRBD) and cladogenetic diversification rate shift (ClaDS)---are used in real-world settings and are of considerable interest to evolutionary biologists~\cite{nee2006birth,maliet2019model}.
The documentations of both Anglican~\cite{wood2014new} and Turing~\cite{ge2018turing} acknowledge the importance of alignment for SMC and state that all likelihood updates must be aligned.
However, Turing and Anglican neither formalize nor enforce this property---it is up to the users to \emph{manually} guarantee it, often requiring non-standard program rewrites~\cite{ronquist2021universal}.

We also design \emph{aligned lightweight MCMC}, a new version of lightweight MCMC~\cite{wingate2011lightweight}.
Aligned lightweight MCMC constructs a Markov chain over the program using the aligned random draws as \emph{synchronization points} to match and reuse aligned random draws and a subset of unaligned draws between executions.
Aligned lightweight MCMC does not require a runtime database of random draws and therefore reduces runtime overhead.
We evaluate aligned lightweight MCMC for latent Dirichlet allocation (LDA)~\cite{blei2003latent} and CRBD~\cite{ronquist2021universal}, demonstrating significantly reduced execution times and no decrease in inference accuracy.
Furthermore, automatic alignment is orthogonal to and easily combines with the lightweight MCMC optimizations introduced by Ritchie et al.~\cite{ritchie2016c3}.

We implement the analysis, aligned SMC, and aligned lightweight MCMC in Miking CorePPL~\cite{lunden2022compiling,broman2019vision}.
In addition to analyzing stochastic \ttt{if}-branching, the implementation analyzes stochastic branching at a standard pattern-matching construct.
Compared to \ttt{if} expressions, the pattern-matching construct requires a more sophisticated analysis of the pattern and the value matched against it to determine if the pattern-matching causes a stochastic branch.

In summary, we make the following contributions.
\begin{itemize}
  \item We invent and formalize alignment for PPLs.
    Aligned parts of a program occur in the same order in every execution (Section~\ref{sec:anf}).
  \item We formalize and prove the soundness of a novel static analysis technique that determines stochastic value flow and stochastic branching, and in turn alignment, in higher-order probabilistic programs (Section~\ref{sec:analysis}).
  \item We design aligned SMC inference that only resamples at aligned likelihood updates, improving execution time and inference accuracy (Section~\ref{sec:smcalign}).
  \item We design aligned lightweight MCMC inference that only reuses aligned random draws, improving execution time (Section~\ref{sec:mcmcalign}).
  \item We implement the analysis and inference algorithms in Miking CorePPL.
  The implementation extends the alignment analysis to identify stochastic branching resulting from pattern matching (Section~\ref{sec:implementation}).
\end{itemize}
Section~\ref{sec:eval} describes the evaluation and discusses its results.
The paper also has an accompanying artifact that supports the evaluation~\cite{lunden2023automaticartifact}.
Section~\ref{sec:relatedwork} discusses related work and Section~\ref{sec:conclusion} concludes.
Next, Section~\ref{sec:motivating} considers a simple motivating example to illustrate the key ideas.
Section~\ref{sec:syntaxsemantics} introduces syntax and semantics for the calculus used to formalize the alignment analysis.

\section{A Motivating Example}\label{sec:motivating}
This section presents a motivating example that illustrates the key alignment ideas in relation to aligned SMC (Section~\ref{sec:alignedsmcmot}) and aligned lightweight MCMC (Section~\ref{sec:alignedlwmot}).
We assume basic knowledge of probability theory.
Knowledge of PPLs is helpful, but not a strict requirement.
The book by van de Meent et al.~\cite{vandemeent2018introduction} provides a good introduction to PPLs.

Probabilistic programs encode Bayesian statistical inference problems with two fundamental constructs: \texttt{assume} and \texttt{weight}.
The \texttt{assume} construct defines random variables, which make execution nondeterministic.
Intuitively, a probabilistic program then encodes a probability distribution over program executions (the prior distribution), and it is possible to sample from this distribution by executing the program with random sampling at \texttt{assume}s.
The \texttt{weight} construct updates the \emph{likelihood} of individual executions.
Updating likelihoods for executions modifies the probability distribution induced by \texttt{assume}s, and the inference problem encoded by the program is to determine or approximate this modified distribution (the posterior distribution).
The main purpose of \texttt{weight} in real-world models is to condition executions on observed data.%
\footnote{A number of more specialized constructs for likelihood updating are also available in various PPLs, for example \emph{observe}~\cite{wood2014new,goodman2014design} and \emph{condition}~\cite{goodman2014design}.}


\begin{figure}[tb]
  \centering
  \begin{subfigure}{0.4\columnwidth}
    \lstset{%
      basicstyle=\ttfamily\scriptsize,
      numbers=left,
      showlines=true,
      numbersep=3pt,
      framexleftmargin=-2pt,
      xleftmargin=2em,
      language=calc
    }
    \centering
    \hspace*{2mm}%
    \begin{tabular}{c}
      \begin{lstlisting}
let $\mi{rate}$ = assume $\textrm{Gamma}(2,2)$ in$\label{line:mot:gamma}$
let rec $\mi{survives}$ = $\lambda n$.$\label{line:mot:survives}$
  if $n = 0$ then $()$ else$\label{line:mot:stochbranch}$
    if assume $\textrm{Bernoulli}(0.9)$ then$\label{line:mot:bern}$
      weight $0.5$;$\label{line:mot:weight1}$
      $\mi{survives}$ ($n - 1$)$\label{line:mot:survrec}$
    else
      weight $0$$\label{line:mot:weight2}$
in
let rec $\mi{iter}$ = $\lambda i$.$\label{line:mot:iter}$
  if $i = 0$ then $()$ else$\label{line:mot:branch}$
    weight $\mi{rate}$;$\label{line:mot:weight3}$
    let $n$ = assume $\textrm{Poisson}(\mi{rate})$ in$\label{line:mot:poisson}$
    $\mi{survives}$ $n$;$\label{line:mot:scall}$
    $\mi{iter}$ ($i - 1$)$\label{line:mot:iterrec}$
in
$\mi{iter}$ $3$;$\label{line:mot:iterinit}$
$\mi{rate}$$\label{line:mot:return}$
      \end{lstlisting}
    \end{tabular}
    \caption{Probabilistic program.}
    \label{fig:mot:prog}
  \end{subfigure}
  \begin{minipage}{0.58\columnwidth}
    \centering
    \begin{subfigure}{0.49\textwidth}
      \centering
      \begin{tikzpicture}[
        declare function={gamma(\z)=
        (2.506628274631*sqrt(1/\z) + 0.20888568*(1/\z)^(1.5) + 0.00870357*(1/\z)^(2.5) - (174.2106599*(1/\z)^(3.5))/25920 - (715.6423511*(1/\z)^(4.5))/1244160)*exp((-ln(1/\z)-1)*\z);},
        declare function={gammapdf(\x,\k,\theta) = \x^(\k-1)*exp(-\x/\theta) / (\theta^\k*gamma(\k));}
        ]
        \scriptsize
        \begin{axis}[
          axis x line*=bottom,
          y axis line style={draw=none},
          tick style={draw=none},
          ytick=\empty,
          xmin=0, xmax=15,
          width=1.2\columnwidth,
          height=24mm,
          ymin=0,ymax=0.42,
          ]
          \addplot [smooth,domain=0:15] {gammapdf(x,2,2)};
        \end{axis}
      \end{tikzpicture}
      \caption{Gamma$(2,2)$.}
      \label{fig:mot:prior}
    \end{subfigure}
    \begin{subfigure}{0.49\textwidth}
      \centering
      \begin{tikzpicture}
        \scriptsize
        \begin{axis}[
          ymax=0.42,
          axis line style={draw=none},
          xtick style={draw=none},
          ytick=\empty,
          xmin=0, xmax=15,
          enlargelimits=false,
          width=1.2\columnwidth,
          height=24mm,
          ]
          \addplot+[
            black, fill=gray, mark=none, ybar interval
            ] table [header=false] {examples/motivating/example.dat};
        \end{axis}
      \end{tikzpicture}
      \caption{Histogram.}
      \label{fig:mot:post}
    \end{subfigure}\\[2mm]
    \begin{subfigure}{\textwidth}
      \hypersetup{hidelinks}
      \scriptsize
      \centering
      \begin{tikzpicture}[
        scale=0.40,
        minimum width=3.7mm,
        minimum height=2.5mm,
        inner sep=0
        ]

        \node at (-1,1.7) {$w_1$};
        \node[draw] at (0,1.7) {$\ref{line:mot:weight3}$};
        \node[draw] at (1,1.7) {$\ref{line:mot:weight3}$};
        \node[draw] at (2,1.7) {$\ref{line:mot:weight1}$};
        \node[draw] at (3,1.7) {$\ref{line:mot:weight1}$};
        \node[draw] at (4,1.7) {$\ref{line:mot:weight1}$};
        \node[draw] at (5,1.7) {$\ref{line:mot:weight3}$};
        \node[draw] at (6,1.7) {$\ref{line:mot:weight1}$};
        \node[draw] at (7,1.7) {$\ref{line:mot:weight1}$};

        \node at (-1,1.0) {$w_2$};
        \node[draw] at (0,1.0) {$\ref{line:mot:weight3}$};
        \node[draw] at (1,1.0) {$\ref{line:mot:weight1}$};
        \node[draw] at (2,1.0) {$\ref{line:mot:weight1}$};
        \node[draw] at (3,1.0) {$\ref{line:mot:weight3}$};
        \node[draw] at (4,1.0) {$\ref{line:mot:weight2}$};
        \node[draw] at (5,1.0) {$\ref{line:mot:weight3}$};
        \node[draw] at (6,1.0) {$\ref{line:mot:weight1}$};

        \node at                 (-1,0.2) {$w_1$};
        \node[draw,fill=gray] at (0,0.2) {$\ref{line:mot:weight3}$};
        \node[draw,fill=gray] at (3,0.2) {$\ref{line:mot:weight3}$};
        \node[draw] at           (4,0.2) {$\ref{line:mot:weight1}$};
        \node[draw] at           (5,0.2) {$\ref{line:mot:weight1}$};
        \node[draw] at           (6,0.2) {$\ref{line:mot:weight1}$};
        \node[draw,fill=gray] at (7,0.2) {$\ref{line:mot:weight3}$};
        \node[draw] at           (8,0.2) {$\ref{line:mot:weight1}$};
        \node[draw] at           (9,0.2) {$\ref{line:mot:weight1}$};

        \node at                 (-1,-0.5) {$w_2$};
        \node[draw,fill=gray] at (0,-0.5) {$\ref{line:mot:weight3}$};
        \node[draw] at           (1,-0.5) {$\ref{line:mot:weight1}$};
        \node[draw] at           (2,-0.5) {$\ref{line:mot:weight1}$};
        \node[draw,fill=gray] at (3,-0.5) {$\ref{line:mot:weight3}$};
        \node[draw] at           (4,-0.5) {$\ref{line:mot:weight2}$};
        \node[draw,fill=gray] at (7,-0.5) {$\ref{line:mot:weight3}$};
        \node[draw] at           (8,-0.5) {$\ref{line:mot:weight1}$};

      \end{tikzpicture}
      \caption{Aligning \ttt{weight}.}
      \label{fig:mot:weight}
    \end{subfigure}\\[2mm]
    \begin{subfigure}{\textwidth}
      \hypersetup{hidelinks}
      \scriptsize
      \centering
      \begin{tikzpicture}[
        scale=0.40,
        minimum width=3.7mm,
        minimum height=2.5mm,
        inner sep=0
        ]

        \node at (-1,1.7) {$s_1$};
        \node[draw] at (0,1.7)  {$\ref{line:mot:gamma}$};
        \node[draw] at (1,1.7)  {$\ref{line:mot:poisson}$};
        \node[draw] at (2,1.7)  {$\ref{line:mot:bern}$};
        \node[draw] at (3,1.7)  {$\ref{line:mot:poisson}$};
        \node[draw] at (4,1.7)  {$\ref{line:mot:bern}$};
        \node[draw] at (5,1.7)  {$\ref{line:mot:bern}$};
        \node[draw] at (6,1.7)  {$\ref{line:mot:bern}$};
        \node[draw] at (7,1.7)  {$\ref{line:mot:poisson}$};

        \node at (-1,1.0) {$s_2$};
        \node[draw] at (0,1.0) {$\ref{line:mot:gamma}$};
        \node[draw] at (1,1.0) {$\ref{line:mot:poisson}$};
        \node[draw] at (2,1.0) {$\ref{line:mot:bern}$};
        \node[draw] at (3,1.0) {$\ref{line:mot:bern}$};
        \node[draw] at (4,1.0) {$\ref{line:mot:poisson}$};
        \node[draw] at (5,1.0) {$\ref{line:mot:bern}$};
        \node[draw] at (6,1.0) {$\ref{line:mot:bern}$};
        \node[draw] at (7,1.0) {$\ref{line:mot:poisson}$};
        \node[draw] at (8,1.0) {$\ref{line:mot:bern}$};

        \node at (-1,0.2) {$s_1$};
        \node[draw,fill=gray] at (0,0.2) {$\ref{line:mot:gamma}$};
        \node[draw,fill=gray] at (1,0.2) {$\ref{line:mot:poisson}$};
        \node[draw] at           (2,0.2) {$\ref{line:mot:bern}$};
        \node[draw,fill=gray] at (4,0.2) {$\ref{line:mot:poisson}$};
        \node[draw] at           (5,0.2) {$\ref{line:mot:bern}$};
        \node[draw] at           (6,0.2) {$\ref{line:mot:bern}$};
        \node[draw] at           (7,0.2) {$\ref{line:mot:bern}$};
        \node[draw,fill=gray] at (8,0.2) {$\ref{line:mot:poisson}$};

        \node at (-1,-0.5) {$s_2$};
        \node[draw,fill=gray] at (0,-0.5) {$\ref{line:mot:gamma}$};
        \node[draw,fill=gray] at (1,-0.5) {$\ref{line:mot:poisson}$};
        \node[draw] at           (2,-0.5) {$\ref{line:mot:bern}$};
        \node[draw] at           (3,-0.5) {$\ref{line:mot:bern}$};
        \node[draw,fill=gray] at (4,-0.5) {$\ref{line:mot:poisson}$};
        \node[draw] at           (5,-0.5) {$\ref{line:mot:bern}$};
        \node[draw] at           (6,-0.5) {$\ref{line:mot:bern}$};
        \node[draw,fill=gray] at (8,-0.5) {$\ref{line:mot:poisson}$};
        \node[draw] at           (9,-0.5) {$\ref{line:mot:bern}$};

      \end{tikzpicture}
      \caption{Aligning \ttt{assume}.}
      \label{fig:mot:assume}
    \end{subfigure}
  \end{minipage}

  \caption{%
    A simple example illustrating alignment.
    Fig. (a) gives a probabilistic program using functional-style PPL pseudocode.
    Fig. (b) illustrates the Gamma$(2,2)$ probability density function.
    Fig. (c) illustrates a histogram over weighted $\mi{rate}$ samples produced by running the program in (a) a large number of times.
    Fig. (d) shows two line number sequences $w_1$ and $w_2$ of \ttt{weight}s encountered in two program runs (top) and how to align them (bottom).
    Fig. (e) shows two line number sequences $s_1$ and $s_2$ of \ttt{assume}s encountered in two program runs (top) and how to align them (bottom).
  }
  \label{fig:mot}
\end{figure}%
Consider the probabilistic program in Fig.~\ref{fig:mot:prog}.
The program is contrived and purposefully constructed to compactly illustrate alignment, but the real-world diversification models in Ronquist et al.~\cite{ronquist2021universal} that we also consider in Section~\ref{sec:eval} inspired the program's general structure.
The program defines (line~\ref{line:mot:gamma}) and returns (line~\ref{line:mot:return}) a Gamma-distributed random variable $\mi{rate}$.
Figure~\ref{fig:mot:prior} illustrates the Gamma distribution.
To modify the likelihood for values of $\mi{rate}$, the program executes the $\mi{iter}$ function (line~\ref{line:mot:iter}) three times, and the $\mi{survives}$ function (line~\ref{line:mot:survives}) a random number of times $n$ (line~\ref{line:mot:poisson}) within each $\mi{iter}$ call.

Conceptually, to infer the posterior distribution of the program, we execute the program infinitely many times.
In each execution, we draw samples for the random variables defined at \ttt{assume}, and accumulate the likelihood at \texttt{weight}.
The return value of the execution, weighted by the accumulated likelihood, represents one sample from the posterior distribution.
Fig.~\ref{fig:mot:post} shows a histogram of such weighted samples of $\mi{rate}$ resulting from a large number of executions of Fig.~\ref{fig:mot:prog}.
The fundamental inference algorithm that produces such weighted samples is called \emph{likelihood weighting} (a type of \emph{importance sampling}~\cite{naesseth2019elements}).
We see that, compared to the prior distribution for $\mi{rate}$ in Fig.~\ref{fig:mot:prior}, the posterior is more sharply peaked due to the likelihood modifications.

\subsection{Aligned SMC}\label{sec:alignedsmcmot}
Likelihood weighting can only handle the simplest of programs.
In Fig.~\ref{fig:mot:prog}, a problem with likelihood weighting is that we assign the weight 0 to many executions at line~\ref{line:mot:weight2}.
These executions contribute nothing to the final distribution.
SMC solves this by executing many program instances \emph{concurrently} and occasionally \emph{resampling} them (with replacement) based on their current likelihoods.
Resampling discards executions with lower weights (in the worst case, 0) and replaces them with executions with higher weights.
The most common approach in popular PPLs is to resample just after likelihood updates (i.e., calls to \ttt{weight}).

Resampling at all calls to \ttt{weight} in Fig.~\ref{fig:mot:prog} is suboptimal.
The best option is instead to \emph{only} resample at line~\ref{line:mot:weight3}.
This is because executions encounter lines~\ref{line:mot:weight1} and \ref{line:mot:weight2} a \emph{random} number of times due to the stochastic branch at line~\ref{line:mot:stochbranch}, while they encounter line \ref{line:mot:weight3} a fixed number of times.
As a result of resampling at lines~\ref{line:mot:weight1} and \ref{line:mot:weight2}, executions become \emph{unaligned}; in each resampling, executions can have reached either line~\ref{line:mot:weight1}, line~\ref{line:mot:weight2}, or line~\ref{line:mot:weight3}.
On the other hand, if we resample only at line~\ref{line:mot:weight3}, all executions will always have reached line~\ref{line:mot:weight3} for the same iteration of $\mi{iter}$ in every resampling.
Intuitively, this is a sensible approach since, when resampling, executions have progressed the same distance through the program.
We say that the \ttt{weight} at line~\ref{line:mot:weight3} is \emph{aligned}, and resampling only at aligned \texttt{weight}s results in our new inference approach called \emph{aligned SMC}.
Fig.~\ref{fig:mot:weight} visualizes the \texttt{weight} alignment for two sample executions of Fig.~\ref{fig:mot:prog}.

\subsection{Aligned Lightweight MCMC}\label{sec:alignedlwmot}
Another improvement over likelihood weighting is to construct a Markov chain over program executions.
It is beneficial to propose new executions in the Markov chain by making small, rather than large, modifications to the previous execution.
The lightweight MCMC \cite{wingate2011lightweight} algorithm does this by redrawing a single random draw in the previous execution, and then reusing as many other random draws as possible.
Random draws in the current and previous executions match through \emph{stack traces}---the sequences of applications leading up to a random draw.
Consider the random draw at line~\ref{line:mot:poisson} in Fig.~\ref{fig:mot:prog}.
It is called exactly three times in every execution.
If we identify applications and \texttt{assume}s by line numbers, we get the stack traces%
{\hypersetup{hidelinks}
$[\ref{line:mot:iterinit},\ref{line:mot:poisson}]$,
$[\ref{line:mot:iterinit},\ref{line:mot:iterrec},\ref{line:mot:poisson}]$, and
$[\ref{line:mot:iterinit},\ref{line:mot:iterrec},\ref{line:mot:iterrec},\ref{line:mot:poisson}]$}
for these three \texttt{assume}s in every execution.
Consequently, lightweight MCMC can reuse these draws by storing them in a database indexed by stack traces.

The stack trace indexing in lightweight MCMC is overly complicated when reusing aligned random draws.
Note that the \ttt{assume}s at lines~\ref{line:mot:gamma} and \ref{line:mot:poisson} in Fig~\ref{fig:mot:prog} are aligned, while the \ttt{assume} at line~\ref{line:mot:bern} is unaligned.
Fig.~\ref{fig:mot:assume} visualizes the \texttt{assume} alignment for two sample executions of Fig.~\ref{fig:mot:prog}.
Aligned random draws occur in the same same order in every execution, and are therefore trivial to match and reuse between executions through indexing by counting.
The appeal with stack trace indexing is to additionally allow reusing a subset of \emph{unaligned} draws.

A key insight in this paper is that aligned random draws can also act as \emph{synchronization points} in the program to allow reusing unaligned draws \emph{without} a stack trace database.
After an aligned draw, we reuse unaligned draws occurring up until the next aligned draw, as long as they syntactically originate at the same \ttt{assume} as the corresponding unaligned draws in the previous execution.
As soon as an unaligned draw does not originate from the same \ttt{assume} as in the previous execution, we redraw all remaining unaligned draws up until the next aligned draw.
Instead of a trace-indexed database, this approach requires storing a list of unaligned draws (tagged with identifiers of the \ttt{assume}s at which they originated) for each execution segment in between aligned random draws.
For example, for the execution $s_1$ in Fig.~\ref{fig:mot:assume}, we store lists of unaligned Bernoulli random draws from line~\ref{line:mot:bern} for each execution segment in between the three aligned random draws at line~\ref{line:mot:poisson}.
If a Poisson random draw $n$ at line~\ref{line:mot:poisson} does not change or decreases, we can reuse the stored unaligned Bernoulli draws up until the next Poisson random draw as $\mi{survives}$ executes $n$ or fewer times.
If the drawn $n$ instead increases to $n'$, we can again reuse all stored Bernoulli draws, but must supplement them with new Bernoulli draws to reach $n'$ draws in total.

As we show in Section~\ref{sec:eval}, using aligned draws as synchronization points works very well in practice and avoids the runtime overhead of the lightweight MCMC database.
However, manually identifying aligned parts of programs and rewriting them so that inference can make use of alignment is, if even possible, tedious, error-prone, and impractical for large programs.
This paper presents an automated approach to identifying aligned parts of programs.
Combining static alignment analysis and using aligned random draws as synchronization points form the key ideas of the new algorithm that we call \emph{aligned lightweight MCMC}.

\section{Syntax and Semantics}\label{sec:syntaxsemantics}

In preparation for the alignment analysis in Section~\ref{sec:align}, we require an idealized base calculus capturing the key features of expressive PPLs.
This section introduces such a calculus with a formal syntax (Section~\ref{sec:syntax}) and semantics (Section~\ref{sec:semantics}).
We assume a basic understanding of the lambda calculus (see, e.g., Pierce~\cite{pierce2002types} for a complete introduction).
Section~\ref{sec:implementation} further describes extending the idealized calculus and the analysis in Section~\ref{sec:align} to a full-featured PPL.

\subsection{Syntax}\label{sec:syntax}
We use the untyped lambda calculus as the base for our calculus.
We also add \ttt{let} expressions for convenience, and \ttt{if} expressions to allow intrinsic booleans to affect control flow.
The calculus is a subset of the language used in Fig.~\ref{fig:mot:prog}.
We inductively define terms $\term$ and values $\termv$ as follows.
\begin{definition}[Terms and values]
  \begin{equation}\label{eq:ast}
    \begin{gathered}
      \begin{aligned}
        \term \Coloneqq& \s
        x
        \s | \s
        c
        \s | \s
        \lambda x. \s \term
        \s | \s
        \term \s \term
        \s | \s
        \ttt{let } x = \term \ttt{ in } \term
        & \termv{} \Coloneqq& \s
        c
        \s | \s
        \langle\lambda x. \s \term,\rho\rangle
        \\ |& \s
        \ttt{if } \term \ttt{ then } \term \ttt{ else } \term
        \s | \s
        \ttt{assume } \term
        \s | \s
        \ttt{weight } \term &&\\
      \end{aligned} \\
      x,y \in X \quad
      \rho \in P \quad
      c \in C \quad
      \{ \false, \true, () \} \cup \mathbb{R} \cup D \subseteq C.
    \end{gathered}
  \end{equation}
  $X$ is a countable set of variable names, $C$ a set of intrinsic values and operations, and $D \subset C$ a set of probability distributions.
  The set $P$ contains all evaluation environments $\rho$, that is, partial functions mapping names in $X$ to values $\termv$.
  We use $T$ and $V$ to denote the set of all terms and values, respectively.
\end{definition}
\noindent
Values $\termv$ are intrinsics or closures, where closures are abstractions with an environment binding free variables in the abstraction body.
We require that $C$ include booleans, the unit value $()$, and real numbers.
The reason is that \ttt{weight} takes real numbers as argument and returns $()$ and that \ttt{if} expression conditions are booleans.
Furthermore, probability distributions are often over booleans and real numbers.
For example, we can include the normal distribution constructor $\mathcal{N} \in C$ that takes real numbers as arguments and produces normal distributions over real numbers.
For example, $\mathcal{N} \enspace 0 \enspace 1 \in D$, the standard normal distribution.
We often write functions in $C$ in infix position or with standard function application syntax for readability.
For example, $1 + 2$ with $+ \in C$ means $+$ 1 2, and $\mathcal N(0,1)$ means $\mathcal{N} \enspace 0 \enspace 1$.
Additionally, we use the shorthand $\term_1; \term_2$ for \ttt{let \textrm{\_} = $\term_1$ in $\term_2$}, where \_ is the do-not-care symbol.
That is, $\term_1; \term_2$ evaluates $\term_1$ for side effects only before evaluating $\term_2$.
Finally, the untyped lambda calculus supports recursion through fixed-point combinators.
We encapsulate this in the shorthand \ttt{let~rec~$f$~=~$\lambda x.\term_1$~in~$\term_2$} to conveniently define recursive functions.

The \ttt{assume} and \ttt{weight} constructs are PPL-specific.
We define random variables from intrinsic probability distributions with \ttt{assume}
(also known as \emph{sample} in PPLs with sampling-based inference).
For example, the term \ttt{let $x$ = assume~$\mathcal{N}(0,1)$~in~$\term$} defines $x$ as a random variable with a standard normal distribution in \term.
Boolean random variables combined with \ttt{if} expressions result in \emph{stochastic branching}---causing the alignment problem.
Lastly, \ttt{weight} (also known as \emph{factor} or \emph{score}) is a standard construct for likelihood updating (see, e.g., Borgström et al.~\cite{borgstrom2016lambda}).
Next, we illustrate and formalize a semantics for~\eqref{eq:ast}.

\subsection{Semantics}\label{sec:semantics}
\begin{figure}[bt]
  \centering
  \begin{subfigure}[b]{0.5\columnwidth}
    \lstset{%
      basicstyle=\ttfamily\scriptsize,
      numbers=left,
      showlines=true,
      numbersep=3pt,
      framexleftmargin=-2pt,
      xleftmargin=2em,
      language=calc
    }
    \centering
    \begin{tikzpicture}[node distance=4mm]
      \footnotesize
      \node[fill=black!30,minimum width=15mm,minimum height=3.1mm,anchor=north west]
        at (-17.5mm,1.3mm) {};
      \node (program) {%
        \begin{lstlisting}
let rec $\mi{geometric}$ = $\lambda \textrm{\_}$.
  let $x$ = assume $\!\textrm{Bernoulli}(0.5)\!$ in $\label{line:flip}$
  if $x$ then
    weight $1.5$; $\label{line:weight}$
    $1 + \mi{geometric}$ $()$
  else $1$
in $\mi{geometric}$ $()$
          \end{lstlisting}
        };
    \end{tikzpicture}
    \caption{Probabilistic program $\termgeo$.}
    \label{fig:geo:a}
  \end{subfigure}
  \begin{subfigure}[b]{0.44\columnwidth}
    \centering
    \begin{tikzpicture}[trim axis left, trim axis right]
      \pgfplotsset{title style={at={(0.58,0.20)}}}
      \begin{axis}[%
        axis x line*=bottom,
        y axis line style={draw=none},
        title=Standard geometric,
        width=53mm,
        height=23mm,
        tick style={draw=none},
        ymin=0,
        ymax=0.51,
        xmax=10,
        ytick=\empty,
        xtick=\empty,
        bar width=9pt,
        ]
        \def\p{0.5}
        \def\w{1.0}
        \def\Z{(1/\w*((1/(1-\p*\w))-1))}
        \def\sumNine{(1/\w*(((1-(\p*\w)^10)/(1-\p*\w))-1))}
        \addplot [
          fill=white,
          ybar,
          ] coordinates {%
            (1,  \p            / \Z)
            (2,  \p^2   * \w^1 / \Z)
            (3,  \p^3   * \w^2 / \Z)
            (4,  \p^4   * \w^3 / \Z)
            (5,  \p^5   * \w^4 / \Z)
            (6,  \p^6   * \w^5 / \Z)
            (7,  \p^7   * \w^6 / \Z)
            (8,  \p^8   * \w^7 / \Z)
            (9,  \p^9   * \w^8 / \Z)
          };
        \Z
      \end{axis}
      \begin{axis}[%
        axis x line*=bottom,
        y axis line style={draw=none},
        title=Weighted geometric,
        yshift=-9mm,
        width=53mm,
        height=23mm,
        tick style={draw=none},
        ymin=0,
        ymax=0.51,
        xmax=10,
        ytick=\empty,
        xtick={1,2,3,4,5,6,7,8,9},
        bar width=9pt,
        ]
        \def\p{0.5}
        \def\w{1.5}
        \def\Z{(1/\w*((1/(1-\p*\w))-1))}
        \def\sumNine{(1/\w*(((1-(\p*\w)^10)/(1-\p*\w))-1))}
        \addplot [
          fill=gray,
          ybar,
          ] coordinates {%
            (1,  \p            / \Z)
            (2,  \p^2   * \w^1 / \Z)
            (3,  \p^3   * \w^2 / \Z)
            (4,  \p^4   * \w^3 / \Z)
            (5,  \p^5   * \w^4 / \Z)
            (6,  \p^6   * \w^5 / \Z)
            (7,  \p^7   * \w^6 / \Z)
            (8,  \p^8   * \w^7 / \Z)
            (9,  \p^9   * \w^8 / \Z)
          };
        \Z
      \end{axis}
    \end{tikzpicture}
    \caption{Probability distributions.}
    \label{fig:geo:b}
  \end{subfigure}
  \caption{%
    A probabilistic program $\termgeo$~\cite{lunden2022compiling}, illustrating~\eqref{eq:ast}.
    Fig. (a) gives the program, and (b) the corresponding probability distributions.
    In (b), the $y$-axis gives the probability, and the $x$-axis gives the outcome (the number of coin flips).
    The upper part of (b) excludes the shaded \ttt{weight} at line~\ref{line:weight} in (a).
  }
  \label{fig:geo}
\end{figure}
Consider the small probabilistic program $\termgeo \in T$ in Fig.~\ref{fig:geo:a}.
The program encodes the standard geometric distribution via a function $\mi{geometric}$, which recursively flips a fair coin (a Bernoulli$(0.5)$ distribution) at line~\ref{line:flip} until the outcome is false (i.e., tails).
At that point, the program returns the total number of coin flips, including the last tails flip.
The upper part of Fig.~\ref{fig:geo:b} illustrates the result distribution for an infinite number of program runs with line~\ref{line:weight} ignored.

To illustrate the effect of \ttt{weight}, consider $\termgeo$ with line~\ref{line:weight} included.
This \ttt{weight} modifies the likelihood with a factor $1.5$ each time the flip outcome is true (or, heads).
Intuitively, this emphasizes larger return values, illustrated in the lower part of Fig.~\ref{fig:geo:b}.
Specifically, the (unnormalized) probability of seeing $n$ coin flips is $0.5^n\cdot1.5^{n-1}$, compared to $0.5^n$ for the unweighted version.
The factor $1.5^{n-1}$ is the result of the calls to \ttt{weight}.

\begin{figure}[tb]
  \renewcommand{\columnwidth}{\textwidth} 
  \[\footnotesize
    \begin{gathered}
      \frac{}
      { \rho \vdash x \sem{[]}{[]}{1} \rho(x) }
      (\textsc{Var}) \qquad
      \frac{}
      { \rho \vdash c \sem{[]}{[]}{1} c }
      (\textsc{Const}) \qquad
      \frac{}
      { \rho \vdash \lambda x. \term \sem{[]}{[]}{1} \langle\lambda x. \term,\rho\rangle }
      (\textsc{Lam}) \\
      \frac{ \rho \vdash \term_1 \sem{l_1}{s_1}{w_1} \langle\lambda x. \term,\rho'\rangle \quad \rho \vdash \term_2 \sem{l_2}{s_2}{w_2} \termv_2 \quad \rho' ,x \mapsto \termv_2 \vdash \term \sem{l_3}{s_3}{w_3} \termv}
      { \rho \vdash \term_1 \s \term_2 \sem{l_1 \concat l_2 \concat l_3}{s_1 \concat s_2 \concat s_3}{w_1 \cdot w_2 \cdot w_3} \termv{} }
      (\textsc{App}) \\
      \frac{ \rho \vdash \term_1 \sem{l_1}{s_1}{w_1} c_1 \quad \rho \vdash \term_2 \sem{l_2}{s_2}{w_2} c_2}
      { \rho \vdash \term_1 \s \term_2 \sem{l_1 \concat l_2}{s_1 \concat s_2}{w_1 \cdot w_2} \delta(c_1,c_2) }
      (\textsc{Const-App}) \quad
      \frac{\rho \vdash \term \sem{l}{s}{w} d \quad w' = f_d(c) }
      {\rho \vdash \ttt{assume } \term \sem{l}{s \concat [c]}{w \cdot w'} c}
      (\textsc{Assume}) \\
      \frac{ \rho \vdash \term_1 \sem{l_1}{s_1}{w_1} \termv_1 \quad \rho , x \mapsto \termv_1 \vdash \term_2 \sem{l_2}{s_2}{w_2} \termv}
      { \rho \vdash \ttt{let } x = \term_1 \ttt{ in } \term_2 \sem{l_1 \concat [x] \concat l_2}{s_1 \concat s_2}{w_1 \cdot w_2} \termv{} }
      (\textsc{Let}) \quad
      \frac{\rho \vdash \term \sem{l}{s}{w} w'}
      {\rho \vdash \ttt{weight } \term \sem{l}{s}{w \cdot w'} ()}
      (\textsc{Weight}) \\
      \frac{ \rho \vdash \term_1 \sem{l_1}{s_1}{w_1} \true \quad \rho \vdash \term_2 \sem{l_2}{s_2}{w_2} \termv_2 }
      {\rho \vdash \ttt{if } \term_1 \ttt{ then } \term_2 \ttt{ else } \term_3 \sem{l_1 \concat l_2}{s_1 \concat s_2}{w_1 \cdot w_2} \termv_2}
      (\textsc{If-True}) \\
      \frac{ \rho \vdash \term_1 \sem{l_1}{s_1}{w_1} \false \quad \rho \vdash \term_3 \sem{l_3}{s_3}{w_3} \termv_3 }
      {\rho \vdash \ttt{if } \term_1 \ttt{ then } \term_2 \ttt{ else } \term_3 \sem{l_1 \concat l_3}{s_1 \concat s_3}{w_1 \cdot w_3} \termv_3}
      (\textsc{If-False}) \\
    \end{gathered}
  \]
  \caption{%
    A big-step operational semantics for terms, formalizing single runs of programs $\term \in T$.
    The operation $\rho, x \mapsto \termv$ produces a new environment extending $\rho$ with a binding \termv{} for $x$.
    For each distribution $d \in D$, $f_d$ is its \emph{probability density} or \emph{probability mass} function---encoding the relative probability of drawing particular values from the distribution.
    For example, $f_{\text{Bernoulli}(0.3)}(\true) = 0.3$ and $f_{\text{Bernoulli}(0.3)}(\false) = 1-0.3 = 0.7$.
    We use $\cdot$ to denote multiplication.
  }
  \label{fig:semantics}
\end{figure}
We now introduce a big-step operational semantics for single runs of programs $\term$.
Such a semantics is essential to formalize the probability distributions encoded by probabilistic programs (e.g., Fig.~\ref{fig:geo:b} for Fig.~\ref{fig:geo:a}) and to prove the correctness of PPL inference algorithms.
For example, Borgström et al.~\cite{borgstrom2016lambda} define a PPL calculus and semantics similar to this paper and formally proves the correctness of an MCMC algorithm.
Another example is Lundén et al.~\cite{lunden2021correctness}, who also define a similar calculus and semantics and prove the correctness of PPL SMC algorithms.
In particular, the correctness of our aligned SMC algorithm (Section~\ref{sec:smcalign}) follows from this proof.
The purpose of the semantics in this paper is to formalize alignment and prove the soundness of our analysis in Section~\ref{sec:align}.
We use a big-step semantics as the finer granularity in a small-step semantics is redundant.
We begin with a definition for intrinsics.
\begin{definition}[Intrinsic functions]\label{def:const}
  For every $c \in C$, we attach an \emph{arity} $|c| \in \mathbb{N}$.
  We define a partial function $\delta : C \times C \rightarrow C$ such that $\delta(c, c_1) = c_2$ is defined for $|c| > 0$.
  For all $c$, $c_1$, and $c_2$, such that $\delta(c,c_1) = c_2$, $|c_2|$ = $|c| - 1$.
\end{definition}
\noindent
Intrinsic functions are curried and produce intrinsic or intrinsic functions of one arity less through $\delta$.
For example, for $+ \in C$, we have $\delta(\delta(+, 1),2) = 3$, $|\!+\!| = 2$, $|\delta(+,1)| = 1$, and $|\delta(\delta(+, 1),2)| = 0$.
Next, randomness in our semantics is deterministic via a \emph{trace} of random draws in the style of Kozen~\cite{kozen1981semantics}.
\begin{definition}[Traces]\label{def:trace}
  The set $S$ of traces is the set such that, for all $s \in S$, $s$ is a sequence of intrinsics from $C$ with arity 0.
\end{definition}
\noindent
In the following, we use the notation $[c_1,c_2,\ldots,c_n]$ for sequences and $\concat$ for sequence concatenation.
For example, $[c_1,c_2]\concat[c_2,c_4] = [c_1,c_2,c_3,c_4]$.
We also use subscripts to select elements in a sequence, e.g., $[c_1,c_2,c_3,c_4]_2 = c_2$.
In practice, traces are often sequences of real numbers, e.g., $[1.1,3.2,8.4] \in S$.

Fig.~\ref{fig:semantics} presents the semantics as a relation $\rho \vdash \term \sem{l}{s}{w} \termv$ over $P \times T \times S \times \mathbb{R} \times L \times V$.
$L$ is the set of sequences over $X$, i.e., sequences of names.
For example, $[x,y,z] \in L$, where $x,y,z \in X$.
We use $l \in L$ to track the sequence of \ttt{let}-bindings during evaluation.
For example, evaluating \ttt{let $x$ = $1$ in let $y$ = $2$ in $x + y$} results in $l = [x,y]$.
In Section~\ref{sec:align}, we use the sequence of encountered \ttt{let}-bindings to define alignment.
For simplicity, from now on we assume that bound variables are always unique (i.e., variable shadowing is impossible).

It is helpful to think of $\rho$, \term, and $s$ as the input to $\Downarrow$, and $l$, $w$ and $\termv$ as the output.
In the environment $\rho$, \term, with trace $s$, evaluates to \termv{}, encounters the sequence of \ttt{let} bindings $l$, and accumulates the weight $w$.
The trace $s$ is the sequence of all random draws, and each random draw in (\tsc{Assume}) consumes precisely one element of $s$.
The rule (\tsc{Let}) tracks the sequence of bindings by adding $x$ at the correct position in $l$.
The number $w$ is the likelihood of the execution---the probability density of all draws in the program, registered at (\tsc{Assume}), combined with direct likelihood modifications, registered at (\tsc{Weight}).
The remaining aspects of the semantics are standard (see, e.g., Kahn~\cite{kahn1987semantics}).
To give an example of the semantics, we have
\ifextended
\begin{equation}
  \varnothing \vdash \termgeo \sem{[\mi{geometric},x,x,x,x]}{[\true,\true,\true,\false]}{0.5\cdot1.5\cdot0.5\cdot1.5\cdot0.5\cdot1.5\cdot0.5} 4
\end{equation}
\fi
for the particular execution of $\termgeo$ making three recursive calls.
Next, we formalize and apply the alignment analysis to~\eqref{eq:ast}.

\section{Alignment Analysis}\label{sec:align}
This section presents the main contribution of this paper: automatic alignment in PPLs.
Section~\ref{sec:anf} introduces A-normal form and gives a precise definition of alignment.
Section~\ref{sec:analysis} formalizes and proves the correctness of the alignment analysis.
Lastly, Section~\ref{sec:dynamic} discusses a dynamic version of alignment.

\subsection{A-Normal Form and Alignment}\label{sec:anf}
To reason about all subterms $\term'$ of a program $\term$ and to enable the analysis in Section~\ref{sec:analysis}, we need to uniquely label all subterms.
A straightforward approach is to use variable names within the program itself as labels (remember that we assume bound variables are always unique).
This leads us to the standard A-normal form (ANF) representation of programs~\cite{flanagan1993essence}.
\begin{definition}[A-normal form]
  \begin{equation}\label{eq:anf}
    \begin{aligned}
      \termanf \Coloneqq& \s
      x
      \s | \s
      \ttt{let } x = \termanf' \ttt{ in } \termanf
      \\
      \termanf' \Coloneqq& \s
      x
      \s | \s
      c
      \s | \s
      \lambda x. \s \termanf
      \s | \s
      x \s y
      \\ |& \s
      \ttt{if } x \ttt{ then } \termanf \ttt{ else } \termanf
      \s | \s
      \ttt{assume } x
      \s | \s
      \ttt{weight } x
    \end{aligned}
  \end{equation}
\end{definition}
We use $\Termanf$ to denote the set of all terms $\termanf$.
\noindent Unlike $\term \in T$, $\termanf \in \Termanf$ enforces that a variable bound by a \ttt{let} labels each subterm in the program.
Furthermore, we can automatically transform any program in $T$ to a semantically equivalent $\Termanf$ program, and $\Termanf \subset T$.
Therefore, we assume in the remainder of the paper that all terms are in ANF.

Given the importance of alignment in universal PPLs, it is somewhat surprising that there are no previous attempts to give a formal definition of its meaning. Here, we give a first such formal definition, but before defining alignment, we require a way to restrict, or filter, sequences.
\begin{definition}[Restriction of sequences]\label{def:seqrestr}
  For all $l \in L$ and $Y \subseteq X$, $l|_Y$ (the restriction of $l$ to $Y$) is the subsequence of $l$ with all elements not in $Y$ removed.
\end{definition}
\noindent
For example, $[x,y,z,y,x]|_{\{x,z\}} = [x,z,x]$.
We now formally define alignment.
\begin{definition}[Alignment]\label{def:align}
  For $\term \in \Termanf$, let $X_\term$ denote all variables that occur in $\term$.
  The sets $A_\term \in \mathcal{A}_\term$, $A_\term \subseteq X_\term$, are the largest sets such that, for arbitrary $\varnothing \vdash \term \sem{l_1}{s_1}{w_1} \termv_1$ and $\varnothing \vdash \term \sem{l_2}{s_2}{w_2} \termv_2$, $l_1|_{A_\term} = l_2|_{A_\term}$.
\end{definition}
\noindent For a given $A_\term$, the \emph{aligned expressions}---expressions bound by a \texttt{let} to a variable name in $A_\term$---are those that occur in the same order in every execution, regardless of random draws.
We seek the largest sets, as $A_\term = \varnothing$ is always a trivial solution.
Assume we have a program with $X_\term = \{x,y,z\}$ and such that $l = [x,y,x,z,x]$ and $l = [x ,y,x,z,x,y]$ are the only possible sequences of \texttt{let} bindings.
Then, $A_\term = \{x,z\}$ is the only possibility.
It is also possible to have multiple choices for $A_\term$.
For example, if $l = [x,y,z]$ and $l = [x,z,y]$ are the only possibilities, then $\mathcal{A}_\term = \{\{x,z\}, \{x,y\}\}$.
Next, assume that we transform the programs in Fig.~\ref{fig:geo:a} and Fig.~\ref{fig:mot:prog} to ANF.
The expression labeled by $x$ in Fig.~\ref{fig:geo:a} is then clearly not aligned, as random draws determine how many times it executes ($l$ could be, e.g., $[x,x]$ or $[x,x,x,x]$).
Conversely, the expression $n$ (line~\ref{line:mot:poisson}) in Fig.~\ref{fig:mot:prog} is aligned, as its number and order of evaluations do not depend on any random draws.

Definition~\ref{def:align} is \emph{context insensitive}: for a given $A_\term$, each $x$ is either aligned or unaligned.
One could also consider a context-sensitive definition of alignment in which $x$ can be aligned in some contexts and unaligned in others.
A context could, for example, be the sequence of function applications (i.e., the call stack) leading up to an expression.
Considering different contexts for $x$ is complicated and difficult to take full advantage of.
We justify the choice of context-insensitive alignment with the real-world models in Section~\ref{sec:eval}, neither of which requires a context-sensitive alignment.

With alignment defined, we now move on to the static alignment analysis.

\subsection{Alignment Analysis}\label{sec:analysis}
The basis for the alignment analysis is 0-CFA~\cite{nielson1999principles,shivers1991control}---a static analysis framework for higher-order functional programs.
The prefix 0 indicates that 0-CFA is context insensitive.
There is also a set of analyses $k$-CFA~\cite{midtgaard2021control} that adds increasing amounts (with $k \in \mathbb{N}$) of context sensitivity to 0-CFA.
We could use such analyses with a context-sensitive version of Definition~\ref{def:align}.
However, the potential benefit of $k$-CFA is also offset by the worst-case exponential time complexity, already at $k = 1$.
In contrast, the time complexity of 0-CFA is polynomial (cubic in the worst-case).
The alignment analysis for the models in Section~\ref{sec:eval} runs instantaneously, justifying that the time complexity is not a problem in practice.

The extensions to 0-CFA required to analyze alignment are non-trivial to design, but the resulting formalization is surprisingly simple.
The challenge is instead to prove that the extensions correctly capture the alignment property from Definition~\ref{def:align}.
We extend 0-CFA to analyze stochastic values and alignment in programs $\term \in \Termanf$.
As with most static analyses, our analysis is sound but conservative (i.e., sound but incomplete)---the analysis may mark aligned expressions of programs as unaligned, but not vice versa.
That the analysis is conservative does not degrade the alignment analysis results for any model in Section~\ref{sec:eval}, which justifies the approach.
We divide the formal analysis into two algorithms.
Algorithm~\ref{alg:gencstr} generates \emph{constraints} for $\term$ that a valid analysis solution must satisfy.
This section describes Algorithm~\ref{alg:gencstr} and the generated constraints.
\ifextended
Appendix~\ref{sec:analysisalg} provides the second algorithm, Algorithm~\ref{alg:flow}, that computes a solution satisfying the generated constraints.
We provide examples of applying Algorithm~\ref{alg:flow} here, but defer the complete description to Appendix~\ref{sec:analysisalg}.
\fi

For soundness of the analysis, we require $\langle\lambda x. \s \term, \rho \rangle \not\in C$ (recall that $C$ is the set of intrinsics).
That is, closures are \emph{not} in $C$.
By Definition~\ref{def:trace}, this implies that closures are not in the sample space of probability distributions in $D$ and that evaluating intrinsics never produces closures (this would unnecessarily complicate the analysis without any benefit).

\begin{figure}[tb]
  \lstset{%
    basicstyle=\ttfamily\scriptsize,
    numbers=left,
    showlines=true,
    numbersep=3pt,
    framexleftmargin=-2pt,
    xleftmargin=2em,
    language=calc
  }
  \centering
  \begin{multicols}{2}
    \begin{lstlisting}[]
let $n_1$$\label{line:n1}$ = $\neg$ in let $n_2$ = $\neg$ in
let $\mi{one}$ = 1 in
let $\mi{half}$ = 0.5 in let $c$ = $\textrm{true}$ in
let $f_1$ = $\lambda x_1\!$.$\!$ let $\!t_1\!$ = weight $\!\mi{one}\!$ in $x_1\!$ in$\label{line:lambdaf1}$
let $f_2$ = $\lambda x_2\!$.$\!$ let $\!t_2\!$ = weight $\!\mi{one}\!$ in $t_2$ in
let $f_3$ = $\lambda x_3\!$.$\!$ let $\!t_3\!$ = weight $\!\mi{one}\!$ in $t_3$ in
let $f_4$ = $\lambda x_4\!$.$\!$ let $\!t_4\!$ = weight $\!\mi{one}\!$ in $t_4$ in
let $\mi{bern}$ = $\textrm{Bernoulli}$ in
let $d_1$ = $\mi{bern}$ $\mi{half}$ in
let $a_1$ = assume $d_1$$\label{line:a1}$
let $v_1$ = $f_1$ $\mi{one}$ in (*@ \columnbreak @*)
let $v_2$ = $n_1$ $a_1$ in$\label{line:v2}$
let $v_3$ = $n_2$ $c$ in
let $f_5$ =$\label{line:f5}$
  if $a_1$ then let $t_5$ = $f_4$ $\mi{one}$ in $f_2$$\label{line:if2}$
  else $f_3$
in
let $v_4$ = $f_5$ $\mi{one}$ in$\label{line:appf5}$$\label{line:v4}$
let $i_1$ =
  if $c$ then let $t_6$ = $f_1$ $\mi{one}$ in $t_6$
  else $\mi{one}$
in $i_1$
    \end{lstlisting}
  \end{multicols}
  \caption{%
    A program $\termanfex \in \Termanf$ illustrating the analysis.
  }
  \label{fig:running}
\end{figure}
In addition to standard 0-CFA constraints, Algorithm~\ref{alg:gencstr} generates new constraints for \emph{stochastic values} and \emph{unalignment}.
We use the contrived but illustrative program in Fig.~\ref{fig:running} as an example.
Note that, while omitted from Fig.~\ref{fig:running} for ease of presentation, the analysis also supports recursion introduced through \texttt{let~rec}.
Stochastic values are values in the program affected by random variables.
Stochastic values initially originate at \ttt{assume} and then propagate through programs via function applications and \ttt{if} expressions.
For example, $a_1$ (line~\ref{line:a1}) is stochastic because of \ttt{assume}.
We subsequently use $a_1$ to define $v_2$ via $n_1$ (line~\ref{line:v2}), which is then also stochastic.
Similarly, $a_1$ is the condition for the \ttt{if} resulting in $f_5$ (line~\ref{line:f5}), and the function $f_5$ is therefore also stochastic.
When we apply $f_5$, it results in yet another stochastic value, $v_4$ (line~\ref{line:v4}).
In conclusion, the stochastic values are $a_1$, $v_2$, $f_5$, and $v_4$.

Consider the flow of unalignment in Fig.~\ref{fig:running}.
We mark expressions that may execute due to stochastic branching as unaligned.
From our analysis of stochastic values, the program's only stochastic \ttt{if} condition is at line~\ref{line:if2}, and we determine that all expressions directly within the branches are unaligned.
That is, the expression labeled by $t_5$ is unaligned.
Furthermore, we apply the variable $f_4$ when defining $t_5$.
Thus, \emph{all} expressions in bodies of lambdas that flow to $f_4$ are unaligned.
Here, it implies that $t_4$ is unaligned.
Finally, we established that the function $f_5$ produced at line~\ref{line:if2} is stochastic.
Due to the application at line~\ref{line:appf5}, all names bound by \ttt{let}s in bodies of lambdas that flow to $f_5$ are unaligned.
Here, it implies that $t_2$ and $t_3$ are unaligned.
In conclusion, the unaligned expressions are named by $t_2$, $t_3$, $t_4$, and $t_5$.
For example, aligned SMC therefore resamples at the \ttt{weight} at $t_1$, but not at the \ttt{weight}s at $t_2$, $t_3$, and $t_4$.

Consider the program in Fig.~\ref{fig:mot:prog} again, and assume it is transformed to ANF.
The alignment analysis must mark all names bound within the stochastic \ttt{if} at line~\ref{line:mot:stochbranch} as unaligned because a stochastic value flows to its condition.
In particular, the \ttt{weight} expressions at lines~\ref{line:mot:weight1} and \ref{line:mot:weight2} are unaligned (and the \ttt{weight} at line~\ref{line:mot:weight3} is aligned).
Thus, aligned SMC resamples only at line~\ref{line:mot:weight3}.

\ifextended
To formalize the flow of stochastic values, we define \emph{abstract values} $\absval \in A$, that flow within the program, as follows.
\begin{definition}[Abstract values]
  $
  \absval \Coloneqq
  \lambda x. y
  \s | \s
  \ttt{stoch}
  \s | \s
  \ttt{const} \s n
  $
  where
  $x,y \in X$ and $n \in \mathbb{N}$.
\end{definition}
\noindent
\fi
The \ttt{stoch} abstract value is new and represents stochastic values.
The $\lambda x. y$ and \ttt{const}~$n$ abstract values are standard and represent abstract closures and intrinsics, respectively.
For each variable name $x$ in the program, we define a set $S_x$ containing abstract values that may occur at $x$.
For example, in Fig.~\ref{fig:running}, we have $\ttt{stoch} \in S_{a_1}$, $(\lambda x_2. t_2) \in S_{f_2}$, and $(\ttt{const}$~$1) \in S_{n_1}$.
The abstract value $\lambda x_2. t_2$ represents all closures originating at $\lambda x_2$, and $\ttt{const}$~$1$ represents intrinsic functions in $C$ of arity 1 (in our example, $\neg$).
The body of the abstract lambda is the variable name labeling the body, not the body itself.
For example, $t_2$ labels the body \ttt{let $t_2$ = $\mi{one}$ in $t_2$} of $\lambda x_2$.
Due to ANF, all terms have a label, which the function \textsc{name} in Algorithm~\ref{alg:gencstr} formalizes.

\begin{algorithm}[tb]
  \renewcommand{\s}{\hphantom{|}}
  \caption{%
    Constraint generation function for $\term \in \Termanf$.
    We denote the power set of a set $E$ with $\mathcal P(E)$.
  }\label{alg:gencstr}
  \raggedright
  \lstinline[style=alg]!function $\tsc{generateConstraints}$($\term$): $\Termanf \rightarrow \mathcal P(R)$ =!\\
  \vspace{-1mm}
  \hspace{4.0mm}%
  \begin{minipage}{0.96\textwidth}
    \begin{multicols}{2}
      \begin{lstlisting}[
          style=alg,
          basicstyle=\sffamily\scriptsize,
          numbers=left,
          showlines=true,
          numbersep=3pt
        ]
match $\term$ with
| $x \rightarrow$ $\varnothing$$\label{line:genvar}$
| $\ttt{let } x = \term_1 \s \ttt{in} \s \term_2 \rightarrow$$\label{line:genlet}$
$\s$ $\tsc{generateConstraints}(\term_2) \s \cup$
  $\s$ match $\term_1$ with
  $\s$ | $y \rightarrow \{S_y \subseteq S_x\}$$\label{line:genalias}$
  $\s$ | $c \rightarrow$ if $|c| > 0$ then $\{\ttt{const} \s |c| \in S_x\}$$\label{line:genconst}$
  $\s$ $\s$ $\hspace{5mm}$ else $\varnothing$
  $\s$ | $\lambda y. \s \term_y \rightarrow$ $\tsc{generateConstraints}(\term_y)$$\label{line:genabs}$
  $\s$ $\s$ $\s$ $\cup$ $\{\lambda y. \s \tsc{name}(\term_y) \in S_x \}$
  $\s$ $\s$ $\s$ $\cup$ $\{ \unaligned_y \Rightarrow \unaligned_n$
  $\s$ $\s$ $\s$ $\s$ $\s$ $\s$ $\mid n \in \tsc{names}(t_y)\}$
  $\s$ | $\mi{lhs} \s \mi{rhs} \rightarrow$ $\{$$\label{line:genapp}$
    $\s$ $\s$ $\forall z \forall y \s \lambda z. y \in S_{\mi{lhs}}$
    $\s$ $\s$ $\s$ $\Rightarrow (S_\mi{rhs} \subseteq S_z) \land (S_y \subseteq S_x)$,
    $\s$ $\s$ $\forall n \s (\ttt{const} \s n \in S_\mi{lhs}) \land (n > 1)$
    $\s$ $\s$ $\s$ $\Rightarrow \ttt{const} \s n-1 \in S_x$,
    $\s$ $\s$ $\ttt{stoch} \in S_\mi{lhs} \Rightarrow \ttt{stoch} \in S_x$,
    $\s$ $\s$ $\ttt{const} \s \_ \in S_\mi{lhs}$
    $\s$ $\s$ $\s$ $\Rightarrow (\ttt{stoch} \in S_\mi{rhs} \Rightarrow \ttt{stoch} \in S_x)$,
    $\s$ $\s$ $\unaligned_x$
    $\s$ $\s$ $\s$ $\Rightarrow (\forall y \s \lambda y. \_ \in S_{\mi{lhs}} \Rightarrow \unaligned_y)$,
    $\s$ $\s$ $\ttt{stoch} \in S_{\mi{lhs}}$
    $\s$ $\s$ $\s$ $\Rightarrow (\forall y \s \lambda y. \_ \in S_{\mi{lhs}} \Rightarrow \unaligned_y)$
  $\s$ $\s$ $\}$(*@ \columnbreak @*)
  $\s$ | $\ttt{if } y \ttt{ then } \term_t \ttt{ else } \term_e \rightarrow$ $\label{line:genif}$
  $\s$ $\s$ $\tsc{generateConstraints}(\term_t)$
    $\s$ $\s$ $\cup \s \tsc{generateConstraints}(\term_e)$
    $\s$ $\s$ $\cup \s \{ S_{\tsc{name}(\term_t)} \subseteq S_x, S_{\tsc{name}(\term_e)} \subseteq S_x,$
    $\s$ $\s$ $\s$ $\s$ $\s$ $\ttt{stoch} \in S_y \Rightarrow \ttt{stoch} \in S_x \}$
    $\s$ $\s$ $\cup \s \{\unaligned_x \Rightarrow \unaligned_n$
           $\mid n \in \tsc{names}(\term_t) \cup \tsc{names}(\term_e)\}$
    $\s$ $\s$ $\cup \s \{\ttt{stoch} \in S_y \Rightarrow \unaligned_n$
           $\mid n \in \tsc{names}(\term_t) \cup \tsc{names}(\term_e)\}$
  $\s$ | $\ttt{assume } \textrm{\_} \rightarrow \{\ttt{stoch} \in S_x\}$$\label{line:genassume}$
  $\s$ | $\ttt{weight } \textrm{\_} \rightarrow \varnothing$$\label{line:genweight}$$\label{line:genletend}$

function $\tsc{name}$($\term$): $\Termanf \rightarrow X$ =
  match $\term$ with
  | $x \rightarrow x$ $\qquad$
  | $\ttt{let } x = \term_1 \s \ttt{in} \s \term_2 \rightarrow$ $\tsc{name}$(t$_2$)

function $\tsc{names}$($\term$): $\Termanf \rightarrow \mathcal P (X)$ =$\label{line:names}$
  match $\term$ with
  | $x \rightarrow \varnothing$ $\qquad$
  | $\ttt{let } x = \_ \s \ttt{in} \s \term_2 \rightarrow$ $\{x\} \cup \tsc{names}$(t$_2$)




      \end{lstlisting}
    \end{multicols}
  \end{minipage}
\end{algorithm}%
We also define booleans $\unaligned_x$ that state whether or not the expression labeled by $x$ is unaligned.
For example, we previously reasoned that $\unaligned_x = \true$ for $x \in \{ t_2, t_3, t_4, t_5 \}$ in Fig.~\ref{fig:running}.
The alignment analysis aims to determine \emph{minimal} sets $S_x$ and boolean assignments of $\unaligned_x$ for every program variable $x \in X$.
A trivial solution is that all abstract values (there is a finite number of them in the program) flow to each program variable and that $\unaligned_x = \true$ for all $x \in X$.
This solution is sound but useless.
To compute a more precise solution, we follow the rules given by \emph{constraints}
$\cstr \in R$\ifextended\ (see Appendix~\ref{sec:aligncont} for a formal definition).\fi

We present the constraints through the \textsc{generateConstraints} function in Algorithm~\ref{alg:gencstr} and for the example in Fig.~\ref{fig:running}.
There are no constraints for variables that occur at the end of ANF \texttt{let} sequences (line~\ref{line:genvar} in Algorithm~\ref{alg:gencstr}), and the case for \texttt{let} expressions (lines~\ref{line:genlet}--\ref{line:genletend}) instead produces all constraints.
The cases for aliases (line~\ref{line:genalias}), intrinsics (line~\ref{line:genconst}), \ttt{assume} (line~\ref{line:genassume}), and \ttt{weight} (line~\ref{line:genweight}) are the most simple.
Aliases of the form \texttt{let $x$ = $y$ in $\term_2$} establish $S_y \subseteq S_x$.
That is, all abstract values at $y$ are also in $x$.
Intrinsic operations results in a \texttt{const} abstract value.
For example, the definition of $n_1$ at line~\ref{line:n1} in Fig.~\ref{fig:running} results in the constraint \texttt{const $1$ $\in$ $S_{n_1}$}.
Applications of \texttt{assume} are the source of stochastic values.
For example, the definition of $a_1$ at line~\ref{line:a1} results in the constraint \texttt{stoch $\in$ $S_{a_1}$}.
Note that \texttt{assume} cannot produce any other abstract values, as we only allow distributions over intrinsics with arity 0 (see Definition~\ref{def:trace}).
Finally, we use \ttt{weight} only for its side effect (likelihood updating), and therefore \texttt{weight}s do not produce any abstract values and consequently no constraints.

The cases for abstractions (line~\ref{line:genabs}), applications (line~\ref{line:genapp}), and \ttt{if}s (line~\ref{line:genif}) are more complex.
The abstraction at line~\ref{line:lambdaf1} in Fig.~\ref{fig:running} generates (omitting the recursively generated constraints for the abstraction body $\term_y$) the constraints
$
\{\lambda x_1. x_1 \in S_{f_1} \}
\cup \{ \unaligned_{x_1} \Rightarrow \unaligned_{t_1} \}
$.
The first constraint is standard: the abstract lambda $\lambda x_1. x_1$ flows to $S_{f_1}$.
The second constraint states that if the abstraction is unaligned, all expressions in its body (here, only $t_1$) are unaligned.
We define the sets of expressions within abstraction bodies and \texttt{if} branches through the \tsc{names} function in Algorithm~\ref{alg:gencstr} (line~\ref{line:names}).

The application $f_5 \, \mi{one}$ at line~\ref{line:appf5} in Fig.~\ref{fig:running} generates the constraints
\begin{equation}
  \begin{aligned}
    \{
      &\forall z \forall y \s \lambda z. y \in S_{f_5} \Rightarrow (S_\mi{one} \subseteq S_z) \land (S_y \subseteq S_{v_4}),\\
      &\forall n \s (\ttt{const} \s n \in S_{f_5}) \land (n > 1) \Rightarrow \ttt{const} \s n-1 \in S_{v_4},\\
      &\ttt{stoch} \in S_{f_5} \Rightarrow \ttt{stoch} \in S_{v_4},\\
      &\ttt{const} \s \_ \in S_{f_5} \Rightarrow (\ttt{stoch} \in S_\mi{one} \Rightarrow \ttt{stoch} \in S_{v_4}),\\
      &\unaligned_{v_4} \Rightarrow (\forall y \s \lambda y. \_ \in S_{{f_5}} \Rightarrow \unaligned_y),\\
      &\ttt{stoch} \in S_{{f_5}} \Rightarrow (\forall y \s \lambda y. \_ \in S_{\mi{lhs}} \Rightarrow \unaligned_y)
    \}
  \end{aligned}
\end{equation}
The first constraint is standard: if an abstract value $\lambda z. y$ flows to $f_5$, the abstract values of $\mi{one}$ (the right-hand side) flow to $z$.
Furthermore, the result of the application, given by the body name $y$, must flow to the result $v_4$ of the application.
The second constraint is also relatively standard: if an intrinsic function of arity $n$ is applied, it produces a const of arity $n-1$.
The other constraints are new and specific for stochastic values and unalignment.
The third constraint states that if the function is stochastic, the result is stochastic.
The fourth constraint states that if we apply an intrinsic function to a stochastic argument, the result is stochastic.
We could also make the analysis of intrinsic applications less conservative through intrinsic-specific constraints.
The fifth and sixth constraints state that if the expression (labeled by $v_4$) is unaligned or the function is stochastic, all abstract lambdas that flow to the function are unaligned.

The \ttt{if} resulting in $f_5$ at line~\ref{line:f5} in Fig.~\ref{fig:running} generates (omitting the recursively generated constraints for the branches $\term_t$ and $\term_e$) the constraints
\begin{equation}
  \begin{aligned}
    \{
      &S_{\tsc{name}(f_2)} \subseteq S_{f_5}, S_{\tsc{name}(f_3)} \subseteq S_{f_5},
      \ttt{stoch} \in S_{a_1} \Rightarrow \ttt{stoch} \in S_{f_5}
    \} \\
    &\cup \{ \unaligned_{f_5} \Rightarrow \unaligned_{t_5} \}
    \cup \{\ttt{stoch} \in S_{a_1} \Rightarrow \unaligned_{t_5}\}
  \end{aligned}
\end{equation}
The first two constraints are standard and state that the result of the branches flows to the result of the \ttt{if} expression.
The remaining constraints are new.
The third constraint states that if the condition is stochastic, the result is stochastic.
The last two constraints state that if the \ttt{if} is unaligned or if the condition is stochastic, all names in the branches (here, only $t_5$) are unaligned.

Given constraints for a program, we need to compute a solution satisfying all constraints.
We do this by repeatedly iterating through all the constraints and propagating abstract values accordingly.
We terminate when we reach a fixed point, i.e., when no constraint results in an update of either $S_x$ or $\unaligned_x$ for any $x$ in the program.
\ifextended
Algorithm~\ref{alg:flow} in Appendix~\ref{sec:analysisalg} formalizes our extension of the 0-CFA constraint propagation algorithm that also handles the constraints generated for tracking stochastic values and unalignment.
The analysis function \tsc{analyzeAlign}: $\Termanf \rightarrow ((X \rightarrow \mathcal{P}(A)) \times \mathcal P(X))$ returns a map associating each variable to a set of abstract values and a set of unaligned variables.
\fi
In other words, \tsc{analyzeAlign} computes a solution to $S_x$ and $\unaligned_x$ for each $x$ in the analyzed program.
For example, \tsc{analyzeAlign}$(\termanfex)$ results in
\begin{equation}\label{eq:vres}
  \begin{gathered}
    S_{n_1} = \{ \ttt{const } 1 \} \enspace
    S_{n_2} = \{ \ttt{const } 1 \} \enspace
    S_{f_1} = \{ \lambda x_1. x_1 \} \enspace
    S_{f_2} = \{ \lambda x_2. t_2 \} \\
    S_{f_3} = \{ \lambda x_3. t_3 \} \enspace
    S_{f_4} = \{ \lambda x_4. t_4 \} \enspace
    S_{a_1} = \{ \ttt{stoch} \} \enspace
    S_{v_2} = \{ \ttt{stoch} \} \\
    S_{f_5} = \{ \lambda x_2. t_2, \lambda x_3. t_3, \ttt{stoch} \} \enspace
    S_{v_4} = \{ \ttt{stoch} \} \enspace
    S_n = \varnothing \mid \text{other $n \in X$} \\
    \unaligned_n = \true \mid n \in \{ t_2, t_3, t_4, t_5 \} \enspace
    \unaligned_n = \false \mid \text{other $n \in X$}.
  \end{gathered}
\end{equation}
The example confirms our earlier intuition: an intrinsic ($\neg$) flows to $n_1$, \ttt{stoch} flows to $a_1$, $f_5$ is stochastic and originates at either $(\lambda x_2. t_2)$ or $(\lambda x_3. t_3)$, and the unaligned variables are $t_2$, $t_3$, $t_4$, and $t_5$.
We now give soundness results.
\begin{lemma}[0-CFA soundness]\label{lemma:cfa}
  For every $\term \in \Termanf$, the solution produced by \tsc{analyzeAlign}$(\term)$ satisfies the constraints \tsc{generateConstraints}$(\term)$.
\end{lemma}
\begin{proof}
  The well-known soundness of 0-CFA extends to the new alignment constraints. See, e.g., Nielson et al.~\cite[Chapter 3]{nielson1999principles} and Shivers~\cite{shivers1991control}.
  \qed
\end{proof}
\begin{theorem}[Alignment analysis soundness]\label{thm:main}
  Assume $\term \in \Termanf$, $\mathcal{A}_\term$ from Definition~\ref{def:align}, and an assignment to $S_x$ and $\unaligned_x$ for $x \in X$ according to \tsc{analyzeAlign}$(\term)$.
  Let $\restt = \{ x \mid \neg\unaligned_x\}$ and take arbitrary
  $\varnothing \vdash \term \sem{l_1}{s_1}{w_1} \termv_1$ and
  $\varnothing \vdash \term \sem{l_2}{s_2}{w_2} \termv_2$.
  Then, $l_1|_\restt = l_2|_\restt$ and consequently $\restt \subseteq A_\term$ for at least one $A_\term \in \mathcal{A}_\term$.
\end{theorem}
\ifextended
\begin{proof}
  Follows by Lemma~\ref{lemma:aligned} in Appendix~\ref{sec:proof} with $\term' = \term$ and $\rho_1 = \rho_2 = \varnothing$.
  The proof uses simultaneous structural induction over the derivations $\varnothing \vdash \term \sem{l_1}{s_1}{w_1} \termv_1$ and
  $\varnothing \vdash \term \sem{l_2}{s_2}{w_2} \termv_2$.
  At corresponding stochastic branches or stochastic function applications in the two derivations, a separate structural induction argument shows that, for the \ttt{let}-sequences $l_1'$ and $l_2'$ of the two stochastic subderivations, $l_{1}'|_\restt = l_{2}'|_\restt = []$.
  Combined, the two arguments give the result.
  \qed
\end{proof}
\fi

The result $\restt \subseteq A_\term$ (cf. Definition~\ref{def:align}) shows that the analysis is conservative.

\subsection{Dynamic Alignment}\label{sec:dynamic}
An alternative to static alignment is \emph{dynamic} alignment, which we explored in early stages when developing the alignment analysis.
Dynamic alignment is fully context sensitive and amounts to introducing variables in programs that track (at runtime) when evaluation enters stochastic branching.
To identify these stochastic branches, dynamic alignment also requires a runtime data structure that keeps track of the stochastic values.
Similarly to $k$-CFA, dynamic alignment is potentially more precise than the 0-CFA approach.
However, we discovered that dynamic alignment introduces significant runtime overhead.
Again, we note that the models in Section~\ref{sec:eval} do not require a context-sensitive analysis, justifying the choice of 0-CFA over dynamic alignment and $k$-CFA.

\section{Aligned SMC and MCMC}\label{sec:mcalign}
This section presents detailed algorithms for aligned SMC (Section~\ref{sec:smcalign}) and aligned lightweight MCMC (Section~\ref{sec:mcmcalign}).
For a more pedagogical introduction to the algorithms, see Section~\ref{sec:motivating}.
We assume a basic understanding of SMC and Metropolis--Hastings MCMC algorithms (see, e.g., Bishop~\cite{bishop2006pattern}).

\subsection{Aligned SMC}\label{sec:smcalign}
\begin{algorithm}[tb]
  \caption{%
    Aligned SMC. The input is a program $\term \in \Termanf$ and the number of execution instances $n$.
  }\label{alg:smc}
  \scriptsize
  \vspace{-3mm}
  \begin{enumerate}
    \item Run the alignment analysis on $\term$, resulting in $\restt$ (see Theorem~\ref{thm:main}).
    \item Initiate $n$ execution instances $\{e_i \mid i \in \mathbb{N}, 1 \leq i \leq n\}$ of $\term$.
    \item\label{alg:smc:start}
      Execute all $e_i$ and suspend execution upon reaching an aligned weight (i.e., \ttt{let $x$ = weight $w$ in \term} and $x \in \restt$) or when the execution terminates naturally.
      The result is a new set of execution instances $e_i'$ with weights $w_i'$ accumulated from unaligned \ttt{weight}s and the single final aligned \ttt{weight} during execution.
    \item\label{alg:smc:terminate}
      If all $e_i' = \termv_i'$ (i.e., all executions have terminated and returned a value), terminate inference and return the set of weighted samples $(\termv_i',w_i')$.
      The samples approximate the posterior probability distribution encoded by $\term$.
    \item
      Resample the $e_i'$ according to their weights $w_i'$.
      The result is a new set of unweighted execution instances $e_i''$.
      Set $e_i \gets e_i''$.
      Go to \ref{alg:smc:start}.
  \end{enumerate}
  \vspace{-3mm}
\end{algorithm}%
We saw in Section~\ref{sec:alignedsmcmot} that SMC operates by executing many instances of $\term$ concurrently, and resampling them at calls to \ttt{weight}.
Critically, resampling requires that the inference algorithm can both suspend and resume executions.
Here, we assume that we can create execution instances $e$ of the probabilistic program \term, and that we can arbitrarily suspend and resume the instances.
The technical details of suspension are beyond the scope of this paper.
See Goodman and Stuhlmüller~\cite{goodman2014design}, Wood et al.~\cite{wood2014new}, and Lundén et al.~\cite{lunden2022compiling} for further details.

Algorithm~\ref{alg:smc} presents all steps for the aligned SMC inference algorithm.
After running the alignment analysis and setting up the $n$ execution instances, the algorithm iteratively executes and resamples the instances.
Note that the algorithm resamples only at aligned weights (see Section~\ref{sec:alignedsmcmot}).

\begin{figure}[tb]
  \lstset{%
    basicstyle=\ttfamily\scriptsize,
    numbers=left,
    showlines=true,
    numbersep=3pt,
    framexleftmargin=-2pt,
    xleftmargin=2em,
    language=calc
  }
  \centering
  \begin{subfigure}{0.49\textwidth}
    \centering
    \begin{tabular}{c}
      \begin{lstlisting}
if assume $\textrm{Bernoulli}(0.5)$ then
  weight $1$; weight $10$; $\true$
else
  weight $10$; weight $1$; $\false$
      \end{lstlisting}
    \end{tabular}
    \caption{Aligned better than unaligned.}
    \label{fig:smclim:align}
  \end{subfigure}
  \begin{subfigure}{0.49\textwidth}
    \centering
    \begin{tabular}{c}
      \begin{lstlisting}
if assume $\textrm{Bernoulli}(0.1)$ then
  weight $9$;
  if assume $\textrm{Bernoulli}(0.5)$
  then weight $1.5$ else weight $0.5$;
  $\true$
else (weight $1$; $\false$)
      \end{lstlisting}
    \end{tabular}
    \caption{Unaligned better than aligned.}
    \label{fig:smclim:unalign}
  \end{subfigure}
  \caption{%
    Programs illustrating properties of aligned and unaligned SMC.
    Fig.~(a) shows a program better suited for aligned SMC.
    Fig.~(b) shows a program better suited for unaligned SMC.
  }
  \label{fig:smclim}
\end{figure}
We conjecture that aligned SMC is preferable over unaligned SMC for all practically relevant models, as the evaluation in Section~\ref{sec:eval} justifies.
However, it is possible to construct contrived programs in which unaligned SMC has the advantage.
Consider the programs in Fig.~\ref{fig:smclim}, both encoding Bernoulli$(0.5)$ distributions in a contrived way using \texttt{weight}s.
Fig.~\ref{fig:smclim:align} takes one of two branches with equal probability.
Unaligned SMC resamples at the first \texttt{weight}s in each branch, while aligned SMC does not because the branch is stochastic.
Due to the difference in likelihood, many more \texttt{else} executions survive resampling compared to \texttt{then} executions.
However, due to the final \texttt{weight}s in each branch, the branch likelihoods even out.
That is, resampling at the first \texttt{weight}s is detrimental, and unaligned SMC performs worse than aligned SMC.
Fig.~\ref{fig:smclim:unalign} also takes one of two branches, but now with unequal probabilities.
However, the two branches still have equal posterior probability due to the \texttt{weight}s.
The nested if in the \texttt{then} branch does not modify the overall branch likelihood, but adds variance.
Aligned SMC does not resample for any \texttt{weight} within the branches, as the branch is stochastic.
Consequently, only $10\%$ of the executions in aligned SMC take the \texttt{then} branch, while half of the executions take the \texttt{then} branch in unaligned SMC (after resampling at the first \texttt{weight}).
Therefore, unaligned SMC better explores the \texttt{then} branch and reduces the variance due to the nested \texttt{if}, which results in overall better inference accuracy.
We are not aware of any real model with the property in Fig.~\ref{fig:smclim:unalign}.
In practice, it seems best to always resample when using \texttt{weight} to condition on observed data.
Such conditioning is, in practice, always done outside of stochastic branches, justifying the benefit of aligned SMC.

\subsection{Aligned Lightweight MCMC}\label{sec:mcmcalign}
\begin{algorithm}[tb]
  \caption{%
    Aligned lightweight MCMC. The input is a program $\term \in \Termanf$, the number of steps $n$, and the global step probability $g > 0$.
  }\label{alg:mcmc}
  \scriptsize
  \vspace{-3mm}
  \begin{enumerate}
    \item Run the alignment analysis on $\term$, resulting in $\restt$ (see Theorem~\ref{thm:main}).
    \item Set $i \gets 0$, $k \gets 1$, and $l \gets 1$. Call \textsc{Run}.
    \item\label{alg:mcmc:loop}
      Set $i \gets i + 1$.
      If $i = n$, terminate inference and return the samples $\{\termv_j \mid j \in \mathbb{N}, 0 \leq j < n\}$. They approximate the probability distribution encoded by \term.
    \item Uniformly draw an index $1 \leq j \leq |s_{i-1}|$ at random. Set $\mi{global} \gets \true$ with probability $g$, and $\mi{global} \gets \false$ otherwise. Set $w'_{-1} \gets 1$, $w' \gets 1$, $k \gets 1$, $l \gets 1$, and $\mi{reuse} \gets \true$. Call \textsc{Run}.
    \item\label{alg:mcmc:metropolis}
      Compute the Metropolis--Hastings acceptance ratio
      $\displaystyle
      A = \min\left(1,\frac{w_i}{w_{i-1}}\frac{w'}{w'_{-1}}\right)
      $.
    \item
      With probability $A$, accept $\termv_i$ and go to \ref{alg:mcmc:loop}.
      Otherwise, set $\termv_i \gets \termv_{i-1}$, $w_i \gets w_{i-1}$, $s_i \gets s_{i-1}$, $p_i \gets p_{i-1}$, $s'_i \gets s'_{i-1}$, $p'_i \gets p'_{i-1}$, and $n'_i \gets n'_{i-1}$. Go to \ref{alg:mcmc:loop}.
  \end{enumerate}
  \vspace{-1mm}
  \lstinline[style=alg,basicstyle=\sffamily\scriptsize]!function $\tsc{run}$() =! Run $\term$ and do the following:
  \vspace{-2mm}
  \begin{itemize}
    \item Record the total weight $w_i$ accumulated from calls to \ttt{weight}.
    \item Record the final value $\termv_i$.
    \item At \emph{unaligned} terms \ttt{let $c$ = assume $d$ in \term} ($c \not\in \restt$), do the following.
      \begin{enumerate}
        \item
          If $\mi{reuse} = \false$, $\mi{global} = \true$, $n'_{i-1,k,l} \neq c$, or if $s'_{i-1,k,l}$ does not exist, sample a value $x$ from $d$ and set $\mi{reuse} \gets \false$.
          Otherwise, reuse the sample $x = s'_{i-1,k,l}$ and set \mbox{$w'_{-1} \gets w'_{-1} \cdot p'_{i-1,k,l}$} and \mbox{$w' \gets w' \cdot f_d(c)$}.
        \item
          Set $s'_{i,k,l} \gets x$, $p'_{i,k,l} \gets f_d(x)$, and $n'_{i,k,l} \gets c$.
        \item
          Set $l \gets l + 1$.
          In the program, bind $c$ to the value $x$ and resume execution.
      \end{enumerate}
    \item
      At \emph{aligned} terms \ttt{let $c$ = assume $d$ in \term} ($c \in \restt$), do the following.
      \begin{enumerate}
        \item
          If $j = k$, $\mi{global} = \true$, or if $s_{i-1,k}$ does not exist, sample a value $x$ from $d$ normally.
          Otherwise, reuse the sample $x = s_{i-1,k}$. Set \mbox{$w'_{-1} \gets w'_{-1} \cdot p_{i-1,k}$} and \mbox{$w' \gets w' \cdot f_d(x)$}.
        \item Set $s_{i,k} \gets x$ and $p_{i,k} \gets f_d(x)$.
        \item Set $k \gets k + 1$, $l \gets 1$, and $\mi{reuse} \gets \true$.
          In the program, bind $c$ to the value $x$ and resume execution.
      \end{enumerate}
  \end{itemize}
  \vspace{-2mm}
\end{algorithm}
Aligned lightweight MCMC is a version of lightweight MCMC~\cite{wingate2011lightweight}, where the alignment analysis provides information about how to reuse random draws between executions.
Algorithm~\ref{alg:mcmc}, a Metropolis--Hastings algorithm in the context of PPLs, presents the details.
Essentially, the algorithm executes the program repeatedly using the \tsc{Run} function, and redraws one aligned random draw in each step, while reusing all other aligned draws and as many unaligned draws as possible (illustrated in Section~\ref{sec:alignedlwmot}).
\ifextended
We provide a derivation of the Metropolis--Hastings acceptance ratio in step \ref{alg:mcmc:metropolis} in Appendix~\ref{sec:mcmcaligncont}.
\fi
A key property in Algorithm~\ref{alg:mcmc} due to alignment (Definition~\ref{def:align}) is that the length of $s_i$ (and $p_i$) is constant, as executing $\term$ always results in the same number of aligned random draws.

In addition to redrawing only one aligned random draw, each step has a probability $g > 0$ of being \emph{global}---meaning that inference redraws \emph{every} random draw in the program.
Occasional global steps fix problems related to slow mixing and ergodicity of lightweight MCMC identified by Kiselyov~\cite{kiselyov2016problems}.
In a global step, the Metropolis--Hastings acceptance ratio reduces to
$
A = \min\left(1,\frac{w_i}{w_{i-1}}\right)
$.

\section{Implementation}\label{sec:implementation}
We implement the alignment analysis (Section~\ref{sec:align}), aligned SMC (Section~\ref{sec:smcalign}), and aligned lightweight MCMC (Section~\ref{sec:mcmcalign}) for the functional PPL \emph{Miking CorePPL}~\cite{lunden2022compiling}, implemented as part of the \emph{Miking} framework~\cite{broman2019vision}.
We implement the alignment analysis as a core component in the Miking CorePPL compiler, and then use the analysis when compiling to two Miking CorePPL backends: RootPPL and Miking Core.
RootPPL is a low-level PPL with built-in highly efficient SMC inference~\cite{lunden2022compiling}, and we extend the CorePPL to RootPPL compiler introduced by Lundén et al.~\cite{lunden2022compiling} to support aligned SMC inference.
Furthermore, we implement aligned lightweight MCMC inference standalone as a translation from Miking CorePPL to Miking Core.
Miking Core is the general-purpose programming language of the Miking framework, currently compiling to OCaml.

The idealized calculus in~\eqref{eq:ast} does not capture all features of Miking CorePPL.
In particular, the alignment analysis implementation must support records, variants, sequences, and pattern matching over these.
Extending 0-CFA to such language features is not new, but it does introduce a critical challenge for the alignment analysis: identifying all possible stochastic branches.
Determining stochastic \ttt{if}s is straightforward, as we simply check if \ttt{stoch} flows to the condition.
However, complications arise when we add a \ttt{match} construct (and, in general, any type of branching construct).
Consider the extension
\begin{equation}
  \begin{aligned}
    &\begin{aligned}
      \term \Coloneqq& \s
      \ldots
      \s | \s
      \ttt{match } \term \ttt{ with } \pat \ttt{ then } \term \ttt{ else } \term
      \s | \s
      \ttt{\{$k_1 = x_1$, $\ldots$, $k_n = x_n$\}}
      \\
      \pat \Coloneqq& \s
      x
      \s | \s
      \textrm{true}
      \s | \s
      \textrm{false}
      \s | \s
      \ttt{\{$k_1 = \pat$, $\ldots$, $k_n = \pat$\}}
    \end{aligned} \\
    &
    x, x_1, \ldots, x_n \in X \quad k_1, \ldots, k_n \in K \quad n \in \mathbb{N}
  \end{aligned}
\end{equation}
of~\eqref{eq:ast}, adding records and simple pattern matching.
$K$ is a set of record keys.
Assume we also extend the abstract values as
$
\absval \Coloneqq \ldots \s | \s \ttt{\{}k_1 = X_1, \ldots, k_n = X_n\ttt{\}},
$
where $X_1, \ldots, X_n \subseteq X$.
That is, we add an abstract record tracking the names in the program that flow to its entries.
Consider the program
\texttt{match $t_1$ with \{ $a = x_1$, $b = \textrm{false}$ \} then $t_2$ else $t_3$}.
This \ttt{match} is, similar to \ttt{if}s, stochastic if $\ttt{stoch} \in S_{t_1}$.
It is also, however, stochastic in other cases.
Assume we have two program variables, $x$ and $y$, such that $\ttt{stoch} \in S_x$ and $\ttt{stoch} \not\in S_y$.
Now, the \ttt{match} is stochastic if, e.g.,
$
  \ttt{\{$a = \{y\}$, $b = \{x\}$\}} \in  S_{t_1},
$
because the random value flowing from $x$ to the pattern false may not match because of randomness.
However, it is \emph{not} stochastic if, instead,
$
  S_{t_1} = \{\ttt{\{$a = \{x\}$, $b = \{y\}$\}}\}
$.
The randomness of $x$ does \emph{not} influence whether or not the branch is stochastic---the variable pattern $x_1$ for label $a$ always matches.

Our alignment analysis implementation handles the intricacies of identifying stochastic \ttt{match} cases for nested record, variant, and sequence patterns.
In total, the alignment analysis, aligned SMC, and aligned lightweight MCMC implementations consist of approximately 1000 lines of code directly contributed as part of this paper.
The code is available on GitHub~\cite{mikingdpplgithub}.

\section{Evaluation}\label{sec:eval}
This section evaluates aligned SMC and aligned lightweight MCMC on a set of models encoded in Miking CorePPL: CRBD~\cite{nee2006birth,ronquist2021universal} in Sections~\ref{sec:crbd} and \ref{sec:crbdmcmc}, ClaDS~\cite{maliet2019model,ronquist2021universal} in Section~\ref{sec:clads}, state-space aircraft localization in Section~\ref{sec:ssm}, and latent Dirichlet allocation in Section~\ref{sec:lda}.
CRBD and ClaDS are non-trivial models of considerable interest in evolutionary biology and phylogenetics~\cite{ronquist2021universal}.
Similarly, LDA is a non-trivial topic model~\cite{blei2003latent}.
Running the alignment analysis took approximately $5$ ms--$30$ ms for all models considered in the experiment, justifying that the time complexity is not a problem in practice.

We compare aligned SMC with standard unaligned SMC~\cite{goodman2014design}, which is identical to Algorithm~\ref{alg:smc}, except that it resamples at \emph{every} call to \ttt{weight}\ifextended\ (see Appendix~\ref{sec:smcunaligned}).\fi{}
We carefully checked that automatic alignment corresponds to previous manual alignments of each model.
For all SMC experiments, we estimate the \emph{normalizing constant} produced as a by-product of SMC inference rather than the complete posterior distributions.
The normalizing constant, also known as marginal likelihood or model evidence, frequently appears in Bayesian inference and gives the probability of the observed data averaged over the prior.
The normalizing constant is useful for model comparison as it measures how well different probabilistic models fit the data (a larger normalizing constant indicates a better fit).

We ran aligned and unaligned SMC with Miking CorePPL and the RootPPL backend configured for a single-core (compiled with GCC 7.5.0).
Lundén et al.~\cite{lunden2022compiling} shows that the RootPPL backend is significantly more efficient than other state-of-the-art PPL SMC implementations.
We ran aligned and unaligned SMC inference 300 times (and with 3 warmup runs) for each experiment for $10^4$, $10^5$, and $10^6$ executions (also known as \emph{particles} in SMC literature).

We compare aligned lightweight MCMC to lightweight MCMC\ifextended\ (see Appendix~\ref{sec:mcmcunaligned}).\fi{}
We implement both versions as compilers from Miking CorePPL to Miking Core, which in turn compiles to OCaml (version 4.12).
The lightweight MCMC databases are functional-style maps from the OCaml \texttt{Map} library.
We set the global step probability to $0.1$ for both aligned lightweight MCMC and lightweight MCMC.
We ran aligned lightweight and lightweight MCMC inference 300 times for each experiment.
We burned $10\%$ of samples in all MCMC runs.

For all experiments, we used an Intel Xeon 656 Gold 6136 CPU (12 cores) and 64 GB of memory running Ubuntu 18.04.5.

\subsection{SMC: Constant Rate Birth-Death (CRBD)}\label{sec:crbd}
This experiment considers the CRBD diversification model from~\cite{ronquist2021universal} applied to the Alcedinidae phylogeny (Kingfisher birds, 54 extant species)~\cite{jetz2012global}.
We use fixed diversification rates to simplify the model, as unaligned SMC inference accuracy is too poor for the full model with priors over diversification rates.
Aligned SMC is accurate for both the full and simplified models.
\ifextended
We provide the source code for the complete model in Listing~\ref{lst:crbd} of Appendix~\ref{sec:crbdcont} (130 lines of code).
\fi
The total experiment execution time was $16$ hours.

\pgfplotsset{exprun/.style={%
    hide y axis,
    axis x line*=bottom,
    ybar=2mm,
    width=6cm,
    height=30mm,
    enlarge x limits=0.3,
    ymin=0,
    xtick=data,
    xtick style={draw=none},
    xtick align=inside,
    nodes near coords,
    bar width=12pt,
    ytick=\empty,
}}
\newcommand{\exprun}[3]{%
  \addplot+[black, fill=\alignedcolor,
    error bars/.cd,
    y dir=both,y explicit
    ] table [x=particles,y=runtime,y error=std] {evaluation/#1-#2-bar.dat};
  \addplot+[black, fill=\unalignedcolor,
    error bars/.cd,
    y dir=both,y explicit
    ] table [x=particles,y=runtime,y error=std] {evaluation/#1-#3-bar.dat};
}

\pgfplotsset{expbox/.style={%
    axis line style={draw=none},
    ytick style={draw=none},
    boxplot/draw direction=y,
    width=6cm,
    height=30mm,
    xtick style={draw=none},
}}
\newcommand{\expbox}[3]{%
  \addplot [
    fill=\alignedcolor,boxplot={draw position=0.8,box extend=0.38}
  ] table [y index=0] {evaluation/#1-#2-1-box.dat};
  \addplot [
    fill=\unalignedcolor,boxplot={draw position=1.2,box extend=0.38}
  ] table [y index=0] {evaluation/#1-#3-1-box.dat};
  \pgfplotsset{cycle list shift=-2}
  \addplot [
    fill=\alignedcolor,boxplot={draw position=1.8,box extend=0.38}
  ] table [y index=0] {evaluation/#1-#2-2-box.dat};
  \addplot [
    fill=\unalignedcolor,boxplot={draw position=2.2,box extend=0.38}
  ] table [y index=0] {evaluation/#1-#3-2-box.dat};
  \pgfplotsset{cycle list shift=-4}
  \addplot [
    fill=\alignedcolor,boxplot={draw position=2.8,box extend=0.38}
  ] table [y index=0] {evaluation/#1-#2-3-box.dat};
  \addplot [
    fill=\unalignedcolor,boxplot={draw position=3.2,box extend=0.38}
  ] table [y index=0] {evaluation/#1-#3-3-box.dat};
}

\begin{figure}[tb]
  \centering
  \begin{subfigure}{0.4\columnwidth}
    \centering
    \begin{tikzpicture}[trim axis left, trim axis right]
      \scriptsize
      \begin{axis}[
        exprun,
        symbolic x coords={10000,100000,1000000},
        xticklabels={$10^6$,$10^5$,$10^4$},
        ymax=160,
        ]
        \exprun{smc/crbd}{align}{likelihood}
      \end{axis}
    \end{tikzpicture}
    \caption{Execution times.}
  \end{subfigure}
  \hfill
  \begin{subfigure}{0.5\columnwidth}
    \centering
    \begin{tikzpicture}[trim axis left, trim axis right]
      \scriptsize
      \begin{axis}[
        expbox,
        xtick={1,2,3},
        xticklabels={$10^4$,$10^5$,$10^6$},
        ytick={-315,-330},
        extra y ticks={-304.746805},
        ]
        \expbox{smc/crbd}{align}{likelihood}
      \end{axis}
    \end{tikzpicture}
    \caption{Log normalizing constant estimates.}
  \end{subfigure}%
  \caption{%
    SMC experiment results for CRBD.
    The x-axes give the number of particles.
    Fig. (a) shows execution times (in seconds) for aligned (\alignedcolor) and unaligned (\unalignedcolor) SMC.
    Error bars show one standard deviation.
    Fig. (b) shows box plot log normalizing constant estimates for aligned (\alignedcolor) and unaligned (\unalignedcolor) SMC.
    The analytically computed log normalizing constant is $-304.75$.
  }
  \label{fig:crbdres}
\end{figure}
Fig.~\ref{fig:crbdres} presents the experiment results.
Aligned SMC is roughly twice as fast and produces superior estimates of the normalizing constant.
Unaligned SMC has not yet converged to the correct value $-304.75$ (available for this particular model due to the fixing the diversification rates) for $10^6$ particles, while aligned SMC produces precise estimates already at $10^4$ particles.
Excess resampling is a significant factor in the increase in execution time for unaligned SMC, as each execution encounters far more resampling checkpoints than in aligned SMC.

\subsection{SMC: Cladogenetic Diversification Rate Shift (ClaDS)}\label{sec:clads}
A limitation of CRBD is that the diversification rates are constant.
ClaDS~\cite{maliet2019model,ronquist2021universal} is a set of diversification models that allow \emph{shifting} rates over phylogenies.
We evaluate the ClaDS2 model for the Alcedinidae phylogeny.
As in CRBD, we use fixed (initial) diversification rates to simplify the model on account of unaligned SMC.
\ifextended
The source code for the complete model is available in Listing~\ref{lst:clads} of Appendix~\ref{sec:cladscont} (147 lines of code).
\fi
Automatic alignment simplifies the ClaDS2 model significantly, as manual alignment requires collecting and passing weights around in unaligned parts of the program, which are later consumed by aligned \ttt{weight}s.
The total experiment execution time was $67$ hours.

\begin{figure}[tb]
  \centering
  \begin{subfigure}{0.4\columnwidth}
    \centering
    \begin{tikzpicture}[trim axis left, trim axis right]
      \scriptsize
      \begin{axis}[
        exprun,
        symbolic x coords={10000,100000,1000000},
        xticklabels={$10^6$,$10^5$,$10^4$},
        ]
        \exprun{smc/clads}{align}{likelihood}
      \end{axis}
    \end{tikzpicture}
    \caption{Execution times.}
  \end{subfigure}
  \hfill
  \begin{subfigure}{0.5\columnwidth}
    \centering
    \begin{tikzpicture}[trim axis left, trim axis right]
      \scriptsize
      \begin{axis}[
        expbox,
        xtick={1,2,3},
        xticklabels={$10^4$,$10^5$,$10^6$},
        extra y ticks={-314.35},
        ytick={-400,-500}
        ]
        \expbox{smc/clads}{align}{likelihood}
      \end{axis}
    \end{tikzpicture}
    \caption{Log normalizing constant estimates.}
  \end{subfigure}%
  \caption{%
    SMC experiment results for ClaDS.
    The x-axes give the number of particles.
    Fig. (a) shows execution times (in seconds) for aligned (\alignedcolor) and unaligned (\unalignedcolor) SMC.
    Error bars show one standard deviation.
    Fig. (b) shows box plot log normalizing constant estimates for aligned (\alignedcolor) and unaligned (\unalignedcolor) SMC.
    The average estimate for aligned SMC with $10^6$ particles is $-314.35$.
  }
  \label{fig:clads}
\end{figure}

Fig.~\ref{fig:clads} presents the experiment results.
12 unaligned runs for $10^6$ particles and nine runs for $10^5$ particles ran out of the preallocated stack memory for each particle ($10$ kB).
We omit these runs from Fig.~\ref{fig:clads}.
The consequence of not aligning SMC is more severe than for CRBD.
Aligned SMC is now almost seven times faster than unaligned SMC and the unaligned SMC normalizing constant estimates are significantly worse compared to the aligned SMC estimates.
The unaligned SMC estimates do not even improve when moving from $10^4$ to $10^6$ particles (we need even more particles to see improvements).
Again, aligned SMC produces precise estimates already at $10^4$ particles.

\subsection{SMC: State-Space Aircraft Localization}\label{sec:ssm}
This experiment considers an artificial but non-trivial state-space model for aircraft localization.
\ifextended
Appendix~\ref{sec:ssmcont} presents the model as well as the source code in Listing~\ref{lst:ssm} (62 lines of code).
\fi
The total experiment execution time was $1$ hour.

\begin{figure}[tb]
  \centering
  \begin{subfigure}{0.4\columnwidth}
    \centering
    \begin{tikzpicture}[trim axis left, trim axis right]
      \scriptsize
      \begin{axis}[
        ymax=7,
        exprun,
        symbolic x coords={10000,100000,1000000},
        xticklabels={$10^6$,$10^5$,$10^4$},
        nodes near coords style={/pgf/number format/.cd,fixed}
        ]
        \exprun{smc/ssm}{align}{likelihood}
      \end{axis}
    \end{tikzpicture}
    \caption{Execution times.}
  \end{subfigure}
  \hfill
  \begin{subfigure}{0.5\columnwidth}
    \centering
    \begin{tikzpicture}[trim axis left, trim axis right]
      \scriptsize
      \begin{axis}[
        expbox,
        xtick={1,2,3},
        xticklabels={$10^4$,$10^5$,$10^6$},
        ytick={-55,-65},
        extra y ticks={-61.26},
        ]
        \expbox{smc/ssm}{align}{likelihood}
      \end{axis}
    \end{tikzpicture}
    \caption{Log normalizing constant estimates.}
  \end{subfigure}%
  \caption{%
    SMC experiment results for the state-space aircraft localization model.
    The x-axes give the number of particles.
    Fig. (a) shows execution times (in seconds) for aligned (\alignedcolor) and unaligned (\unalignedcolor) SMC.
    Error bars show one standard deviation.
    Fig. (b) shows box plot log normalizing constant estimates on the y-axis for aligned (\alignedcolor) and unaligned (\unalignedcolor) SMC.
    The average estimate for aligned SMC with $10^6$ particles is $-61.26$.
  }
  \label{fig:ssm}
\end{figure}

Fig.~\ref{fig:ssm} presents the experiment results.
The execution time difference is not as significant as for CRBD and ClaDS.
However, the unaligned SMC normalizing constant estimates are again much less precise.
Aligned SMC is accurate (centered at approximately $-61.26$) already at $10^4$ particles.
The model's straightforward control flow explains the less dramatic difference in execution time---there are at most ten unaligned likelihood updates in the aircraft model, while the number is, in theory, unbounded for CRBD and ClaDS.
Therefore, the cost of extra resampling compared to aligned SMC is not as significant.

\subsection{MCMC: Latent Dirichlet Allocation (LDA)}\label{sec:lda}
This experiment considers latent Dirichlet allocation (LDA), a topic model used in the evaluations by Wingate et al.~\cite{wingate2011lightweight} and Ritchie et al.~\cite{ritchie2016c3}.
We use a synthetic data set, comparable in size to the data set used by Ritchie et al.~\cite{ritchie2016c3}, with a vocabulary of 100 words, 10 topics, and 25 documents each containing 30 words. Note that we are not using methods based on collapsed Gibbs sampling~\cite{griffiths2004finding}, and the inference task is therefore computationally challenging even with a rather small number of words and documents.
\ifextended
The source code for the complete model is available in Listing~\ref{lst:lda} of Appendix~\ref{sec:ldacont} (31 lines of code).
\fi
The total experiment execution time was 41 hours.

The LDA model consists of only aligned random draws.
As a consequence, aligned lightweight and lightweight MCMC reduces to the same inference algorithm, and we can compare the algorithms by just considering the execution times.
\ifextended
We justify the correctness of both algorithms in Appendix~\ref{sec:ldacont}.
\fi

Fig.~\ref{fig:lda} presents the experiment results.
Aligned lightweight MCMC is almost three times faster than lightweight MCMC.
To justify the execution times with our implementations, we also implemented and ran the experiment with lightweight MCMC in WebPPL~\cite{goodman2014design} for $10^5$ iterations, repeated $50$ times (and with 3 warmup runs).
The mean execution time was $383$ s with standard deviation $5$ s.
We used WebPPL version 0.9.15 and Node version 16.18.0.

\pgfplotsset{exphist/.style={%
    axis line style={draw=none},
    xtick style={draw=none},
    ytick=\empty,
    enlargelimits=false,
}}
\newcommand{\exphist}[3]{%
  \addplot+[
      black, fill=#3, mark=none, ybar interval
    ] table [header=false] {evaluation/#1-#2-hist.dat};
}
\begin{figure}[tb]
  \centering
  \begin{tikzpicture}[trim axis left, trim axis right]
    \scriptsize
    \begin{axis}[
      exprun,
      ymax=410,
      symbolic x coords={1000,10000,100000},
      xticklabels={$10^5$,$10^4$,$10^3$},
      nodes near coords style={/pgf/number format/.cd,fixed}
      ]
      \exprun{mcmc/lda}{align}{lightweight}
    \end{axis}
  \end{tikzpicture}
  \caption{%
    MCMC experiment results for LDA showing execution time (in seconds) for aligned lightweight MCMC (\alignedcolor) and lightweight MCMC (\unalignedcolor).
    Error bars show one standard deviation and the x-axis the number of MCMC iterations.
  }
  \label{fig:lda}
\end{figure}

\subsection{MCMC: Constant Rate Birth-Death (CRBD)}\label{sec:crbdmcmc}
This experiment again considers CRBD.
MCMC is not as suitable for CRBD as SMC, and therefore we use a simple synthetic phylogeny with six leaves and an age span of 5 age units (Alcedinidae used for the SMC experiment has 54 leaves and an age span of 35 age units).
The source code for the complete model is the same as in Section~\ref{sec:crbd}, but we now allow the use of proper prior distributions for the diversification rates.
The total experiment execution time was 7 hours.

Unlike LDA, the CRBD model contains both unaligned and aligned random draws.
Because of this, aligned lightweight MCMC and standard lightweight MCMC do \emph{not} reduce to the same algorithm.
To judge the difference in inference accuracy, we consider the mean estimates of the birth diversification rate produced by the two algorithms, in addition to execution times.
The experiment results shows that the posterior distribution over the birth rate is unimodal\ifextended\ (see Appendix~\ref{sec:crbdmcmccont})\fi, which motivates using the posterior mean as a measure of accuracy.

Fig.~\ref{fig:crbdmcmcres} presents the experiment results.
Aligned lightweight MCMC is approximately $3.5$ times faster than lightweight MCMC.
There is no obvious difference in accuracy.
To justify the execution times and correctness of our implementations, we also implemented and ran the experiment with lightweight MCMC in WebPPL~\cite{goodman2014design} for $3\cdot10^6$ iterations, repeated $50$ times (and with 3 warmup runs).
The mean estimates agreed with Fig.~\ref{fig:crbdmcmcres}.
The mean execution time was $37.1$ s with standard deviation $0.8$ s.
The speedup compared to standard lightweight MCMC in Miking CorePPL is likely explained by the use of early termination in WebPPL, which benefits CRBD.
Early termination easily combines with alignment but relies on execution suspension, which we do not currently use in our implementations.
Note that aligned lightweight MCMC is faster than WebPPL even without early termination.

\begin{figure}[tb]
  \centering
  \begin{subfigure}{0.45\columnwidth}
    \centering
    \begin{tikzpicture}[trim axis left, trim axis right]
      \scriptsize
      \begin{axis}[
        exprun,
        symbolic x coords={30000,300000,3000000},
        xticklabels={$3\cdot10^6$,$3\cdot10^5$,$3\cdot10^4$},
        nodes near coords style={/pgf/number format/.cd,fixed}
        ]
        \exprun{mcmc/crbd}{align}{lightweight}
      \end{axis}
    \end{tikzpicture}
    \caption{Execution times.}
  \end{subfigure}
  \hfill
  \begin{subfigure}{0.45\columnwidth}
    \centering
    \begin{tikzpicture}[trim axis left, trim axis right]
      \scriptsize
      \begin{axis}[
        ymax=0.45,
        ytick={0.4,0.45},
        expbox,
        xtick={1,2,3},
        xticklabels={$3\cdot10^4$,$3\cdot10^5$,$3\cdot10^6$},
        extra y ticks={0.33},
        ]
        \expbox{mcmc/crbd}{align}{lightweight}
      \end{axis}
    \end{tikzpicture}
    \caption{Birth rate mean estimates.}
  \end{subfigure}%
  \caption{%
    MCMC experiment results for CRBD.
    The x-axes give the number of iterations.
    Fig. (a) shows execution times (in seconds) for aligned lightweight MCMC (\alignedcolor) and lightweight MCMC (\unalignedcolor).
    Error bars show one standard deviation.
    Fig. (b) shows box plot posterior mean estimates of the birth rate for aligned lightweight MCMC (\alignedcolor) and lightweight MCMC (\unalignedcolor).
    The average estimate for aligned lightweight MCMC with $3\cdot10^6$ iterations is $0.33$.
  }
  \label{fig:crbdmcmcres}
\end{figure}

In conclusion, the experiments clearly demonstrate the need for alignment.

\section{Related Work}\label{sec:relatedwork}

The approach by Wingate et al.~\cite{wingate2011lightweight} is closely related to ours.
A key similarity with alignment is that executions reaching the same aligned checkpoint also have matching stack traces according to Wingate et al.'s addressing transform.
However, Wingate et al. do not consider the separation between unaligned and aligned parts of the program, their approach is not static, and they do not generalize to other inference algorithms such as SMC.

Ronquist et al.~\cite{ronquist2021universal}, Turing~\cite{ge2018turing}, Anglican~\cite{wood2014new}, Paige and Wood~\cite{paige2014compilation}, and van de Meent et al.~\cite{vandemeent2018introduction} consider the alignment problem.
Manual alignment is critical for the models in Ronquist et al.~\cite{ronquist2021universal} to make SMC inference tractable, which strongly motivates the automatic alignment approach.
The documentation of Turing states that:
``The \texttt{observe} statements [i.e., likelihood updates] should be arranged so that every possible run traverses all of them in exactly the same order. This is equivalent to demanding that they are not placed inside stochastic control flow''~\cite{turing2021}.
Turing does not include any automatic checks for this property.
Anglican~\cite{wood2014new} checks, at runtime (resulting in overhead), that all SMC executions encounter the same number of likelihood updates, and thus resamples the same number of times.
If not, Anglican reports an error: ``some \ttt{observe} directives [i.e., likelihood updates] are not global''.
This error refers to the alignment problem, but the documentation does not explain it further.
Probabilistic C, introduced by Paige and Wood~\cite{paige2014compilation}, similarly assumes that the number of likelihood updates is the same in all executions.
Van de Meent et al. \cite{vandemeent2018introduction} state, in reference to SMC: ``Each breakpoint [i.e., checkpoint] needs to occur at an expression that is evaluated in every execution of a program''.
Again, they do not provide any formal definition of alignment nor an automatic solution to enforce it.


Lundén et al.~\cite{lunden2021correctness} briefly mention the general problem of selecting optimal resampling locations in PPLs for SMC but do not consider the alignment problem in particular.
They also acknowledge the overhead resulting from not all SMC executions resampling the same number of times, which alignment avoids.

The PPLs Birch~\cite{murray2018automated}, Pyro~\cite{bingham2019pyro}, and WebPPL~\cite{goodman2014design} support SMC inference.
Birch and Pyro enforce alignment for SMC as part of model construction.
Note that this is only true for SMC in Pyro---other Pyro inference algorithms use other modeling approaches.
The approaches in Birch and Pyro are sound but demand more of their users compared to the alignment approach.
WebPPL does not consider alignment and resamples at all likelihood updates for SMC.

Ritchie et al.~\cite{ritchie2016c3} and Nori et al.~\cite{nori2014r2} present MCMC algorithms for probabilistic programs.
Ritchie et al.~\cite{ritchie2016c3} optimize lightweight MCMC by Wingate et al.~\cite{wingate2011lightweight} through execution suspensions and callsite caching.
The optimizations are independent of and potentially combines well with aligned lightweight MCMC.
Another MCMC optimization which potentially combines well with alignment is due to Nori et al.~\cite{nori2014r2}.
They use static analysis to propagate observations backwards in programs to improve inference.

Information flow analyses~\cite{sabelfeld2003language} may determine if particular parts of a program execute as a result of different program inputs.
Specifically, if program input is random, such approaches have clear similarities to the alignment analysis.

Many other PPLs exist, such as Gen~\cite{towner2019gen}, Venture~\cite{mansinghka2018probabilistic}, Edward~\cite{tran2017deep}, Stan~\cite{carpenter2017stan}, and AugurV2~\cite{huang2017compiling}.
Gen, Venture, and Edward focus on simplifying the joint specification of a model and its inference to give users low-level control, and do not consider automatic alignment specifically.
However, the incremental inference approach~\cite{towner2018incremental} in Gen does use the addressing approach by Wingate et al.~\cite{wingate2011lightweight}.
Stan and AugurV2 have less expressive modeling languages to allow more powerful inference.
Alignment is by construction due to the reduced expressiveness.

Borgström et al.~\cite{borgstrom2016lambda}, Staton et al.~\cite{staton2016semantics}, \'{S}cibior et al.~\cite{scibior2017denotational}, and {V{\'a}k{\'a}r} et al.~\cite{vakar2019domain} treat semantics and correctness for PPLs, but do not consider alignment.

\section{Conclusion}\label{sec:conclusion}
This paper gives, for the first time, a formal definition of alignment in PPLs.
Furthermore, we introduce a static analysis technique and use it to align checkpoints in PPLs and apply it to SMC and MCMC inference.
We formalize the alignment analysis, prove its correctness, and implement it in Miking CorePPL.
We also implement aligned SMC and aligned lightweight MCMC, and evaluate the implementations on non-trivial CRBD and ClaDS models from phylogenetics, the LDA topic model, and a state-space model, demonstrating significant improvements compared to standard SMC and lightweight MCMC.

\subsubsection{Acknowledgments}
We thank Lawrence Murray, Johannes Borgström, and Jan Kudlicka for early discussions on the alignment idea, and Viktor Senderov for implementing ClaDS in Miking CorePPL.
We also thank the anonymous reviewers at ESOP for their valuable comments.

\clearpage
%
%
%
\bibliographystyle{splncs04} 
\bibliography{references}
%

\ifextended
\clearpage
\appendix

\section{Evaluation, Continued}\label{sec:evalcont}
This section presents further details related to the evaluation in Section~\ref{sec:eval}.
In particular, we attach code listings for the experiment models.
Note that these listings only give the model code.
The code for the analysis itself and all inference algorithms are available on GitHub~\cite{mikingdpplgithub}.

\subsection{SMC: Constant Rate Birth-Death (CRBD)}\label{sec:crbdcont}
Listing~\ref{lst:crbd} gives the Miking CorePPL source code used for the case study model in Section~\ref{sec:crbd}.
\begin{lstlisting}[language=CorePPL,caption=The source code for the experiment in Sections~\ref{sec:crbd} and~\ref{sec:crbdmcmc},label=lst:crbd]
------------------------------------------------
-- The Constant-Rate Birth-Death (CRBD) model --
------------------------------------------------

-- The prelude includes a few PPL helper functions
include "pplprelude.mc"

-- The tree.mc file defines the general tree structure
include "tree.mc"

-- The tree-instance.mc file includes the actual tree and the rho constant
include "tree-instance.mc"

mexpr

-- CRBD goes undetected, including iterations. Mutually recursive functions.
recursive
  let iter: Int -> Float -> Float -> Float -> Float -> Float -> Bool =
    lam n: Int.
    lam startTime: Float.
    lam branchLength: Float.
    lam lambda: Float.
    lam mu: Float.
    lam rho: Float.
      if eqi n 0 then
        true
      else
        let eventTime = assume (Uniform (subf startTime branchLength) startTime) in
        if crbdGoesUndetected eventTime lambda mu rho then
          iter (subi n 1) startTime branchLength lambda mu rho
        else
          false

  let crbdGoesUndetected: Float -> Float -> Float -> Float -> Bool =
    lam startTime: Float.
    lam lambda: Float.
    lam mu: Float.
    lam rho: Float.
      let duration = assume (Exponential mu) in
      let cond =
        -- `and` does not use short-circuiting: using `if` as below is more
        -- efficient
        if (gtf duration startTime) then
          (eqBool (assume (Bernoulli rho)) true)
        else false
      in
      if cond then
        false
      else
        let branchLength = if ltf duration startTime then duration else startTime in
        let n = assume (Poisson (mulf lambda branchLength)) in
        iter n startTime branchLength lambda mu rho
in

-- Simulation of branch
recursive
let simBranch: Int -> Float -> Float -> Float -> Float -> Float -> () =
  lam n: Int.
  lam startTime: Float.
  lam stopTime: Float.
  lam lambda: Float.
  lam mu: Float.
  lam rho: Float.
    if eqi n 0 then ()
    else
      let currentTime = assume (Uniform stopTime startTime) in
      if crbdGoesUndetected currentTime lambda mu rho then
        let w1 = weight (log 2.) in
        simBranch (subi n 1) startTime stopTime lambda mu rho
      else
        let w2 = weight (negf inf) in
        ()
in

-- Simulating along the tree structure
recursive
let simTree: Tree -> Tree -> Float -> Float -> Float -> () =
  lam tree: Tree.
  lam parent: Tree.
  lam lambda: Float.
  lam mu: Float.
  lam rho: Float.
    let lnProb1 = mulf (negf mu) (subf (getAge parent) (getAge tree)) in
    let lnProb2 = match tree with Node _ then log lambda else log rho in

    let startTime = getAge parent in
    let stopTime = getAge tree in
    let n = assume (Poisson (mulf lambda (subf startTime stopTime))) in
    simBranch n startTime stopTime lambda mu rho;

    let w3 = weight (addf lnProb1 lnProb2) in

    match tree with Node { left = left, right = right } then
      simTree left tree lambda mu rho;
      simTree right tree lambda mu rho
    else ()
in

-- Fixed priors used for the SMC experiment
-- let lambda = 0.2 in
-- let mu = 0.1 in

--
let lambda = assume (Gamma 1.0 1.0) in
let mu = assume (Gamma 1.0 0.5) in

-- Adjust for normalizing constant
let numLeaves = countLeaves tree in
let corrFactor =
  subf (mulf (subf (int2float numLeaves) 1.) (log 2.)) (lnFactorial numLeaves) in
weight corrFactor;

-- Start of the simulation along the two branches
(match tree with Node { left = left, right = right } then
  simTree left tree lambda mu rho;
  simTree right tree lambda mu rho
else ());

-- Compute the joint posterior over lambda and mu ...
(lambda,mu)

-- ... or the marginal posterior over lambda (MCMC experiment) ...
-- lambda

-- ... or nothing for just estimating the normalizing constant (SMC experiment)
-- ()

\end{lstlisting}

\subsection{SMC: Cladogenetic Diversification Rate Shift (ClaDS)}\label{sec:cladscont}
Listing~\ref{lst:clads} gives the Miking CorePPL source code used for the case study model in Section~\ref{sec:clads}.
\begin{lstlisting}[language=CorePPL,caption=The source code for the experiment in Section~\ref{sec:clads},label=lst:clads]
------------------------------------------------------------------
-- The ClaDogenetic Diversification Shifts model (ClaDS2) model --
------------------------------------------------------------------

-- The prelude includes a few PPL helper functions
include "pplprelude.mc"

-- The tree.mc file defines the general tree structure
include "tree.mc"

-- The tree-instance.mc file includes the actual tree and the rho constant
include "tree-instance.mc"

mexpr

-- Multiplier guards
let maxM = 10e5 in
let minM = 0. in

-- Clads2 goes undetected.
recursive
let clads2GoesUndetected: Float -> Float -> Float -> Float
                            -> Float -> Float -> Float -> Bool =
  lam startTime_Mya: Float.
  lam lambda0: Float.
  lam mu0: Float.
  lam m: Float. -- Multiplier
  lam logAlpha: Float. -- Logarithm of alpha
  lam sigma: Float. -- Standard deviation
  lam rho: Float.

    -- Guard: m is not allowed to exceed maxM or be 0.
    if or (gtf m maxM) (leqf m minM) then false
    else
      let eventTime_My =
        assume (Exponential (addf (mulf m lambda0) (mulf m mu0))) in
      let currentTime_Mya = subf startTime_Mya eventTime_My in
      if ltf currentTime_Mya 0. then
        if assume (Bernoulli rho) then false
        else true
      else
	let extinction =
          assume (Bernoulli (divf (mulf m mu0)
                    (addf (mulf m lambda0) (mulf m mu0)))) in
	if extinction then true
	else
          let m1 = mulf m (exp (assume (Gaussian logAlpha sigma))) in
          let m2 = mulf m (exp (assume (Gaussian logAlpha sigma))) in
          if clads2GoesUndetected currentTime_Mya
               lambda0 mu0 m1 logAlpha sigma rho then
            clads2GoesUndetected currentTime_Mya
              lambda0 mu0 m2 logAlpha sigma rho
          else false
in

-- Simulation of branch
recursive
let simBranch: Float -> Float -> Float -> Float
                 -> Float -> Float -> Float -> Float -> Float =
  lam startTime_Mya: Float.
  lam stopTime_Mya: Float.
  lam lambda0: Float.
  lam mu0: Float.
  lam m: Float. -- multiplier
  lam logAlpha: Float.
  lam sigma: Float.
  lam rho: Float.

    -- Guard: m is not allowed to exceed maxM or be 0.
    if or (gtf m maxM) (ltf m minM) then
       let w0 = weight (negf inf) in
       m
    else
      let tSpeciation_My = assume (Exponential (mulf m lambda0)) in
      let currentTime_Mya = subf startTime_Mya tSpeciation_My in
      let branchLength_My = subf startTime_Mya stopTime_Mya in
      if (ltf currentTime_Mya stopTime_Mya) then
        let w1 = weight (mulf (negf (mulf m mu0)) branchLength_My) in
        m
      else
        let m1 = mulf m (exp (assume (Gaussian logAlpha sigma))) in
        if clads2GoesUndetected currentTime_Mya lambda0
             mu0 m1 logAlpha sigma rho then
          let m2 = mulf m (exp (assume (Gaussian logAlpha sigma))) in
          let w2 = weight (log 2.) in
  	let w3 = weight (mulf (negf (mulf m mu0)) tSpeciation_My) in
          simBranch currentTime_Mya stopTime_Mya
            lambda0 mu0 m2 logAlpha sigma rho
        else -- side branch detected
          let w4 = weight (negf inf) in
          m
in

-- Simulating along the tree structure
recursive
let simTree: Tree -> Tree -> Float -> Float
               -> Float -> Float -> Float -> Float -> () =
  lam tree: Tree.
  lam parent: Tree.
  lam lambda0: Float.
  lam mu0: Float.
  lam m: Float.
  lam logAlpha: Float.
  lam sigma: Float.
  lam rho: Float.

    let startTime_Mya = getAge parent in
    let stopTime_Mya = getAge tree in

    let mEnd =
      simBranch startTime_Mya stopTime_Mya lambda0 mu0 m logAlpha sigma rho in
    (match tree with Node _
     then weight (log (mulf mEnd lambda0)) else weight (log rho));

    let m1 = mulf mEnd (exp (assume (Gaussian logAlpha sigma))) in
    let m2 = mulf mEnd (exp (assume (Gaussian logAlpha sigma))) in
    match tree with Node { left = left, right = right } then
      simTree left tree lambda0 mu0 m1 logAlpha sigma rho;
      simTree right tree lambda0 mu0 m2 logAlpha sigma rho
    else ()
in

-- Priors
let lambda0 = 0.2 in
let mu0 = 0.1 in
let logAlpha = negf 0.3 in
let sigma = sqrt 0.1 in
let m = 1.0 in

-- Adjust for normalizing constant
let numLeaves = countLeaves tree in
let corrFactor =
  subf (mulf (subf (int2float numLeaves) 1.) (log 2.)) (lnFactorial numLeaves) in
weight corrFactor;

let m1 = mulf m (exp (assume (Gaussian logAlpha sigma))) in
let m2 = mulf m (exp (assume (Gaussian logAlpha sigma))) in

-- Start of the simulation along the two branches
(match tree with Node { left = left, right = right } then
   simTree left tree lambda0 mu0 m1 logAlpha sigma rho;
   simTree right tree lambda0 mu0 m2 logAlpha sigma rho
 else ());

-- Returns nothing, as the current model is only used to compute the
-- normalizing constant
()
\end{lstlisting}

\subsection{SMC: State-Space Aircraft Localization}\label{sec:ssmcont}
Fig.~\ref{fig:ssmexample} presents the aircraft model used for the experiment in Section~\ref{sec:ssm}.
An aircraft flies along a one-dimensional axis in discrete time steps, and the crew needs to estimate the aircraft's current position using noisy satellite position data available for the ten most recent time steps (defined at line~\ref{line:databegin}).
A second model component---the aircraft's altitude---further complicates the model as the crew \emph{cannot} observe it (the altimeter is not functioning).
The aircraft's velocity and the precision of the satellite observations depend on the altitude, as dictated by the functions $\mi{velocity}$ (defined at line~\ref{line:velocity}) and $\mi{positionObsStDev}$ (defined at line~\ref{line:obsdev}).
The velocity (in meters per second) increases linearly with increasing altitude (less air resistance) but is capped to the range $[100,500]$.
On the other hand, the observation standard deviation (in meters) decreases linearly with increasing altitude (less interference between the satellites and the aircraft) but is never less than ten.

Lines~\ref{line:simbegin} to~\ref{line:simend} define the main function $\mi{simulate}$ iterating over the ten data items.
The critical component illustrating the need for alignment is the \ttt{weight $0.5$} at line~\ref{line:weightair}.
This \ttt{weight} encodes that the pilot adjusts the aircraft's pitch when air traffic control signals altitude deviations more than 100 feet from the assigned altitude of $35\,000$ feet.
Each time step where the actual altitude deviates more than 100 feet from the assigned altitude thus gives a penalty factor of $0.5$.
Unlike the \ttt{weight} at line~\ref{line:weightalignedair}, this \ttt{weight} is unaligned.

The simulation also accounts for variations in, e.g., wind resistance when updating the position at line~\ref{line:posupdate} through a standard deviation of $\mi{positionStDev}$ meters.
Similarly, the altitude varies with a standard deviation of $\mi{altitudeStDev}$ feet when updating the altitude at line~\ref{line:altupdate}.

We generated the ten data points used for the experiment in Section~\ref{sec:ssm} by running the model (ignoring line~\ref{line:weightair}) and sampling from $\mathcal N (\mi{position},\sigma^2)$ at line~\ref{line:weightalignedair}.

\begin{figure}[tbp]
  \lstset{%
    basicstyle=\ttfamily\scriptsize,
    columns=fullflexible,
    numbers=left, showlines=true,
    numbersep=3pt,
    framexleftmargin=-2pt,
    xleftmargin=2em,
    language=calc
  }
  \centering
  \begin{multicols}{2}
    \begin{lstlisting}[name=ssm]
let $\mi{data}$ = [$\label{line:databegin}$
  $603.57$, $860.42$, $1012.07$, $1163.53$,
  $1540.29$, $1818.10$, $2045.38$, $2363.49$,
  $2590.77$, $2801.91$
]$\label{line:dataend}$
let $\mi{holdingAltitude}$ = $35\,000$ in
let $\mi{altitudeRange}$ = $100$ in
let $\mi{position}$ = assume $\textrm{Uniform}(0, 1000)$ in
let $\mi{altitude}$ =
  assume $\mathcal N (\mi{holdingAltitude}, 200^2)$ in
let $\mi{positionStDev}$ = $50$ in
let $\mi{baseVelocity}$ = $250$ in
let $\mi{velocity}$ = $\lambda\mi{altitude}$.$\label{line:velocity}$
  let $k$ = $\frac{\mi{baseVelocity}}{\mi{holdingAltitude}}$ in
  $\min{(500, \max{(100, (k \cdot \mi{altitude}))})}$
in
let $\mi{basePositionObsStDev}$ = $50$ in
let $\mi{positionObsStDev:}$ = $\lambda\mi{altitude}$.$\label{line:obsdev}$
  let $m$ = $100$ in
  let $k$ = $-\frac{\mi{basePositionObsStDev}}{\mi{holdingAltitude}}$ in
  $\max{(10, m + k \cdot\mi{altitude})}$
in
let $\mi{altitudeStDev}$ = $100$ in
let rec $\mi{simulate}$ =
  $\lambda \mi{data}$. $\lambda \mi{position}$. $\lambda \mi{altitude}$.$\label{line:simbegin}$
    match $\mi{data}$ with $d::\mi{ds}$ then
      let $\sigma$ =
        $\mi{positionObsStDev} \hspace{1pt} \mi{altitude}$ in
      weight $f_{\mathcal{N}(\mi{position}, \sigma^2)}(d)$$\label{line:weightalignedair}$
      if $|\mi{altitude} - \mi{holdingAltitude}|$
           $> \mi{altitudeRange}$ then
        weight $0.5$$\label{line:weightair}$
      else $()$;
      let $\mi{position}$ =$\label{line:posupdate}$
        assume $\mathcal N ($
          $\mi{position} + \mi{velocity} \hspace{1pt} \mi{altitude},$
          $\mi{positionStDev}^2$
        $)$ in
      let $\mi{altitude}$ =$\label{line:altupdate}$
        assume $\mathcal N (\mi{altitude}, \mi{altitudeStDev}^2)$
      in
      $\mi{simulate}$ $\mi{ds}$ $\mi{position}$ $\mi{altitude}$
    else $\mi{position}$
in$\label{line:simend}$
$\mi{simulate}$ $\mi{data}$ $\mi{position}$ $\mi{altitude}$
    \end{lstlisting}
  \end{multicols}
  \caption{
    A state-space model for estimating an aircraft's position given a set of noisy position estimates.
    The text contains further details.
    The program uses the syntax~\eqref{eq:ast}, extended with sequences, pattern matching over sequences, and the pattern $::$ for sequence deconstruction.
    The function $f_{\mathcal{N}(\mu, \sigma^2)}$ is the PDF of the normal distribution at $\mu$ with variance $\sigma^2$.
  }
  \label{fig:ssmexample}
\end{figure}

Listing~\ref{lst:ssm} gives the Miking CorePPL source code used for the case study model in Section~\ref{sec:ssm}.
\begin{lstlisting}[language=CorePPL,caption=The source code for the experiment in Section~\ref{sec:ssm},label=lst:ssm,float]
---------------------------------------------------
-- A state-space model for aircraft localization --
---------------------------------------------------

include "math.mc"

mexpr

-- Noisy satellite observations of position (accuracy is improved at higher
-- altitude)
let ysPos: [Float] = [
  603.5736741666899, 860.4207338929477, 1012.0766100484578, 1163.5339974878366,
  1540.2972028551385, 1818.1023092741882, 2045.3888580253108,
  2363.4902615131796, 2590.773153142429, 2801.9143537470927
] in

let holdingAltitude = 35000. in
let altitudeRange = 100. in
let position: Float = assume (Uniform 0. 1000.) in
let altitude: Float = assume (Gaussian holdingAltitude 200.) in

let positionStDev = 50. in

let baseVelocity = 250. in
let velocity: Float -> Float = lam altitude.
  let k = divf baseVelocity holdingAltitude in
  minf 500. (maxf 100. (mulf k altitude))
in

let basePositionObsStDev = 50. in
let positionObsStDev: Float -> Float = lam altitude.
  let m = 100. in
  let k = negf (divf basePositionObsStDev holdingAltitude) in
  maxf 10. (addf m (mulf k altitude))
in

let altitudeStDev = 100. in

recursive let simulate: Int -> Float -> Float -> Float =
  lam t: Int. lam position: Float. lam altitude: Float.

  -- Observe position
  let dataPos: Float = get ysPos t in
  observe dataPos (Gaussian position (positionObsStDev altitude));
  let t = addi t 1 in

  -- Penalize altitude divergence of more than `altitudeRange` feet from
  -- holding altitude
  (if gtf (absf (subf altitude holdingAltitude)) altitudeRange then
     weight (log 0.5)
   else ());

  -- Transition
  let position: Float =
    assume (Gaussian (addf position (velocity altitude)) positionStDev) in
  let altitude: Float = assume (Gaussian altitude altitudeStDev) in

  if eqi (length ysPos) t then position
  else simulate t position altitude
in

simulate 0 position altitude

\end{lstlisting}

\subsection{MCMC: Latent Dirichlet Allocation (LDA)}\label{sec:ldacont}
Listing~\ref{lst:lda} gives the Miking CorePPL source code used for the case study model in Section~\ref{sec:lda}.
Furthermore, we conduct an additional LDA experiment justifying the correctness of the aligned lightweight MCMC and lightweight MCMC implementations.
The experiment uses a simplified generated data set with only two topics, a vocabulary of two words, and three documents with 10 words each.
To generate the data, we use the true values $\theta_1 = 0.95$, $\theta_2 = 0.05$, and $\theta_3 = 0.5$ for the document topic distributions, and $\phi_1 = 0.99$ and $\phi_2 = 0.01$ for the word distribution within the two topics.
Note that the true proportions above are uniquely determined by the proportion of the first topic and first word, as there are only two topics and two words in the vocabulary.
The simplicity of the model and rather extreme true values used to generate the data allows for easy visualization of the document topic posteriors and justification of their correctness.
Fig.~\ref{fig:ldacont} presents the posterior topic distributions for the three documents for a very large number of MCMC iterations.
As expected, aligned lightweight MCMC and lightweight MCMC produce identical results agreeing with the true values for $\theta_1$, $\theta_2$, and $\theta_3$.
The bimodal posteriors for $\theta_1$ and $\theta_2$ are due to the interchangeability of topics in LDA.

\begin{figure}[tb]
  \centering
    \pgfplotsset{
      ymax=7.5,
      ymin=0,
      xmin=0,
      xmax=1,
      width=0.4\textwidth,
      height=3cm,
      xtick={0,0.5,1},
      title style={yshift=-8mm}
    }
    \centering
    \begin{subfigure}{\columnwidth}
      \centering
      \begin{tikzpicture}[trim axis left]
        \scriptsize
        \begin{axis}[
          exphist,
          title=$\theta_1$,
          ]
          \exphist{mcmc/lda}{align-1}{\alignedcolor}
        \end{axis}
      \end{tikzpicture}
      \begin{tikzpicture}
        \scriptsize
        \begin{axis}[
          exphist,
          title=$\theta_2$,
          ]
          \exphist{mcmc/lda}{align-2}{\alignedcolor}
        \end{axis}
      \end{tikzpicture}
      \begin{tikzpicture}[trim axis right]
        \scriptsize
        \begin{axis}[
          exphist,
          title=$\theta_3$,
          ]
          \exphist{mcmc/lda}{align-3}{\alignedcolor}
        \end{axis}
      \end{tikzpicture}
      \caption{Aligned lightweight MCMC posteriors.}
    \end{subfigure}
    \begin{subfigure}{\columnwidth}
      \centering
      \begin{tikzpicture}[trim axis left]
        \scriptsize
        \begin{axis}[
          exphist,
          title=$\theta_1$,
          ]
          \exphist{mcmc/lda}{lightweight-1}{\unalignedcolor}
        \end{axis}
      \end{tikzpicture}
      \begin{tikzpicture}
        \scriptsize
        \begin{axis}[
          exphist,
          title=$\theta_2$,
          ]
          \exphist{mcmc/lda}{lightweight-2}{\unalignedcolor}
        \end{axis}
      \end{tikzpicture}
      \begin{tikzpicture}[trim axis right]
        \scriptsize
        \begin{axis}[
          exphist,
          title=$\theta_3$,
          ]
          \exphist{mcmc/lda}{lightweight-3}{\unalignedcolor}
        \end{axis}
      \end{tikzpicture}
      \caption{Lightweight MCMC posteriors.}
    \end{subfigure}
  \hfill
  \caption{%
    Fig. (a) and (b) plots aligned lightweight MCMC and lightweight MCMC posterior distributions for the three documents $\theta_1$, $\theta_2$, and $\theta_3$ in the simplified LDA data set in Section~\ref{sec:ldacont}.
    The posteriors are the combined samples of 300 independent MCMC runs, each with $3\cdot10^6$ iterations and $10\%$ burn.
  }
  \label{fig:ldacont}
\end{figure}

\begin{lstlisting}[language=CorePPL,caption=The source code for the experiment in Section~\ref{sec:lda},label=lst:lda,float]
include "common.mc"
include "string.mc"
include "seq.mc"
include "ext/dist-ext.mc"

-- The data.mc file contains the generated data
include "data.mc"

mexpr

let alpha: [Float] = make numtopics 1. in
let beta: [Float] = make vocabsize 1. in
let phi = create numtopics (lam. assume (Dirichlet beta)) in
let theta = create numdocs (lam. assume (Dirichlet alpha)) in
repeati (lam w.
    let word = get docs w in
    let counts = assume (Multinomial word.1 (get theta (get docids w))) in
    iteri (lam z. lam e.
        weight (mulf (int2float e)
                  (bernoulliLogPmf (get (get phi z) word.0) true))
      ) counts
  ) (length docs);

-- Returns the joint posterior distribution over theta and phi ...
(theta, phi)

-- ... or the marginal posterior over theta (simple data set) ...
-- theta

-- ... or nothing if only comparing exeuction time (C3 data set)
-- ()
\end{lstlisting}

\subsection{MCMC: Constant Rate Birth-Death (CRBD)}\label{sec:crbdmcmccont}
Listing~\ref{lst:crbd} gives the Miking CorePPL source code used for the case study model in Section~\ref{sec:crbdmcmc}.
Furthermore, Fig~\ref{fig:crbdplot} shows the posterior distributions over $\mi{lambda}$, justifying the use of the mean as a measure of accuracy as the posterior is clearly unimodal.
\begin{figure}
  \pgfplotsset{
     ymax=10,
     ymin=0,
     xmin=0,
     xmax=1
   }
  \centering
  \begin{subfigure}{\columnwidth}
    \centering
    \begin{tikzpicture}
      \begin{axis}[
          exphist,
          width=0.40\textwidth
        ]
        \exphist{mcmc/crbd}{align-1}{\alignedcolor}
      \end{axis}
    \end{tikzpicture}
    \begin{tikzpicture}
      \begin{axis}[
          exphist,
          width=0.40\textwidth
        ]
        \exphist{mcmc/crbd}{align-2}{\alignedcolor}
      \end{axis}
    \end{tikzpicture}
    \begin{tikzpicture}
      \begin{axis}[
          exphist,
          width=0.40\textwidth
        ]
        \exphist{mcmc/crbd}{align-3}{\alignedcolor}
      \end{axis}
    \end{tikzpicture}
    \caption{Aligned lightweight MCMC.}
  \end{subfigure}
  \\
  \begin{subfigure}{\columnwidth}
    \centering
    \begin{tikzpicture}
      \begin{axis}[
        exphist,
        width=0.40\textwidth
        ]
        \exphist{mcmc/crbd}{lightweight-1}{\unalignedcolor}
      \end{axis}
    \end{tikzpicture}
    \begin{tikzpicture}
      \begin{axis}[
        exphist,
        width=0.40\textwidth
        ]
        \exphist{mcmc/crbd}{lightweight-2}{\unalignedcolor}
      \end{axis}
    \end{tikzpicture}
    \begin{tikzpicture}
      \begin{axis}[
        exphist,
        width=0.40\textwidth
        ]
        \exphist{mcmc/crbd}{lightweight-3}{\unalignedcolor}
      \end{axis}
    \end{tikzpicture}
    \caption{Lightweight MCMC}
  \end{subfigure}
  \caption{%
    One iteration of the CRBD experiment in Section~\ref{sec:crbdmcmc}. Fig. (a) shows posteriors for aligned lightweight MCMC (\alignedcolor).
    From left to right: $3\cdot10^4$ iterations, $3\cdot10^5$ iterations, and $3\cdot10^6$ iterations.
    Fig. (b) shows the corresponding posteriors for lightweight MCMC (\unalignedcolor).
  }
  \label{fig:crbdplot}
\end{figure}

\section{Alignment Analysis, Continued}\label{sec:aligncont}
This section presents the full alignment constraint propagation algorithm (Section~\ref{sec:analysisalg}) and proof of soundness of the alignment analysis (Section~\ref{sec:proof}).

\subsection{Algorithm}\label{sec:analysisalg}
\begin{algorithm}
  \renewcommand{\s}{\hphantom{|}}
  \caption{%
    Alignment analysis.
  }\label{alg:flow}
  \raggedright
  \lstinline[style=alg]!function $\tsc{analyzeAlign}$($\term$): $\Termanf \rightarrow ((X \rightarrow \mathcal{P}(A)) \times \mathcal P(X))$ =!\\
  \vspace{-1mm}
  \hspace{4.0mm}%
  \begin{minipage}{0.96\textwidth}
    \begin{multicols}{2}
      \begin{lstlisting}[
          style=alg,
          basicstyle=\sffamily\scriptsize,
          numbers=left,
          showlines=true,
          numbersep=3pt
        ]
worklist$: [X]$ $\coloneqq$ $[]$
data$: X \rightarrow \mathcal{P}(A)$ $\coloneqq \{(x,\varnothing) \mid x \in X\}$
unaligned$: \mathcal P(X)$ $\coloneqq \varnothing$
edges$: X \rightarrow \mathcal{P}(R)$ $\coloneqq \{(x,\varnothing) \mid x \in X\}$
for $\cstr$ $\in$ $\tsc{generateConstraints}$($\term$):
  $\tsc{initializeConstraint}(\cstr)$
$\tsc{iter}$(); $\s$ return (data, unaligned)

function $\tsc{iter}$: $() \rightarrow ()$ = match worklist with
  | $[]$ $\rightarrow ()$
  | $x$ :: worklist' $\rightarrow$
  $\s$ worklist $\coloneqq$ worklist'
  $\s$ for $\cstr$ $\in$ edges(x):
  $\s$ $\s$ $\tsc{propagateConstraint}$($\cstr$)
  $\s$ $\tsc{iter}$ $()$

function $\tsc{initializeConstraint}$($\cstr$): $R \rightarrow ()$ =
  match $\cstr$ with
  | $\absval \in S_x \rightarrow$ $\tsc{addData}$($x$, $\{\absval\}$)
  | $S_x \subseteq S_y \rightarrow$ $\tsc{initializeConstraint}'$($x$, $\cstr$)
  | $\absval_1 \in S_x \Rightarrow \absval_2 \in S_y \rightarrow$
  $\s$ $\tsc{initializeConstraint}'$($x$, $\cstr$)
  | $\forall x \forall y \s \lambda x. y \in S_\mi{lhs}$
  $\s$ $\s$ $\Rightarrow (S_\mi{rhs} \subseteq S_x) \land (S_y \subseteq S_\mi{app}) \rightarrow$
  $\s$ $\tsc{initializeConstraint}'$($\mi{lhs}$, $\cstr$)
  | $\forall n \s (\ttt{const} \s n \in S_\mi{lhs}) \land (n > 1)$
  $\s$ $\s$ $\Rightarrow \ttt{const } n-1 \in S_\mi{app} \rightarrow$
  $\s$ $\tsc{initializeConstraint}'$($\mi{lhs}$, $\cstr$)
  | $\ttt{const} \s \_ \in S_{\mi{lhs}}$
  $\s$ $\s$ $\Rightarrow (\ttt{stoch} \in \mi{rhs} \Rightarrow \ttt{stoch} \in \mi{app}) \rightarrow$
  $\s$ $\tsc{initializeConstraint}'$($\mi{lhs}$, $\cstr$)
  | $\unaligned_x \Rightarrow \unaligned_y \rightarrow$
  $\s$ $\tsc{initializeConstraint}'$($x$, $\cstr$)
  | $\ttt{stoch} \in S_x \Rightarrow \unaligned_y \rightarrow$
  $\s$ $\tsc{initializeConstraint}'$($x$, $\cstr$)
  | $\forall x \s \lambda x.\_ \in S_\mi{lhs} \Rightarrow \unaligned_x \rightarrow$
  $\s$ $\tsc{initializeConstraint}'$($\mi{lhs}$, $\cstr$)
  | $\unaligned_\mi{res} \Rightarrow$
  $\s$ $\s$ $(\forall x \s \lambda x. \_ \in S_{\mi{lhs}} \Rightarrow \unaligned_x) \rightarrow$
  $\s$ $\tsc{initializeConstraint}'$($\mi{res}$, $\cstr$)
  | $\ttt{stoch} \in S_{\mi{lhs}} \Rightarrow$
  $\s$ $\s$ $ (\forall x \s \lambda x. \_ \in S_{\mi{lhs}} \Rightarrow \unaligned_x) \rightarrow$
  $\s$ $\tsc{initializeConstraint}'$($\mi{lhs}$, $\cstr$)

function $\tsc{initializeConstraint}'$($x$,$\cstr$)
    : $X \rightarrow ()$ =
  edges($x$) $\coloneqq$ edges($x$) $\cup \s \{\cstr\}$;
  $\tsc{propagateConstraint}$($\cstr$)
(*@ \columnbreak @*)
function $\tsc{addData}$($x$, A): $X \times \mathcal P(A) \rightarrow ()$ =
  if A $\not\subseteq$ data($x$) then
    data($x$) $\coloneqq$ data($x$) $\cup \s A$
    worklist $\coloneqq$ $x$ $::$ worklist

function $\tsc{addUnaligned}$($x$): $X \rightarrow ()$ =
  if $x \not \in$ unaligned then
    unaligned $\coloneqq$ unaligned $\cup \{ x\}$
    worklist $\coloneqq$ $x$ $::$ worklist

function $\tsc{propagateConstraint}$($\cstr$): $R \rightarrow ()$ =
  match c with
  | $\absval \in S_x \rightarrow ()$
  | $S_x \subseteq S_y \rightarrow$ $\tsc{addData}$($y$, data($x$))
  | $\absval_1 \in S_x \Rightarrow \absval_2 \in S_y \rightarrow$
  $\s$ if $\absval_1 \in$ data($x$) then $\tsc{addData}$($y$,$\{\absval_2\}$)
  | $\forall x \forall y \s \lambda x. y \in S_\mi{lhs}$
  $\s$ $\s$ $\Rightarrow (S_\mi{rhs} \subseteq S_x) \land (S_y \subseteq S_\mi{app}) \rightarrow$ $\label{line:proplambda}$
  $\s$ for $\lambda x. y \in$ data($\mi{lhs}$):
  $\s$ $\s$ $\tsc{initializeConstraint}$($S_\mi{rhs} \subseteq S_x$)
  $\s$ $\s$ $\tsc{initializeConstraint}$($S_y \subseteq S_\mi{app}$)
  | $\forall n \s (\ttt{const} \s n \in S_\mi{lhs}) \land (n > 1)$
  $\s$ $\s$ $\Rightarrow \ttt{const} \s n-1 \in S_\mi{app} \rightarrow$
  $\s$ for $\ttt{const} \s n \in$ data($\mi{lhs}$):
  $\s$ $\s$ if $n > 1$ then
  $\s$ $\s$ $\s$ $\tsc{addData}$($\mi{app}$, $\{\ttt{const} \s n-1\}$)
  | $\ttt{const} \s \_ \in S_{\mi{lhs}}$
  $\s$ $\s$ $\Rightarrow (\ttt{stoch} \in \mi{rhs} \Rightarrow \ttt{stoch} \in \mi{app}) \rightarrow$
  $\s$ if $\exists n \s \ttt{const } n \in S_\mi{lhs}$ then
  $\s$ $\s$ $\tsc{initializeConstraint}$(
  $\s$ $\s$ $\s$ $\ttt{stoch} \in \mi{rhs} \Rightarrow \ttt{stoch} \in \mi{app}$
  $\s$ $\s$ )
  | $\unaligned_x \Rightarrow \unaligned_y \rightarrow$
  $\s$ if $x \in$ unaligned then $\tsc{addUnaligned}$($y$)
  | $\ttt{stoch} \in S_x \Rightarrow \unaligned_y \rightarrow$
  $\s$ if $\ttt{stoch} \in$ data($x$) then $\tsc{addUnaligned}$($y$)
  | $\forall x \s \lambda x.\_ \in S_\mi{lhs} \Rightarrow \unaligned_x \rightarrow$
  $\s$ for $\lambda x. \_ \in$ data($\mi{lhs}$): $\tsc{addUnaligned}$($x$)
  | $\unaligned_\mi{res} \Rightarrow$
  $\s$ $\s$ $ (\forall x \s \lambda x. \_ \in S_{\mi{lhs}} \Rightarrow \unaligned_x) \rightarrow$
  $\s$ if $\mi{res} \in$ unaligned then
  $\s$ $\s$ $\tsc{initializeConstraint}$(
  $\s$ $\s$ $\s$ $\forall x \s \lambda x. \_ \in S_{\mi{lhs}} \Rightarrow \unaligned_x$
  $\s$ $\s$ )
  | $\ttt{stoch} \in S_{\mi{lhs}} \Rightarrow$
  $\s$ $\s$ $ (\forall x \s \lambda x. \_ \in S_{\mi{lhs}} \Rightarrow \unaligned_x) \rightarrow$
  $\s$ if $\ttt{stoch} \in$ data($\mi{lhs}$) then
  $\s$ $\s$ $\tsc{initializeConstraint}$(
  $\s$ $\s$ $\s$ $\forall x \s \lambda x. \_ \in S_{\mi{lhs}} \Rightarrow \unaligned_x$
  $\s$ $\s$ )
      \end{lstlisting}
    \end{multicols}
  \end{minipage}
\end{algorithm}
Algorithm~\ref{alg:flow} presents the full alignment algorithm that produces a solution to the constraints generated by Algorithm~\ref{alg:gencstr}.
For reference, we now also give a more formal definition of constraints $\cstr$.
\begin{definition}[Constraints]
  \begin{equation}\label{eq:cstr}
    \hspace{-2mm}
    \begin{aligned}
      &\begin{aligned}
        \cstr \Coloneqq& \s
        \absval \in S_x
        \s | \s
        S_x \subseteq S_y
        \s | \s
        \absval \in S_x \Rightarrow \absval \in S_y
        \\ |& \s
        \forall x \forall y \s \lambda x. y \in S_{\mi{lhs}} \Rightarrow (S_\mi{rhs} \subseteq S_x) \land (S_y \subseteq S_\mi{app})
        \\ |& \s
        \forall n \s (\ttt{const} \s n \in S_\mi{lhs}) \land (n > 1) \Rightarrow \ttt{const} \s n-1 \in S_\mi{app}
        \\ |& \s
        \ttt{const} \s \_ \in S_\mi{lhs} \Rightarrow (\ttt{stoch} \in S_\mi{rhs} \Rightarrow \ttt{stoch} \in S_\mi{app})
        \\ |& \s
        \unaligned_x \Rightarrow \unaligned_y
        \s | \s
        %
        \ttt{stoch} \in S_x \Rightarrow \unaligned_y
        \\ |& \s
        \forall x \s \lambda x. \_ \in S_{\mi{lhs}} \Rightarrow \unaligned_x
        \\ |& \s
        \unaligned_\mi{res} \Rightarrow (\forall x \s \lambda x. \_ \in S_{\mi{lhs}} \Rightarrow \unaligned_x)
        \\ |& \s
        \ttt{stoch} \in S_{\mi{lhs}} \Rightarrow (\forall x \s \lambda x. \_ \in S_{\mi{lhs}} \Rightarrow \unaligned_x)
      \end{aligned} \\
      &
      x,y,\mi{lhs},\mi{rhs},\mi{app},\mi{res} \in X.
    \end{aligned}
  \end{equation}%
\end{definition}
The main function \tsc{analyzeAlign} consists of two steps: initialization and iteration.
In the initialization step, \tsc{generateConstraints} provides constraints to the \tsc{initializeConstraint} function, which initializes the maps \tsf{data} and \tsf{edges}, and the set \tsf{unaligned}.
The map \tsf{data} contains the sets of abstract values for all program variables and is initially empty.
At termination, $\tsf{data}(x)$ is a sound approximation of $S_x$ for each $x$ (Lemma~\ref{lemma:cfa}).
The map \tsf{edges} associates a set of constraints with each variable in the program.
Specifically, we must propagate the constraints associated with a variable $x$ after updating $\tsf{data}(x)$ with new information.
Finally, the set \tsf{unaligned} tracks unaligned expressions and is initially empty.
At termination, \tsf{unaligned} contains the set of all unaligned variables identified by the analysis.
This set is sound according to Lemma~\ref{lemma:cfa}.

The iteration step \tsc{iter} propagates constraints with \tsc{propagateConstraint} for all variables updated with new abstract values or unalignment since their last propagation.
We store these updated variables in the sequence \tsf{worklist}, which, when empty, signals fixpoint and termination.
Note that, e.g., the lambda application constraint at line~\ref{line:proplambda} initializes new constraints \emph{dynamically} during propagation, depending on which abstract lambdas flow to the left-hand side of the application.

\subsection{Correctness Proof}\label{sec:proof}
This section presents the correctness proof that is ultimately used to prove Theorem~\ref{thm:main}.

Throughout this section, $\term_1 = \term_2$ means that the terms $\term_1$ and $\term_2$ are alpha equivalent.
For constant comparisons $c_1 = c_2$, we assume the prior existence of an equality function over constants.
We first require a specific equality relation on values.
\newcommand\eqv{\stackrel{V}{=}}
\begin{definition}[Value equality]
  $\termv_1 \eqv \termv_2$ iff
  \begin{itemize}
    \item $\termv_1 = \langle\lambda x. \term_1,\rho_1\rangle$, $\termv_2 = \langle\lambda x. \term_2,\rho_2\rangle$, and $\term_1 = \term_2$, or
    \item $\termv_1 = c_1$, $\termv_2 = c_2$, and $c_1 = c_2$.
  \end{itemize}
\end{definition}
Note, in particular, that $\eqv$ treats closures as equal even if their environments differ.
As we will see, this property is critical in the proof of Lemma~\ref{lemma:unaligned}.

Next, we formally define subterms.
\begin{definition}[Subterms]\label{def:subterm}
  We say that $\term'$ is a subterm of $\term$ iff
  \begin{itemize}
    \item[(1)] $\term' = \term$, or
    \item[(2)] either
      \[
        \begin{gathered}
          \term = \lambda x. \s \term_1,\quad
          \term = \term_1 \s \term_2,\quad
          \term = \ttt{let } x = \term_1 \ttt{ in } \term_2, \\
          \term = \ttt{if } \term_1 \ttt{ then } \term_2 \ttt{ else } \term_3,\\
          \term = \ttt{assume } \term_1,\quad
          \text{or} \s \term = \ttt{weight } \term_1,
        \end{gathered}
      \]
    and $\term'$ is a subterm of either $\term_1$, $\term_2$, or $\term_3$.
  \end{itemize}
\end{definition}
\noindent
In the below, we assume a
\begin{itemize}
  \item fixed $\term \in \Termanf$,
  \item an assignment to $S_x$ and $\unaligned_x$ for $x \in X$ from \tsc{analyzeAlign}$(\term)$, and
  \item $\restt = \{ x \mid \neg\unaligned_x\}$.
\end{itemize}

\newcommand\condenv[1]{\tbf{\upshape(C1#1)}}
\newcommand\resu[1]{\tbf{\upshape(R1#1)}}
\newcommand\resv[1]{\tbf{\upshape(R2#1)}}
\newcommand\resunalignedall[1]{\resu{#1}--\resv{#1}}
\newcommand\condenvtwo[1]{\tbf{\upshape(C2#1)}}
\newcommand\condenvrec[1]{\tbf{\upshape(C3#1)}}
\newcommand\condenvneq[1]{\tbf{\upshape(C4#1)}}
\newcommand\condalignedall[1]{\condenvtwo{#1}--\condenvneq{#1}}
\newcommand\resa[1]{\tbf{\upshape(R3#1)}}
\newcommand\resvrec[1]{\tbf{\upshape(R4#1)}}
\newcommand\resvneq[1]{\tbf{\upshape(R5#1)}}
\newcommand\resalignedall[1]{\resa{#1}--\resvneq{#1}}

\noindent We begin with a lemma concerning unaligned expressions in single evaluations of $\Downarrow$.
\begin{lemma}[Unaligned evaluations]\label{lemma:unaligned}
  Let
  \begin{itemize}
    \item $\term'$ be a subterm of $\term$, $\term' \in \Termanf$, and
    \item $\rho \vdash \term' \sem{l}{s}{w} \termv$
  \end{itemize}
  with $\rho$ such that, for each $x \in X$,
  \begin{description}
    \item[\condenv{}]
      $\rho(x) = \langle\lambda y. \term_y, \rho_y \rangle$ implies that
      $(\lambda y. \term_y)$ is a subterm of $\term$,
      $\lambda y. \tsc{name}(\term_y) \in S_x$, and
      that \condenv{} holds for $\rho_y$.
      Also,  $\rho(x) = c$ such that $|c| > 1$ implies $\ttt{const } |c| \in S_x$.
  \end{description}
  Then,
  \begin{description}
    \item[\resu{}] if $\unaligned_n$ for all $n \in \tsc{names}(\term')$, then $l|_\restt = []$, and
    \item[\resv{}]
      $\termv{} = \langle\lambda y. \term_y, \rho_y \rangle$ implies $(\lambda y. \term_y)$ is a subterm of $\term$ and $\lambda y. \tsc{name}(t_y) \in S_{\tsc{name}(\term')}$, and  that \condenv{} holds for $\rho_y$.
      Furthermore, $\termv{} = c$ such that $|c| > 1$ implies $\ttt{const } |c| \in  S_{\tsc{name}(\term')}$.
  \end{description}
\end{lemma}
\begin{proof}
  We proceed by structural induction over
  $\rho \vdash \term' \sem{l}{s}{w} \termv$. \\[2mm]
  \tbf{Case} $\term' = x$:\\
  The derivation is
  \[
    \frac{}
    { \rho \vdash x \sem{[]}{[]}{1} \rho(x) }
    (\textsc{Var})
  \]
  \begin{description}
    \item[\resu{}] Immediate as $l = [] = l|\restt$.
    \item[\resv{}] By definition, $\tsc{name}(\term') = x$ and $\rho(x) = \termv$. The result follows from \condenv{}.
  \end{description}
  \noindent\tbf{Case} $\term' = (\ttt{let } x = \term_1 \ttt{ in } \term_2)$:\\
  The derivation is
  \[
    \frac{ \rho \vdash \term_1 \sem{l_{1}}{s_{1}}{w_{1}} \termv' \quad \rho, x \mapsto \termv' \vdash \term_{2} \sem{l_{2}}{s_{2}}{w_{2}} \termv}
    { \rho \vdash \ttt{let } x = \term_1 \ttt{ in } \term_2 \sem{l_{1} \concat [x] \concat l_{2}}{s_{1} \concat s_{2}}{w_{1} \cdot w_{2}} \termv{} }
    (\textsc{Let})
  \]
  Note that $\unaligned_n$ for all $n \in \tsc{names}(\term')$ and the definition of $\restt$ implies $[x]|_\restt = []$. Also,
  \[
    \begin{aligned}
      l|_\restt = (l_{1} \concat [x] \concat l_{2})|_\restt =
      l_{1}|_\restt \concat [x]|_\restt \concat l_{2}|_\restt.
    \end{aligned}
  \]
  To show \resu{}, we therefore only need $l_{1}|_\restt = l_2|_\restt = []$.
  Now, let $\rho' = \rho,x \mapsto \rho(y)$.
  To apply the induction hypothesis, we must establish \condenv{} for $\rho'$, denoted \condenv{$'$}.
  To prove \condenv{$'$}, note that we only need to consider $\rho'(x)$.
  For all other $x' \in X$, $\rho'(x') = \rho(x')$ and \condenv{$'$} follows directly as a result of \condenv{}.
  We denote the induction hypothesis results \resunalignedall{} for $\rho' \vdash t_{2} \sem{l_{2}}{s_{2}}{w_{2}} \termv$ with \resunalignedall{$'$}.
  Next, we consider each case for $\term_1$ (according to $\termanf'$ in \eqref{eq:anf}, p.~\pageref{eq:anf}).
  Note that \resv{} follows directly from \resv{$'$} as $\tsc{name}(\term_2) = \tsc{name}(\term')$.
  Thus, we only need to consider \condenv{$'$} and \resu{}.
  \\[2mm]\noindent\tbf{Subcase} $\term_1 = y$\\
  The derivation for $\term_1$ is
  \[
    \frac{}
    { \rho \vdash y \sem{[]}{[]}{1} \rho(y) }
    (\textsc{Var})
  \]
  Clearly $\rho'(x) = \rho(y)$.
  Also, $S_y \subseteq S_x$ from Lemma~\ref{lemma:cfa}.
  \begin{description}
    \item[\condenv{$'$}]
      If $\rho'(x) = \langle\lambda z. \term_z,\rho_z\rangle$, then $\lambda z. \term_z$ is a subterm of $\term$ by $\rho'(x) = \rho(y)$ and \condenv{}.
      Also, clearly \condenv{} holds for $\rho_z$ by \condenv{} for $\rho$.
      Furthermore, $\lambda z. \tsc{name}(\term_z) \in S_y$ by \condenv{} and $S_y \subseteq S_x$ implies $\lambda z. \tsc{name}(\term_z) \in S_x$.
      By a similar argument, if $\rho'(x) = c$ such that $|c| > 0$, $\text{const } |c| \in S_x$.
  \end{description}
  We now apply the induction hypothesis and get \resunalignedall{$'$}.
  \begin{description}
    \item[\resu{}] We clearly have $l_1 = []$. The result now follows immediately from \resu{$'$}.
  \end{description}
  \noindent\tbf{Subcase} $\term_1 = c$\\
  The derivation for $\term_1$ is
  \[
    \frac{}
    { \rho \vdash c \sem{[]}{[]}{1} c }
    (\textsc{Const})
  \]
  Clearly, $\rho'(x) = c$.
  \begin{description}
    \item[\condenv{$'$}]
      If $|c| > 0$, we have $\ttt{const } |c| \in S_x$ as a result of Lemma~\ref{lemma:cfa}.
  \end{description}
  We now apply the induction hypothesis and get \resunalignedall{$'$}.
  \begin{description}
    \item[\resu{}] We clearly have $l_1 = []$. The result now follows immediately from \resu{$'$}.
  \end{description}
  \noindent\tbf{Subcase} $\term_1 = \lambda y.  \term_y$\\
  The derivation for $\term_1$ is
  \[
    \frac{}
    { \rho \vdash \lambda y. \term_y \sem{[]}{[]}{1} \langle\lambda y. \term_y,\rho\rangle } (\textsc{Lam})
  \]
  Clearly, $\rho'(x) = \langle\lambda y. \term_y,\rho\rangle$.
  \begin{description}
    \item[\condenv{$'$}]
      First, it is clear that $\lambda y. \term_y$ is a subterm of $\term$ and that \condenv{} holds for $\rho$. Lastly, Lemma~\ref{lemma:cfa} also gives $\lambda y. \tsc{name}(\term_y) \in S_x$.
  \end{description}
  We now apply the induction hypothesis and get \resunalignedall{$'$}.
  \begin{description}
    \item[\resu{}] We clearly have $l_1 = []$. The result now follows immediately from \resu{$'$}.
  \end{description}
  \noindent\tbf{Subcase} $\term_1 = y \s z$\\
  The possible derivations are
  \[
    \begin{gathered}
      \frac{
        \begin{gathered}
          \rho \vdash y \sem{[]}{[]}{1} \langle\lambda y'. \term_{y'},\rho_{y'}\rangle \quad
          \rho \vdash z \sem{[]}{[]}{1} \rho(z) \\[-1mm]
          \rho_{y'},y' \mapsto \rho(z) \vdash \term_{y'} \sem{l_{1}}{s_{1}}{w_{1}} \termv'
        \end{gathered}
      }
      { \rho \vdash y \s z \sem{l_{1}}{s_{1}}{w_{1}} \termv' }
      (\textsc{App}) \\
      \frac{
        \rho \vdash y \sem{[]}{[]}{1} c \quad
        \rho \vdash z \sem{[]}{[]}{1} \rho(z) \quad
        |c| > 0
      }
      { \rho \vdash y \s z \sem{[]}{[]}{1} \delta(c, \rho(z)) }
      (\textsc{Const-App})
    \end{gathered}
  \]
  \begin{description}
    \item[\condenv{$'$}]
      We first consider the case (\tsc{App}).
      Let $\rho'' = \rho_{y'}, y' \mapsto \rho(z)$ and consider the derivation $\rho'' \vdash \term_{y'} \sem{l_{1}}{s_{1}}{w_{1}} \termv'$.
      To apply the induction hypothesis, we must establish \condenv{} for $\rho''$.
      First, \condenv{} holds for $\rho_{y'}$ by \condenv{} for $\rho$.
      \condenv{} therefore also holds for $\rho''$ by a similar argument to $\rho'$ in the subcase $\term_1 = y$ above.
      From the induction hypothesis, we then get \resunalignedall{$''$}.
      From Lemma~\ref{lemma:cfa}, we have $S_{\tsc{name}(\term_{y'})} \subseteq S_x$.
      Combined with \resv{$''$}, the result follows. \\[1mm]
      Now, consider the case (\tsc{Const-App}).
      From Lemma~\ref{lemma:cfa},
      $\forall n \s \ttt{const} \s n \in S_y \land n > 1 \Rightarrow \ttt{const} \s n-1 \in S_x$.
      From Definition~\ref{def:const}, we also have $|\delta(c,\rho(z))| = |c| - 1$.
      The result now follows.
  \end{description}
  We now apply the induction hypothesis and get \resunalignedall{$'$}.
  \begin{description}
    \item[\resu{}] The (\tsc{Const-App}) case is immediate by $l_1 = []$ and \resu{$'$}.
      Therefore, assume the derivation is (\tsc{App}) and that
      $\unaligned_n$ for all $n \in \tsc{names}(\term')$ (in particular $\unaligned_x$). By Lemma~\ref{lemma:cfa}, we have $\unaligned_{y'}$ and $\unaligned_{n'}$ for $n' \in \tsc{names}(\term_{y'})$.
      By \resu{$''$}, we then have $l_1|_\restt = []$.
      From \resu{$'$}, we also have $l_2|_\restt = []$ and the result follows.
  \end{description}
  \noindent\tbf{Subcase} $\term_1 = \ttt{if } y \ttt{ then } \term_t \ttt{ else } \term_e$\\
  The possible derivations are
  \[
    \begin{gathered}
      \frac{ \rho \vdash y \sem{[]}{[]}{1} \true \quad \rho \vdash \term_t \sem{l_1}{s_1}{w_1} \termv_t }
      {\rho \vdash \ttt{if } y \ttt{ then } \term_t \ttt{ else } \term_e \sem{l_1}{s_1}{w_1} \termv_t}
      (\textsc{If-True}) \\
      \frac{ \rho \vdash y \sem{[]}{[]}{1} \false \quad \rho \vdash \term_e \sem{l_1}{s_1}{w_1} \termv_e }
      {\rho \vdash \ttt{if } y \ttt{ then } \term_t \ttt{ else } \term_e \sem{l_1}{s_1}{w_1} \termv_e}
      (\textsc{If-False}) \\
    \end{gathered}
  \]
  Without loss of generality, we only consider (\tsc{If-True}).
  Note that for the subderivation $\rho \vdash \term_t \sem{l_1}{s_1}{w_1} \termv_t$, \resunalignedall{}, denoted \resunalignedall{$_t$} below, holds immediately by the induction hypothesis as \condenv{} holds for $\rho$.
  \begin{description}
    \item[\condenv{$'$}]
      Follows from \resv{$_t$} and $\tsc{name}(\term_t) \subseteq S_x$ from Lemma~\ref{lemma:cfa}.
  \end{description}
  We now apply the induction hypothesis and get \resunalignedall{$'$}.
  \begin{description}
    \item[\resu{}]
      Assume we have $\unaligned_n$ for all $n \in \tsc{names}(\term')$ (including $\unaligned_x$).
      Then, by Lemma~\ref{lemma:cfa}, $\unaligned_{n'}$ for $n' \in \tsc{names}(\term_t)$.
      Therefore, $l_1|_\restt = []$ by \resu{$_t$}.
      From \resu{$'$}, we also have $l_2|_\restt = []$ and the result follows.
  \end{description}
  \noindent\tbf{Subcase} $\term_1 = \ttt{assume } y$\\
  The derivation is
  \[
    \frac{\rho \vdash y \sem{[]}{[]}{1} d \quad w = f_d(c) }
    {\rho \vdash \ttt{assume } y \sem{[]}{[c]}{w} c}
    (\textsc{Assume})
  \]
  \begin{description}
    \item[\condenv{$'$}]
      By Definition~\ref{def:trace}, $|c| = 0$. The result follows immediately.
  \end{description}
  We now apply the induction hypothesis and get \resunalignedall{$'$}.
  \begin{description}
    \item[\resu{}] We clearly have $l_1 = []$. The result now follows immediately from \resu{$'$}.
  \end{description}
  \noindent\tbf{Subcase} $\term_1 = \ttt{weight } y$\\
  The derivation is
  \[
    \frac{\rho \vdash y \sem{[]}{[]}{1} w}
    {\rho \vdash \ttt{weight } y \sem{[]}{[]}{w} ()}
    (\textsc{Weight})
  \]
  \begin{description}
    \item[\condenv{$'$}]
      We have $w \in \mathbb{R}$ and $|w| = 0$. The result follows immediately.
  \end{description}
  We now apply the induction hypothesis and get \resunalignedall{$'$}.
  \begin{description}
    \item[\resu{}] We clearly have $l_1 = []$. The result now follows immediately from \resu{$'$}.
  \end{description}
  \qed
\end{proof}


\noindent With Lemma~\ref{lemma:unaligned} established, we now give the main lemma used to prove Theorem~\ref{thm:main}.
\begin{lemma}[Aligned evaluations]\label{lemma:aligned}
  Let
  \begin{itemize}
    \item $\term' \in \term$, $\term' \in \Termanf$,
    \item $\rho_1 \vdash \term' \sem{l_1}{s_1}{w_1} \termv_1$, and
    \item $\rho_2 \vdash \term' \sem{l_2}{s_2}{w_2} \termv_2$
  \end{itemize}
  with $\rho_1$ and $\rho_2$ such that, for each $x \in X$,
  \begin{description}
    \item[\condenvtwo{}] for $\rho \in \{\rho_1,\rho_2\}$, \condenv{} holds,
    \item[\condenvrec{}] if $\rho_1(x) = \langle\lambda y. \term_y, \rho_1'\rangle$, $\rho_2(x) = \langle\lambda y. \term_y, \rho_2'\rangle$, and $\ttt{stoch} \not\in S_x$, then $\rho_1'$ and $\rho_2'$ fulfill \condalignedall{}, and
    \item[\condenvneq{}] If $\rho_1(x) \not\eqv \rho_2(x)$, then $\ttt{stoch} \in S_x$.
  \end{description}
  Then,
  \begin{description}
    \item[\resa{}] $l_1|_\restt = l_2|_\restt$,
    \item[\resvrec{}]
      if $\termv_1 = \langle\lambda y. \term_y, \rho_1'\rangle$, $\termv_2 = \langle\lambda y. \term_y, \rho_2'\rangle$, and $\ttt{stoch} \not\in S_{\tsc{name}(\term')}$, then $\rho_1'$ and $\rho_2'$ fulfill \condalignedall{}, and
    \item[\resvneq{}] If $\termv_1 \not\eqv \termv_2$, then $\ttt{stoch} \in S_{\tsc{name}(\term')}$.
  \end{description}

\end{lemma}
\begin{proof}
  We proceed by simultaneous structural induction over
  $\rho_1 \vdash \term' \sem{l_1}{s_1}{w_1} \termv_1$ and
  $\rho_2 \vdash \term' \sem{l_2}{s_2}{w_2} \termv_2$.\\[2mm]
  \tbf{Case} $\term' = x$:\\
  The possible derivations are
  \[
    \frac{}
    { \rho_1 \vdash x \sem{[]}{[]}{1} \rho_1(x) }
    (\textsc{Var})
    \quad
    \frac{}
    { \rho_2 \vdash x \sem{[]}{[]}{1} \rho_2(x) }
    (\textsc{Var})
  \]
  \begin{description}
    \item[\resa{}] We have $l_1 = l_2 = [] = l_1|_\restt = l_2|_\restt$.
    \item[\resvrec{}] By $\tsc{name}(\term') = x$ and \condenvrec{}.
    \item[\resvneq{}] By $\tsc{name}(\term') = x$ and \condenvneq{}.
  \end{description}
  \noindent\tbf{Case} $\term' = (\ttt{let } x = \term_1 \ttt{ in } \term_2)$:\\
  The possible derivations are
  \[
    \frac{ \rho_1 \vdash \term_1 \sem{l_{11}}{s_{11}}{w_{11}} \termv_1' \quad \rho_1, x \mapsto \termv_1' \vdash \term_{2} \sem{l_{12}}{s_{12}}{w_{12}} \termv_1}
    { \rho_1 \vdash \ttt{let } x = \term_1 \ttt{ in } \term_2 \sem{l_{11} \concat [x] \concat l_{12}}{s_{11} \concat s_{12}}{w_{11} \cdot w_{12}} \termv_1 }
    (\textsc{Let})
  \]
  \[
    \frac{ \rho_2 \vdash \term_1 \sem{l_{21}}{s_{21}}{w_{21}} \termv_2' \quad \rho_2, x \mapsto \termv_2' \vdash \term_{2} \sem{l_{22}}{s_{22}}{w_{22}} \termv_2}
    { \rho_2 \vdash \ttt{let } x = \term_1 \ttt{ in } \term_2 \sem{l_{21} \concat [x] \concat l_{22}}{s_{21} \concat s_{22}}{w_{21} \cdot w_{22}} \termv_2 }
    (\textsc{Let})
  \]
  Assume $l_{11}|_\restt = l_{21}|_\restt$ and $l_{12}|_\restt = l_{22}|_\restt$.
  Then,
  \[
    \begin{aligned}
      l_1|_\restt
      &= (l_{11} \concat [x] \concat l_{12})|_\restt \\
      &= l_{11}|_\restt \concat [x]|_\restt \concat l_{12}|_\restt
      = l_{21}|_\restt \concat [x]|_\restt \concat l_{22}|_\restt \\
      &= (l_{21} \concat [x] \concat l_{22})|_\restt
      = l_2|_\restt.
    \end{aligned}
  \]
  That is, for \resa{}, we only need $l_{11}|_\restt = l_{21}|_\restt$ and $l_{12}|_\restt = l_{22}|_\restt$.
  Let $\rho_1' = \rho_1, x \mapsto \termv'_1)$ and $\rho_2' = \rho_2, x \mapsto \termv'_2$.
  To apply the induction hypothesis, we must establish \condalignedall{} for $\rho_1'$ and $\rho_2'$.
  To avoid confusion with the original assumptions \condalignedall{} for $\rho_1$ and $\rho_2$, we use the notation \condalignedall{$'$} for the $\rho_1'$ and $\rho_2'$ conditions.
  To prove \condalignedall{$'$}, note that we only need to consider $\rho'_1(x)$ and $\rho'_2(x)$.
  For all other $x' \in X$, $\rho'_1(x') = \rho_1(x')$ and $\rho'_2(x') = \rho_2(x')$, and \condalignedall{$'$} follow directly as a result of \condalignedall{}.
  We denote the induction hypothesis results \resalignedall{} for $\rho_1' \vdash t_{2} \sem{l_{12}}{s_{12}}{w_{12}} \termv_1$ and $\rho_2' \vdash t_{2} \sem{l_{22}}{s_{22}}{w_{22}} \termv_2$ with \resalignedall{$'$}.
  Next, we consider each case for $\term_1$ (according to $\termanf'$ in \eqref{eq:anf}, p.~\pageref{eq:anf}).
  Note that \resvrec{} and \resvneq{} follow directly from \resvrec{$'$} and \resvneq{$'$} as $\tsc{name}(\term_2) = \tsc{name}(\term')$.
  Thus, we only need to consider \condalignedall{$'$} and \resa{}.
  \\[2mm]\noindent\tbf{Subcase} $\term_1 = y$\\
  The derivations for $\term_1$ are
  \[
    \frac{}
    { \rho_1 \vdash y \sem{[]}{[]}{1} \rho_1(y) }
    (\textsc{Var})
    \quad
    \frac{}
    { \rho_2 \vdash y \sem{[]}{[]}{1} \rho_2(y) }
    (\textsc{Var})
  \]
  We first establish \condalignedall{$'$}.
  Clearly $\rho_1'(x) = \rho_1(y)$ and $\rho_2'(x) = \rho_2(y)$.
  Also, $S_y \subseteq S_x$ from Lemma~\ref{lemma:cfa}.
  \begin{description}
    \item[\condenvtwo{$'$}]
      By repeating the corresponding argument for \condenv{$'$} in Lemma~\ref{lemma:unaligned} for both $\rho_1'(x)$ and $\rho_2'(x)$.
    \item[\condenvrec{$'$}]
      Assume $\rho_1'(x) = \langle\lambda z. \term_z, \rho_1''\rangle$, $\rho_2'(x) = \langle\lambda z. \term_z, \rho_2''\rangle$, and $\ttt{stoch} \not\in S_x$.
      By $S_y \subseteq S_x$, $\ttt{stoch} \not\in S_y$.
      Because $\rho_1'(x) = \rho_1(y)$ and $\rho_2'(x) = \rho_2(y)$, the result follows from \condenvrec{}.
    \item[\condenvneq{$'$}]
      If $\rho_1'(x) \not\eqv \rho_2'(x)$, then clearly $\rho_1(y) \not\eqv \rho_2(y)$.
      Hence, $\ttt{stoch} \in S_y$ by \condenvneq{} and the result follows by $S_y \subseteq S_x$.
  \end{description}
  We now apply the induction hypothesis and get \resalignedall{$'$}.
  \begin{description}
    \item[\resa{}] The result follows from $l_{11} = l_{21} = []$ and \resa{$'$}.
  \end{description}
  \noindent\tbf{Subcase} $\term_1 = c$\\
  The derivations for $\term_1$ are
  \[
      \frac{}
      { \rho_1 \vdash c \sem{[]}{[]}{1} c }
      (\textsc{Const})
      \quad
      \frac{}
      { \rho_2 \vdash c \sem{[]}{[]}{1} c }
      (\textsc{Const})
  \]
  We first establish \condalignedall{$'$}.
  \begin{description}
    \item[\condenvtwo{$'$}]
      By repeating the corresponding argument for \condenv{$'$} in Lemma~\ref{lemma:unaligned} for $\rho_1'(x) = \rho_2'(x)$.
    \item[\condenvrec{$'$}]
      Follows directly as $\rho_1'(x) = \rho_2'(x) = c$.
    \item[\condenvneq{$'$}]
      Follows directly because $\rho_1'(x) = \rho_2'(x) = c$.
  \end{description}
  We now apply the induction hypothesis and get \resalignedall{$'$}.
  \begin{description}
    \item[\resa{}] The result follows from $l_{11} = l_{21} = []$ and \resa{$'$}.
  \end{description}
  \noindent\tbf{Subcase} $\term_1 = \lambda y. \term_y$\\
  The derivations are
  \[
    \begin{gathered}
      \frac{}
      { \rho_1 \vdash \lambda y. \term_y \sem{[]}{[]}{1} \langle\lambda y. \term_y,\rho_1\rangle } (\textsc{Lam}) \\
      \frac{}
      { \rho_2 \vdash \lambda y. \term_y \sem{[]}{[]}{1} \langle\lambda y. \term_y,\rho_2\rangle }
      (\textsc{Lam})
    \end{gathered}
  \]
  We first establish \condalignedall{$'$}.
  \begin{description}
    \item[\condenvtwo{$'$}]
      By repeating the corresponding argument for \condenv{$'$} in Lemma~\ref{lemma:unaligned} for both $\rho_1'(x)$ and $\rho_2'(x)$.
    \item[\condenvrec{$'$}]
      Follows because $\rho_1$ and $\rho_2$ fulfills \condalignedall{}.
    \item[\condenvneq{$'$}]
      Follows because $\rho_1'(x) \eqv \rho_2'(x)$.
  \end{description}
  We now apply the induction hypothesis and get \resalignedall{$'$}.
  \begin{description}
    \item[\resa{}] The result follows from $l_{11} = l_{21} = []$ and \resa{$'$}.
  \end{description}
  \noindent\tbf{Subcase} $\term_1 = y \s z$\\
  The possible derivations are
  \[
    \begin{gathered}
      \frac{
        \begin{gathered}
          \rho_1 \vdash y \sem{[]}{[]}{1} \langle\lambda y_1. \term_{y_1},\rho_{y_1}\rangle \quad
          \rho_1 \vdash z \sem{[]}{[]}{1} \rho_1(z) \\[-1mm]
          \rho_{y_1}, y_1 \mapsto \rho_1(z) \vdash \term_{y_1} \sem{l_{11}}{s_{11}}{w_{11}} \termv'_{1}
        \end{gathered}
      }
      { \rho_1 \vdash y \s z \sem{l_{11}}{s_{11}}{w_{11}} \termv'_{1} }
      (\textsc{App}) \\
      \frac{
        \begin{gathered}
          \rho_2 \vdash y \sem{[]}{[]}{1} \langle\lambda y_2. \term_{y_2},\rho_{y_2}\rangle \quad
          \rho_2 \vdash z \sem{[]}{[]}{1} \rho_2(z) \\[-1mm]
          \rho_{y_2}, y_2 \mapsto \rho_2(z) \vdash \term_{y_2} \sem{l_{21}}{s_{21}}{w_{21}} \termv'_{2}
        \end{gathered}
      }
      { \rho_2 \vdash y \s z \sem{l_{21}}{s_{21}}{w_{21}} \termv'_{2} }
      (\textsc{App}) \\
      \frac{
          \rho_1 \vdash y \sem{[]}{[]}{1} c_1 \quad
          \rho_1 \vdash z \sem{[]}{[]}{1} \rho_1(z) \quad
          |c_1| > 0
      }
      { \rho_1 \vdash y \s z \sem{[]}{[]}{1} \delta(c_1,\rho_1(z)) }
      (\textsc{Const-App}) \\
      \frac{
          \rho_2 \vdash y \sem{[]}{[]}{1} c_2 \quad
          \rho_2 \vdash z \sem{[]}{[]}{1} \rho_2(z) \quad
          |c_2| > 0
      }
      { \rho_2 \vdash y \s z \sem{[]}{[]}{1} \delta(c_2,\rho_2(z)) }
      (\textsc{Const-App})
    \end{gathered}
  \]
  We first establish \condalignedall{$'$}.
  \begin{description}
    \item[\condenvtwo{$'$}]
      By repeating the corresponding argument for \condenv{$'$} in Lemma~\ref{lemma:unaligned} for both $\rho_1'(x)$ and $\rho_2'(x)$.
    \item[\condenvrec{$'$}]
      Assume $\ttt{stoch} \not\in S_x$.
      For (\tsc{Const-App}), $\rho_1'(x) \in C$ and $\rho_2'(x) \in C$, and the result follows immediately.
      Therefore, assume that both derivations are (\tsc{App}).
      By Lemma~\ref{lemma:cfa}, $\ttt{stoch} \in S_y \Rightarrow \ttt{stoch} \in S_x$, and consequently $\ttt{stoch} \not\in S_y$.
      By \condenvneq{}, this leads to $\rho_1(y) = \langle\lambda y_1. \term_{y_1},\rho_{y_1}\rangle \eqv \rho_2(y) = \langle\lambda y_2. \term_{y_2},\rho_{y_2}\rangle$.
      That is, $\lambda y_1. \term_{y_1} = \lambda y_2. \term_{y_2} = \lambda y'. \term_{y'}$.
      By \condenvrec{}, $\rho_{y_1}$ and $\rho_{y_2}$ fulfill \condalignedall{}.
      Let $\rho_1'' = \rho_{y_1}, y' \mapsto \rho_1(z)$ and $\rho_2'' = \rho_{y_2}, y' \mapsto \rho_2(z)$ and consider the derivations $\rho_1'' \vdash \term_{y'} \sem{l_{11}}{s_{11}}{w_{11}} \termv'_1$ and $\rho_2'' \vdash \term_{y'} \sem{l_{21}}{s_{21}}{w_{21}} \termv'_2$.
      It is straightforward to check that $\rho_1''$ and $\rho_2''$ fulfill \condenv{}, and we apply the induction hypothesis and get the results \resalignedall{$''$}.
      Now, by Lemma~\ref{lemma:cfa}, $S_{\tsc{name}(\term_{y'})} \subseteq S_x$.
      Combined with \resvrec{$''$}, the result follows.
    \item[\condenvneq{$'$}]
      Assume $\rho_1'(x) \not\eqv \rho_2'(x)$, and consider first the case where $\rho_1(y) \not\eqv \rho_2(y)$.
      Then, by \condenvneq{}, $\ttt{stoch} \in S_y$ and by $\ttt{stoch} \in S_y \Rightarrow \ttt{stoch} \in S_x$ from Lemma~\ref{lemma:cfa}, $\ttt{stoch} \in S_x$ and we are done.
      Therefore, assume $\rho_1(y) \eqv \rho_2(y)$.
      Consequently, both derivations are either (\tsc{Const-App}) or (\tsc{App}).
      If both derivations are (\tsc{Const-App}), $c_1 \eqv c_2 \eqv c$ and $\rho_1'(x) \not\eqv \rho_2'(x)$ implies $\rho_1(z) \not\eqv \rho_2(z)$.
      By \condenvneq{}, this implies $\ttt{stoch} \in S_z$.
      Lemma~\ref{lemma:cfa} gives $\ttt{const} \s \_ \in S_y \Rightarrow (\ttt{stoch} \in S_z \Rightarrow \ttt{stoch} \in S_x)$.
      Clearly, $\ttt{const} \s |c| \in S_y$ by \condenvtwo{} and $|c| > 1$.
      It follows that $\ttt{stoch} \in S_x$.
      If both derivations are instead (\tsc{App}), we repeat the argument for \condenvrec{$'$} and get \resalignedall{$''$}.
      Furthermore, we must have $\termv_1' = \rho_1'(x) \not\eqv \rho_2'(x) = \termv_2'$.
      By Lemma~\ref{lemma:cfa}, $S_{\tsc{name}(\term_{y'})} \subseteq S_x$.
      The result now follows from \resvneq{$''$}.
  \end{description}
  We now apply the induction hypothesis and get \resalignedall{$'$}.
  \begin{description}
    \item[\resa{}]
      First, by \resa{$'$}, we have $l_{12}|_\restt = l_{22}|_\restt$.
      We now show $l_{11}|_\restt = l_{21}|_\restt$.

      Assume that $\ttt{stoch} \in S_y$.
      Then, in all cases we have $l_{11}|_\restt = l_{21}|_\restt = []$ and the result follows.
      To see this, note first that for the (\tsc{Const-App}) derivations, the result holds immediately.
      Therefore, assume both derivations are (\tsc{App}).
      Now, by Lemma~\ref{lemma:cfa}, we have $\ttt{stoch} \in S_y \Rightarrow (\forall y' \s \lambda y'. \_ \in S_y \Rightarrow \unaligned_y')$.
      In other words, $\unaligned_{y_1}$ and $\unaligned_{y_2}$.
      Again by Lemma~\ref{lemma:cfa}, $\unaligned_{n_1}$ for all $n_1 \in \tsc{names}(t_{y_1})$ and $\unaligned_{n_2}$ for all $n_2 \in \tsc{names}(t_{y_2})$.
      Let $\rho_1'' = \rho_{y_1}, y_1 \mapsto \rho_1(z)$ and $\rho_2'' = \rho_{y_2}, y_2 \mapsto \rho_2(z)$ and consider the derivations $\rho_1'' \vdash \term_{y_1} \sem{l_{11}}{s_{11}}{w_{11}} \termv'_1$ and $\rho_2'' \vdash \term_{y_2} \sem{l_{21}}{s_{21}}{w_{21}} \termv'_2$.
      It is straightforward to check that $\rho_1''$ and $\rho_2''$ fulfill \condenv{}
      and double applications of Lemma~\ref{lemma:unaligned} give the required result $l_{11}|_\restt = l_{21}|_\restt = []$.

      Now, assume that $\ttt{stoch} \not\in S_y$.
      Clearly, both derivations are either (\tsc{Const-App}) or (\tsc{App}).
      The (\tsc{Const-App}) case is trivial, because $l_{11}|_\restt = l_{21}|_\restt = []$.
      Therefore, assume both derivations are (\tsc{App}).
      By repeating the reasoning in \condenvrec{$'$}, we get \resalignedall{$''$} by the induction hypothesis for the derivations $\rho_1'' \vdash \term_{y'} \sem{l_{11}}{s_{11}}{w_{11}} \termv'_1$ and $\rho_2'' \vdash \term_{y'} \sem{l_{21}}{s_{21}}{w_{21}} \termv'_2$.
      In other words, $l_{11}|_\restt = l_{21}|_\restt$ by \resa{$''$} and we are done.
  \end{description}
  \noindent\tbf{Subcase} $\term_1 = \ttt{if } y \ttt{ then } \term_t \ttt{ else } \term_e$\\
  The possible derivations are
  \[
    \begin{gathered}
      \frac{ \rho_1 \vdash y \sem{[]}{[]}{1} \true \quad \rho_1 \vdash \term_t \sem{l_{11}}{s_{11}}{w_{11}} \termv_{t_1} }
      {\rho_1 \vdash \ttt{if } y \ttt{ then } \term_t \ttt{ else } \term_e \sem{l_{11}}{s_{11}}{w_{11}} \termv_{t_1}}
      (\textsc{If-True}) \\
      \frac{ \rho_2 \vdash y \sem{[]}{[]}{1} \true \quad \rho_2 \vdash \term_t \sem{l_{21}}{s_{21}}{w_{21}} \termv_{t_2} }
      {\rho_2 \vdash \ttt{if } y \ttt{ then } \term_t \ttt{ else } \term_e \sem{l_{21}}{s_{21}}{w_{21}} \termv_{t_2}}
      (\textsc{If-True}) \\
      \frac{ \rho_1 \vdash y \sem{[]}{[]}{1} \false \quad \rho_1 \vdash \term_e \sem{l_{11}}{s_{11}}{w_{11}} \termv_{e_1} }
      {\rho_1 \vdash \ttt{if } y \ttt{ then } \term_t \ttt{ else } \term_e \sem{l_{11}}{s_{11}}{w_{11}} \termv_{e_1}}
      (\textsc{If-False}) \\
      \frac{ \rho_2 \vdash y \sem{[]}{[]}{1} \false \quad \rho_2 \vdash \term_e \sem{l_{21}}{s_{21}}{w_{21}} \termv_{e_2} }
      {\rho_2 \vdash \ttt{if } y \ttt{ then } \term_t \ttt{ else } \term_e \sem{l_{21}}{s_{21}}{w_{21}} \termv_{e_2}}
      (\textsc{If-False})
    \end{gathered}
  \]
  We first establish \condalignedall{$'$}.
  \begin{description}
    \item[\condenvtwo{$'$}]
      Holds in all four cases by repeating the corresponding argument for \condenv{$'$} in Lemma~\ref{lemma:unaligned}.
    \item[\condenvrec{$'$}]
      Assume $\ttt{stoch} \not\in S_x$.
      By Lemma~\ref{lemma:cfa}, clearly $\ttt{stoch} \not\in S_y$ and both derivations are either (\tsc{If-True}) or (\tsc{If-False}).
      Without loss of generality, assume both derivations are (\tsc{If-True}).
      The induction hypothesis directly applies to $\rho_1 \vdash \term_t \sem{l_{11}}{s_{11}}{w_{11}} \termv_{t_1}$ and $\rho_2 \vdash \term_t \sem{l_{21}}{s_{21}}{w_{21}} \termv_{t_2}$, and we get the result \resalignedall{$_t$}.
      By Lemma~\ref{lemma:cfa}, $\tsc{name}(\term_t) \subseteq S_x$.
      The result now follows from \resvrec{$_t$}.
    \item[\condenvneq{$'$}]
      Assume first that $\ttt{stoch} \in S_y$.
      Then $\ttt{stoch} \in S_x$ by Lemma~\ref{lemma:cfa}, and the result is immediate.
      Therefore, assume $\ttt{stoch} \not\in S_y$.
      Again, both derivations are either (\tsc{If-True}) or (\tsc{If-False}) and we assume, without loss of generality, that both are (\tsc{If-True}).
      The induction hypothesis directly applies to $\rho_1 \vdash \term_t \sem{l_{11}}{s_{11}}{w_{11}} \termv_{t_1}$ and $\rho_2 \vdash \term_t \sem{l_{21}}{s_{21}}{w_{21}} \termv_{t_2}$, and we get the result \resalignedall{$_t$}.
      By Lemma~\ref{lemma:cfa}, $\tsc{name}(\term_t) \subseteq S_x$.
      The result now follows from \resvneq{$_t$}.
  \end{description}
  We now apply the induction hypothesis and get \resalignedall{$'$}.
  \begin{description}
    \item[\resa{}]
      First, by \resa{$'$}, we have $l_{12}|_\restt = l_{22}|_\restt$.
      We now show $l_{11}|_\restt = l_{21}|_\restt$.

      If $\ttt{stoch} \in S_y$, then by Lemma~\ref{lemma:cfa}, $\unaligned_{n_t}$ for all $n_t \in \tsc{names}(\term_{t})$ and $\unaligned_{n_e}$ for all $n_e \in \tsc{names}(\term_e)$.
      By repeating Lemma~\ref{lemma:unaligned} twice, we get $l_{11}|_\restt = l_{21}|_\restt = []$ and the result follows.

      Assume $\ttt{stoch} \not\in S_y$.
      Again, both derivations are either (\tsc{If-True}) or (\tsc{If-False}) and we assume, without loss of generality, that both are (\tsc{If-True}).
      The induction hypothesis directly applies to $\rho_1 \vdash \term_t \sem{l_{11}}{s_{11}}{w_{11}} \termv_{t_1}$ and $\rho_2 \vdash \term_t \sem{l_{21}}{s_{21}}{w_{21}} \termv_{t_2}$, and we get the result \resalignedall{$_t$}.
      By \resa{$_t$}, $l_{11}|_\restt = l_{21}|_\restt$.
  \end{description}
  \noindent\tbf{Subcase} $\term_1 = \ttt{assume } y$\\
  The derivations are
  \[
    \begin{gathered}
      \frac{\rho_1 \vdash y \sem{[]}{[]}{1} d_1 \quad w_1 = f_{d_1}(c_1) }
      {\rho_1 \vdash \ttt{assume } y \sem{[]}{[c_1]}{w_1} c_1}
      (\textsc{Assume})
      \\
      \frac{\rho_2 \vdash y \sem{[]}{[]}{1} d_2 \quad w_2 = f_{d_2}(c_2) }
      {\rho_2 \vdash \ttt{assume } y \sem{[]}{[c_2]}{w_2} c_2}
      (\textsc{Assume})
    \end{gathered}
  \]
  We first establish \condalignedall{$'$}.
  \begin{description}
    \item[\condenvtwo{$'$}]
      By repeating the corresponding argument for \condenv{$'$} in Lemma~\ref{lemma:unaligned} for both $\rho_1'(x)$ and $\rho_2'(x)$.
    \item[\condenvrec{$'$}]
      Immediate as $\rho_1'(x) = c_1$ and $\rho_2'(x) = c_2$.
    \item[\condenvneq{$'$}]
      By Lemma~\ref{lemma:cfa}, $\ttt{stoch} \in S_x$.
  \end{description}
  We now apply the induction hypothesis and get \resalignedall{$'$}.
  \begin{description}
    \item[\resa{}] The result follows from $l_{11} = l_{21} = []$ and \resa{$'$}.
  \end{description}
  \noindent\tbf{Subcase} $\term_1 = \ttt{weight } y$\\
  The derivations are
  \[
    \begin{gathered}
      \frac{\rho_1 \vdash y \sem{[]}{[]}{1} w_1}
      {\rho_1 \vdash \ttt{weight } y \sem{[]}{[]}{w_1} ()}
      (\textsc{Weight})
      \\
      \frac{\rho_2 \vdash y \sem{[]}{[]}{1} w_2}
      {\rho_2 \vdash \ttt{weight } y \sem{[]}{[]}{w_2} ()}
      (\textsc{Weight})
    \end{gathered}
  \]
  We first establish \condalignedall{$'$}.
  \begin{description}
    \item[\condenvtwo{$'$}]
      By repeating the corresponding argument for \condenv{$'$} in Lemma~\ref{lemma:unaligned} for both $\rho_1'(x)$ and $\rho_2'(x)$.
    \item[\condenvrec{$'$}]
      Immediate as $\rho_1'(x) = \rho_2'(x) = ()$.
    \item[\condenvneq{$'$}]
      Immediate as $\rho_1'(x) \eqv \rho_2'(x)$.
  \end{description}
  We now apply the induction hypothesis and get \resalignedall{$'$}.
  \begin{description}
    \item[\resa{}] The result follows from $l_{11} = l_{21} = []$ and \resa{$'$}.
  \end{description}
  \qed
\end{proof}

\section{Unaligned SMC}\label{sec:smcunaligned}
Algorithm~\ref{alg:smcunaligned} presents the unaligned SMC algorithm.
It is in many ways similar to Algorithm~\ref{alg:smc}.

\begin{algorithm}[tb]
  \caption{%
    Unaligned SMC. The input is a program $\term \in \Termanf$ and the number of execution instances $n$.
  }\label{alg:smcunaligned}
  \vspace{-3mm}
  \begin{enumerate}
    \item Initiate $n$ execution instances $\{e_i \mid i \in \mathbb{N}, 1 \leq i \leq n\}$ of $\term$.
    \item\label{alg:smcunaligned:start}
      Execute all $e_i$ (for already terminated $e_i$, do nothing) and suspend execution upon reaching a weight (i.e., \ttt{let $x$ = weight $w$ in \term}) or when the execution terminates naturally.
      The result is a new set of execution instances $e_i'$ with weights $w_i'$ (from $w$, or $1$ if already terminated).
    \item\label{alg:smcunaligned:terminate}
      If all $e_i' = \termv_i'$ (i.e., all executions have terminated and returned a value), terminate inference and return the set of samples $\termv_i'$.
      The samples approximate the probability distribution encoded by $\term$.
    \item
      Resample the $e_i'$ according to their weights $w_i'$.
      The result is a new set of unweighted execution instances $e_i''$.
      Set $e_i \gets e_i''$.
      Go to \ref{alg:smcunaligned:start}.
  \end{enumerate}
  \vspace{-3mm}
\end{algorithm}%

\section{Lightweight MCMC}\label{sec:mcmcunaligned}
Algorithm~\ref{alg:mcmcunaligned} presents the lightweight MCMC algorithm.
The algorithm is in many ways similar to Algorithm~\ref{alg:mcmc}, but relies on databases represented with $D_i$ (random draws) and $p_i$ (probability densities/masses of the draws) to reuse random draws.
The \tsc{Run} function keeps track of the current stack trace $t$ at all times and uses it to index the databases.

\begin{algorithm}[tb]
  \caption{%
    Lightweight MCMC. The input is a program $\term \in \Termanf$, the number of steps $n$, and the global step probability $g > 0$.
  }\label{alg:mcmcunaligned}
  \vspace{-3mm}
  \begin{enumerate}
    \item Set $i \gets 0$. Call \textsc{Run}.
    \item\label{alg:mcmcunaligned:loop}
      Set $i \gets i + 1$.
      If $i = n$, terminate inference and return the samples $\{\termv_j \mid j \in \mathbb{N}, 0 \leq j < n\}$. They approximate the probability distribution encoded by \term.
    \item Uniformly draw a trace $t'$ from $\text{dom}(D_{i-1})$ at random. Set $\mi{global} \gets \true$ with probability $g$, and $\mi{global} \gets \false$ otherwise. Set $w'_{-1} \gets 1$, and $w' \gets 1$. Call \textsc{Run}.
    \item\label{alg:mcmcunaligned:metropolis}
      Compute the Metropolis--Hastings acceptance ratio
      \begin{equation}
        A = \min\left(1,\frac{w_i}{w_{i-1}}\frac{w'}{w'_{-1}}\frac{|\text{dom}(D_{i-1})|}{|\text{dom}(D_{i})|}\right).
      \end{equation}
    \item
      With probability $A$, accept $\termv_i$ and go to \ref{alg:mcmcunaligned:loop}.
      Otherwise, set $\termv_i \gets \termv_{i-1}$, $w_i \gets w_{i-1}$, $D_i \gets D_{i-1}$, and $p_i \gets p_{i-1}$. Go to \ref{alg:mcmcunaligned:loop}.
  \end{enumerate}
  \vspace{-3mm}
  \lstinline[style=alg]!function $\tsc{run}$() =! Let $t$ represent the current stack trace throughout execution. Run $\term$ and do the following:
  \vspace{-2mm}
  \begin{itemize}
    \item Record the total weight $w_i$ accumulated from calls to \ttt{weight}.
    \item Record the final value $\termv_i$.
    \item At terms \ttt{let $c$ = assume $d$ in \term}, do the following.
      \begin{enumerate}
        \item
          If $t = t'$, $\mi{global} = \true$, or if $t \not\in \text{dom}(D_{i-1})$, sample a value $x$ from $d$.
          Otherwise, reuse the sample $x = D_{i-1}(t)$ and set \mbox{$w'_{-1} \gets w'_{-1} \cdot p_{i-1}(t)$} and \mbox{$w' \gets w' \cdot f_d(c)$}.
        \item
          Set $D_i(t) \gets x$ and $p_{i}(t) \gets f_d(x)$.
        \item
          In the program, bind $c$ to the value $x$ and resume execution.
      \end{enumerate}
  \end{itemize}
  \vspace{-2mm}
\end{algorithm}

\section{Metropolis--Hastings Acceptance Ratio}\label{sec:mcmcaligncont}
This section derives the Metropolis--Hastings acceptance ratio used in Algorithm~\ref{alg:mcmc} and Algorithm~\ref{alg:mcmcunaligned}.
We assume basic familiarity with Bayesian statistics and the Metropolis--Hastings algorithm.

Bayes' theorem on probability density/mass functions is usually written as
\begin{equation}
  p(x|y) = \frac{p(y|x)p(x)}{p(y)}
\end{equation}
where $y$ is some fixed \emph{observed} random variable.
The standard Metropolis--Hastings ratio for a proposal distribution with probability density/mass $q(x'|x)$ is then
\begin{equation}
  A(x,x')
  =
  \min\left(1,
  \frac
  {p(x'|y)}
  {p(x|y)}
  \frac
  {q(x|x')}
  {q(x'|x)}
  \right)
  =
  \min\left(1,
  \frac
  {p(y|x')p(x')}
  {p(y|x)p(x)}
  \frac
  {q(x|x')}
  {q(x'|x)}
  \right).
\end{equation}
Assume a fixed program $\term \in T$ in the remainder of this section.
For such a program, Bayes' theorem takes a generalized form
\begin{equation}\label{eq:bayesppl}
  \hat{p}(s) = \frac{\mathcal L (s)p(s)}{Z}.
\end{equation}
Here, we have replaced $x$ with a trace $s$ (a sequence of random values during evaluation of a probabilistic program) and removed the dependence on $y$ entirely.
We use the notation $\hat{p}$ and $p$ to differentiate between the posterior and prior.
The likelihood function is denoted $\mathcal L$.
$Z$ is a normalizing constant that disappears in the Metropolis--Hastings ratio.

One can view \eqref{eq:bayesppl} in the context of the semantics in Fig.~\ref{fig:semantics}.
We (very informally) have $\hat{p}(s) = w$ up to normalization iff $\varnothing \vdash \term \sem{s}{l}{w} \termv$ for some $l$ and $\termv$.
$\mathcal L (s)$ is then the contribution to $w$ from (\tsc{Weight}), and $p(s)$ from (\tsc{Assume}).
The PPL version of the Metropolis--Hastings ratio is
\begin{equation}
  A(s,s')
  =
  \min\left(1,
  \frac
  {\hat{p}(s')}
  {\hat{p}(s)}
  \frac
  {q(s|s')}
  {q(s'|s)}
  \right)
  =
  \min\left(1,
  \frac
  {\mathcal L (s')p(s')}
  {\mathcal L (s)p(s)}
  \frac
  {q(s|s')}
  {q(s'|s)}
  \right).
\end{equation}
The most trivial proposal, amounting to not reusing any draws, is
\begin{equation}
  q(s'|s) = p(s')
\end{equation}
This directly gives the ratio
\begin{equation}
  A(s,s')
  =
  \min\left(1,
  \frac
  {\mathcal L (s')}
  {\mathcal L (s)}
  \right).
\end{equation}
To derive the ratio for aligned lightweight MCMC and lightweight MCMC, we need to first capture the proposal $q$ used in Algorithm~\ref{alg:mcmc} and Algorithm~\ref{alg:mcmcunaligned}.
We capture the reuse mechanisms (alignment and the stack trace database) in both algorithms through functions $D_1 : S \times S \to \mathcal P(\mathbb N)$ and $D_2 : S \times S \to \mathcal P(\mathbb N)$ such that $|D_1(s,s')| = |D_2(s,s')|$, $D_1(s,s') = D_2(s',s)$, and $D_2(s,s') = D_1(s',s)$.
Intuitively, $D_1(s,s')$ gives the indices in $s$ that match the indices $D_2(s,s')$ in $s'$.

We now define the proposal $q$ as
\begin{equation}
  q(s',i|s)
  =
  [s'|_{A'} = s|_{A}] p|_{{A'}^C}(s')p_i(i|s)
\end{equation}
and make the following definitions.
\begin{itemize}
  \item $A' = D_2(s,s') \setminus f(i,s')$ and $A = D_1(s,s') \setminus f(i,s)$.
  \item The function $f(i,s)$ transforms the index $i$ in the context of $s$ (explained further below).
  \item The function $p_i(i|s)$ is the density for selecting an $i$ given $s$.
  \item The trace $s|_A$ is the restriction of $s$ to $A$ (cf. Definition~\ref{def:seqrestr}).
  \item
    $[\cdots]$ is the Iverson bracket (i.e., evaluates to zero if the predicate $\cdots$ is false and to one if the predicate is true).
  \item
    We denote the contribution to $p(s)$ from the indices $A$ in $s$ with $p|_A(s)$.
    Importantly, $p|_A(s) \cdot p|_{A^C}(s) = p(s)$.
\end{itemize}
Note that $q$ also proposes an auxiliary variable $i$---the trace index that we choose to redraw in the proposal.
Due to the auxiliary variable $i$, the acceptance ratio is now a function of three arguments.
\begin{equation}\label{eq:mhratio}
  \begin{aligned}
    A(s,s',i)
    &=
    \min\left(1,
    \frac
    {\mathcal L (s')p(s')}
    {\mathcal L (s)p(s)}
    \frac
    {q(s,i|s')}
    {q(s',i|s)}
    \right)
    \\&=
    \min\left(1,
    \frac
    {\mathcal L (s')p(s')}
    {\mathcal L (s)p(s)}
    \frac
    {[s|_{A} = s'|_{A'}]}
    {[s'|_{A'} = s|_{A}]}
    \frac
    { p|_{{A}^C}(s)}
    { p|_{{A'}^C}(s')}
    \frac
    {p_i(i|s')}
    {p_i(i|s)}
    \right)
    \\&=
    \min\left(1,
    \frac
    {\mathcal L (s')}
    {\mathcal L (s)}
    \frac
    {p(s')}
    {p|_{{A'}^C}(s')}
    \frac
    {p|_{{A}^C}(s)}
    {p(s)}
    \frac
    {p_i(i|s')}
    {p_i(i|s)}
    \right)
    \\&=
    \min\left(1,
    \frac
    {\mathcal L (s')}
    {\mathcal L (s)}
    \frac
    {p|_{A'}(s')}
    {p|_{A}(s)}
    \frac
    {p_i(i|s')}
    {p_i(i|s)}
    \right)
  \end{aligned}
\end{equation}
This acceptance ratio is equivalent to the ratio derived by van de Meent et al.~\cite[Equation 4.21]{vandemeent2018introduction}.

We first view \eqref{eq:mhratio} in the context of aligned lightweight MCMC in Algorithm~\ref{alg:mcmc}.
Here, $f(i,s)$ returns the $i$-th \emph{aligned} index in $s$ (\emph{not} index $i$ in $s$).
In aligned lightweight MCMC, we only select what to redraw among the aligned draws.
As we know, the number of aligned draws is fixed across all possible executions.
$p_i(i|s)$ is thus a constant, and \eqref{eq:mhratio} reduces to
\begin{equation}
  A(s,s',i) =
  \min\left(1,
  \frac
  {\mathcal L (s')}
  {\mathcal L (s)}
  \frac
  {p|_{A'}(s')}
  {p|_{A}(s)}
  \right)
\end{equation}
This is the ratio computed in step~\ref{alg:mcmc:metropolis} of Algorithm~\ref{alg:mcmc}.

Next, we consider lightweight MCMC in Algorithm~\ref{alg:mcmcunaligned}.
Here, we simply choose $f(i,s) = i$ (the identity function), and select an element to redraw uniformly over the previous trace $s$.
Thus, $p_i(i|s) = 1/|s|$ and \eqref{eq:mhratio} reduces to
\begin{equation}
  A(s,s',i) =
  \min\left(1,
  \frac
  {\mathcal L (s')}
  {\mathcal L (s)}
  \frac
  {p|_{A'}(s')}
  {p|_{A}(s)}
  \frac
  {|s|}
  {|s'|}
  \right)
\end{equation}
This is the ratio computed in step~\ref{alg:mcmcunaligned:metropolis} of Algorithm~\ref{alg:mcmcunaligned}.

\fi

\end{document}